\documentclass[preprint,12pt]{elsarticle}

\usepackage{amssymb}
\usepackage{amsmath}
\usepackage{footmisc}
\usepackage{verbatim}
\usepackage{amsthm}   
\usepackage{makecell} 
\usepackage[margin=2.5cm]{geometry}

\newtheorem{theorem}{Theorem}

\newtheorem{lemma}[theorem]{Lemma}

 \usepackage{lineno}
\usepackage{hyperref}
\usepackage{graphics}
\newcount\Comments  
\usepackage{color}
\definecolor{darkgreen}{rgb}{0,0.5,0}
\definecolor{purple}{rgb}{1,0,1}
\newcommand{\kibitz}[2]{\ifnum\Comments=1\textcolor{#1}{#2}\fi}
\newcommand{\robab}[1]{\kibitz{blue}       {[robab: #1]}}
\newcommand{\xingyu}[1]  {\kibitz{darkgreen}   {[Xingyu: #1]}}
\newcommand{\rv}[1]{\kibitz{black}       {#1}}

\usepackage{float}
\usepackage{xcolor}
\usepackage{diagbox}  
\usepackage{array}
\usepackage{booktabs}
\usepackage{makecell}
\usepackage{multirow}
\usepackage{amsthm}
\usepackage{diagbox}  
\usepackage{subcaption} 

\newtheorem{definition}{Definition}
\newtheorem{remark}{Remark}
\usepackage[most]{tcolorbox}

\journal{}

\begin{document}

\begin{frontmatter}



\title{A Hierarchical Imprecise Probability Approach to Reliability Assessment of Large Language Models}

\author[label1]{Robab~Aghazadeh-Chakherlou}
\author[label4,label5]{Qing~Guo}
\author[label1]{Siddartha~Khastgir}
\author[label2]{Peter~Popov}
\author[label3]{Xiaoge~Zhang}
\author[label1]{Xingyu~Zhao\corref{cor1}}
\cortext[cor1]{Corresponding author: Xingyu Zhao, \url{xingyu.zhao@warwick.ac.uk}}
\affiliation[label1]{organization={WMG, University of Warwick},
            city={Coventry},
            country={United Kingdom}}
\affiliation[label4]{organization={Center for Frontier AI Research, A*STAR},
            city={Singapore},
            country={Singapore}}
\affiliation[label5]{organization={School of Computing, National University of Singapore},
            city={Singapore},
            country={Singapore}}
\affiliation[label2]{organization={Center for Software Reliability, City St George's, University of London},
            city={London},
            country={United Kingdom}}
\affiliation[label3]{organization={Department of Industrial and Systems Engineering, The Hong Kong Polytechnic University},
            city={Kowloon},
            country={Hong Kong}}




\begin{abstract}

Large Language Models (LLMs) are increasingly deployed across diverse domains, raising the need for rigorous reliability assessment methods. Existing benchmark-based evaluations primarily offer descriptive statistics of model accuracy over datasets, providing limited insight into the probabilistic behavior of LLMs under real operational conditions. This paper introduces HIP-LLM, a \textbf{H}ierarchical \textbf{I}mprecise \textbf{P}robability framework for modeling and inferring \textbf{LLM} reliability. Building upon the foundations of software reliability engineering, HIP-LLM defines LLM reliability as the probability of failure-free operation over a specified number of future tasks under a given Operational Profile (OP). HIP-LLM represents dependencies across (sub-)domains hierarchically, enabling multi-level inference from subdomain to system-level reliability. HIP-LLM embeds imprecise priors to capture epistemic uncertainty and incorporates OPs to reflect usage contexts. It derives posterior reliability envelopes that quantify uncertainty across priors and data. Experiments on multiple benchmark datasets demonstrate that HIP-LLM offers a more nuanced and standardized reliability characterization than existing benchmark and state-of-the-art approaches. A publicly accessible repository of HIP-LLM is provided.

\end{abstract}


\end{frontmatter}



\section{Introduction}
\label{sec_introduction}

Large language models (LLMs) are increasingly applied across a wide range of tasks, from general-purpose reasoning to highly specialized applications. For instance, recent studies~\cite{pang2024parinfogpt,xiao2025krail,zheng2024empirical,gohar2024codefeater,murugesan2024automating,sultan2024ai} employ LLMs as tools to support reliability analysis and safety assurance for safety-critical systems. This expanding scope of application underscores the urgent need for comprehensive and rigorous evaluation methods to assess the \textit{reliability} of the LLMs themselves, given the high-stakes nature of those applications. While numerous studies have proposed domain-specific performance evaluations in areas such as coding~\cite{yang_intercode_2023,lai_ds_2023,yu_humaneval_2024}, medicine~\cite{croxford_current_2025}, and system safety~\cite{charalampidou2024hazard,kaya2025large,qi2025safety}, others have focused on broader properties of LLMs, including safety (avoiding harmful content)~\cite{zhang_2024_safetybench,liu2024mm,mou2024sg}, robustness (resisting adversarial attacks such as jailbreaks)~\cite{deepseek_NIST,yu2025benchmarking}, fairness~\cite{ICLR2025_38e49155}, and privacy~\cite{privacy_bench,tamkin_clio_2024}. However, there remains a notable lack of research dedicated to reliability, which motivates this work.

While software \textit{reliability} is highly interrelated with the aforementioned properties (a formal definition and its distinction to other properties will be discussed in later sections), it is a unique property that emphasizes the failure probabilities, specified operational time and conditions, and statistically valid uncertainty quantification. As a ``user-centric'' property \cite{littlewood_software_2000,strigini_guidelines_1997,dong2023reliability}, the delivered reliability depends on the behavior of end-users and how the software is expected to be used in practice. This notion is made measurable through the \textit{Operational Profile} (OP) \cite{musa1993operational,musa_adjusting_1994}, which specifies a probability distribution over the types of demands that users actually place on the software. The recent report from OpenAI \cite{chatterji2025people} provide the most comprehensive data on how user-groups use ChatGPT, which aligns with the idea of OPs.


Despite software reliability assessment has been studied for decades \cite{lyu_handbook_1996, littlewood_validation_1993}, state-of-the-art LLM evaluation largely relies on benchmarks \cite{chang_surveyLLM_2024}. This creates several research gaps: 
\begin{itemize}
    \item[Gap-1](static benchmarks vs dynamic operational use): While we acknowledge the usefulness of benchmarks, they are primarily designed for purposes such as comparing and ranking LLMs on given tasks under fixed datasets, neglecting the variability and dynamic nature of user's behaviors in real-world applications. Thus, benchmarks cannot serve the purpose of assessing reliability, especially when the OP of the users is different from the data distribution (implicitly) represented by the benchmark dataset.
    \item[Gap-2](independent benchmarks vs hierarchical dependencies): While LLMs are generally designed for broad use (covering various task types such as legal reasoning and coding), these tasks naturally exist at different levels of abstraction. It is often useful to organize them hierarchically, dividing the input space into domains and subdomains \cite{luettgau_hibayes_2025}. High-level domains\footnote{While the precise definition of what constitutes a domain or subdomain is an application-oriented question that may vary case by case reflecting the assessor's domain knowledge, in this reliability modeling work we do not prescribe specific taxonomies. Instead, we introduce a general hierarchical modeling structure that can accommodate different dependency assumptions (cf.~Section~\ref{sec_discussion} for more discussion).\label{footnote_domain}} (e.g., law vs. coding) can be treated as largely independent, whereas subdomains within a domain (e.g., coding in Python vs. coding in C++) are more likely to exhibit dependencies. In such cases, performance in one subdomain may influence users’ confidence in another, whereas cross-domain effects are minimal. Reliability assessment should therefore reflect these dependency relationships and support evaluation at multiple abstraction levels, depending on the assessor’s role. For instance, a LLM vendor may be interested in system-wide reliability across all domains, a programmer may focus on reliability in coding tasks, and a Python developer may be specifically concerned with the reliability of Python coding. Existing benchmarks, however, typically evaluate LLM capabilities either in isolation or by simply averaging performance across tasks, that obscure these dependencies relationships \cite{luettgau_hibayes_2025}.
    \item[Gap-3](descriptive statistics vs statistical inference): While descriptive statistics summarize the observed dataset by reporting point-estimate benchmark scores, they do not formally model \textit{how the data were generated}. Even when variance is reported as a way of ``quantifying uncertainty'', it quantifies variability \textit{within the sample itself}. In contrast, statistical inference treats the observed results as samples from \textit{an underlying stochastic process or population} and uses probabilistic modeling to estimate the parameters of that process \cite{miller_adding_2024}. This inferential perspective enables generalization beyond the fixed benchmark \textit{dataset} and supports principled uncertainty quantification about reliability which concerns a underlying \textit{population} of future inputs \cite{luettgau_hibayes_2025}.
\end{itemize}


Indeed, acknowledging these research gaps, recent studies \cite{miller_adding_2024,luettgau_hibayes_2025} attempt to bridge them by applying both ``frequentist'' and Bayesian statistical inference methods.
However, despite representing important first steps, they do not fully resolve the aforementioned gaps (e.g., not explicitly incorporate the OP) and introduce new gaps:
\begin{itemize}
  \item [Gap-4](failure probability vs future failure-free runs) While both \cite{luettgau_hibayes_2025,miller_adding_2024} examined the metric of failure probability\footnote{We later show that this is a special case of our reliability model.}, they did not extend it to the more practical and widely adopted definition of reliability that predicts the \textit{probability of failure-free runs} over a specified future operational time \cite{strigini_software_2013,bishop2022bootstrapping,zhao2020safety}.
    \item[Gap-5](non-informative vs informative priors): 
    The work \cite{miller_adding_2024} is a frequentist approach that cannot embed prior knowledge explicitly like Bayesian methods. The work \cite{luettgau_hibayes_2025} employs non-informative priors in its Bayesian models, without facilitating the embedding of any informative prior knowledge that end-users may have. 
\end{itemize}


To bridge the five research gaps collectively, we propose HIP-LLM, a hierarchical imprecise probability approach to assess LLM reliability. HIP-LLM devises a more rigorous and versatile assessment with the following key features. First, similar to \cite{luettgau_hibayes_2025}, it structures the model into independent high-level domains, with each domain further divided into dependent subdomains. In this way, reliability can be evaluated at multiple levels of detail while preserving statistical dependencies, so that information from one subdomain contributes to the inference about other subdomains within the same domain. We then incorporate OPs as operational weights at each level of the hierarchy (subdomain $\rightarrow$ domain $\rightarrow$ general-purpose LLM) to reflect the dynamic operational use of subdomains and domains in practice. 
\rv{We assume tasks constitute independent and identically distributed (i.i.d.) trials drawn from the specified OP. This assumption is appropriate for reset or single-task scenarios, in which each task is executed in a fresh session without shared context or memory---settings commonly encountered in offline evaluation like benchmarking\footnote{In contrast, \textit{long-context} or \textit{agentic workflows}, where task outcomes are sequentially dependent through memory or tool use, violate the i.i.d. assumption and require fundamentally different modeling approaches to our HIP-LLM. We provide more discussions later in Remark~\ref{rm_iid_vs_context} and Section~\ref{sec_discussion}.} Accordingly, HIP-LLM is intended for LLM usage scenarios in which each task is executed in a new chat session without retaining memory from previous sessions, thereby aligning with the assumption of a ``reset'' LLM.}

In contrast to \cite{luettgau_hibayes_2025}, which relies on non-informative priors, HIP-LLM enables the incorporation of informative prior knowledge. Since selecting a single prior in Bayesian inference is often controversial due to its encoding of epistemic uncertainty, we adopt \textit{Imprecise Probability} \cite{augustin2014introduction,troffaes2007decision} that represents uncertainty about the prior itself without committing to a single distribution. Consequently, reliability metrics (defined in terms of both failure probability and the probability of failure-free runs in future operations) are expressed through \textit{posterior distribution envelopes} at different levels of the hierarchy, reflecting both uncertainties from the data and the prior knowledge.


In summary the contributions of this work are as follows:
\begin{itemize}
\item We formally define the reliability assessment problem for LLMs in accordance with established software reliability standards, while delineating its distinction from related properties.
\item We propose HIP-LLM, a hierarchical imprecise probability model that explicitly addresses key research gaps in the current state-of-the-art, by modeling OPs, hierarchical dependencies, statistical principled uncertainty quantification on reliability (both failure probability and future failure-free runs over specified time), and imprecise prior knowledge.
\item We release a public repository containing all experimental data, source code, and models at \url{https://github.com/aghazadehchakherlou-web/llm-imprecise-bayes}.
\end{itemize}



The remainder of the paper is organized as follows. Section \ref{sec_related_work} introduces preliminaries and reviews related works on LLM assessment.
Section~\ref{sec_dependenSub_independenDom} provides the modeling details of HIP-LLM. 
Section~\ref{sec_numerical_example} evaluates and illustrates HIP-LLM via experiments on datasets. 
Section~\ref{sec_discussion} discusses the limitations and assumptions, and finally Section~\ref{sec_conclusion} concludes the paper.


\section{Preliminaries and Related Works}
\label{sec_related_work}

\subsection{Software Reliability}
\label{sec_related_work_OP}

According to American National Standards Institute (ANSI), software reliability is defined as \cite{ansi_1991}: 
\begin{definition}[Software reliability]
\label{def_sw_reliability}
    The probability of failure-free software operation for a specified period of time in a specified environment.
\end{definition}
This probabilistic definition is also largely adopted by software reliability engineering literature \cite{lyu_handbook_1996}, which makes reliability amenable to statistical modeling and permits risk-aware aggregation across heterogeneous operational uses. The specific probabilistic metric used to model reliability is domain-dependent and determined by the operational nature of the software. For example, continuous-time systems are continuously operated in the active control of a process, whereas on-demand systems are only invoked upon receipt of discrete demands. In the latter case, such as nuclear power protection systems, the probability of failure on demand (\textit{pfd}) has been adopted in standards \cite{strigini_guidelines_1997,iec_61508_2010,atwood2003handbook} and extensively studied \cite{littlewood_reasoning_2012,strigini_software_2013,zhao_modeling_2017}.

Although related, reliability is distinct from safety, security, and accuracy. Reliability focuses on the \textit{probability} of failures (relative to a specification), a stochastic and usage-weighted concept by OPs \cite{zhao_modeling_2017}. Safety concerns \textit{critical failures} with \textit{catastrophic consequences}, i.e., whether a failure can lead to unacceptable \textit{harm}. For instance, a system can be reliable but unsafe (if frequent but non-harmful failures are tolerated while rare failures produce catastrophic harm), or safe but unreliable (if it fails frequently but never produce catastrophic harm by fail-safe design)\footnote{We note that for systems in which all failures are catastrophic (e.g., safety-critical systems), safety and reliability assessments do not require different statistical reasoning, such as when measuring the probability that a system will operate safely (i.e., without critical failures) over a given mission.}. Security is concerned with resisting \textit{malicious threats}: it normally assumes an explicit threat model and seeks to prevent or mitigate intentional compromises). While security fundamentally requires the analysis of a malicious threat actor, reliability and safety can be studied within benign operational environments where failure originates from sources such as design flaws and implementation bugs. Within the AI/ML community, accuracy is typically a narrower performance metric, defined as a descriptive statistic (proportion of correct predictions) over a fixed and given dataset. However, such a dataset does not necessarily represent the unknown ground-truth population of inputs that the model will face upon deployment, and accuracy does not quantify the uncertainties of failure-free operations over future time like reliability.

A central insight from software reliability engineering is that reliability is \textit{user-context dependent}: the same software can exhibit very different delivered reliability under different patterns of use. The OP formalizes this dependency, which is defined as \cite{strigini_guidelines_1997,musa1993operational,lyu_handbook_1996}:
\begin{definition}[Operational Profile (OP)]
    An OP is a probability distribution of inputs to a software system, representing the relative frequencies with which different input sequences expected occur during actual operation.
\end{definition} 
There are some caveats regarding the term ``inputs'' in the definition of OP, particularly when considering factors such as software memory and dependencies\footnote{In practice, because software systems may possess memory, their success/failure often depends on entire sequences of inputs rather than isolated ones. Accordingly, ``inputs'' may denote a complete stimulus, which may encompass the full sequence of interactions within a task, or even the cumulative sequence of inputs since the system was last reinitialized.}. We refer readers to \cite{strigini_guidelines_1997} for a detailed discussion. For LLMs, OPs serve an analogous role \cite{chatterji2025people}: they specify the distribution of task types (e.g., factual question answering, coding, summarization) and their relative importance in deployment. Incorporating OPs into LLM reliability analysis ensures that evaluation metrics capture not just aggregate benchmark scores, but also how well a model performs across the realistic mix of tasks it is expected to face.

\subsection{Benchmark-Based Evaluation of LLMs via Descriptive Statistics}

Most evaluation methods for LLMs report only point estimates such as accuracy scores. Question-answering benchmarks such as MMLU \cite{hendrycks_measuring_2020} and RACE \cite{lai_race_2017} measure how often a model picks the right answer in multiple-choice questions. Code benchmarks like DS-1000 \cite{lai_ds_2023}, InterCode \cite{yang_intercode_2023}, and HumanEval Pro/MBPP Pro \cite{yu_humaneval_2024} test whether generated programs run correctly. 
Assistant and autonomy benchmarks such as GAIA \cite{mialon_gaia_2023} and H-CAST \cite{rein_hcast_2025} look at overall task completion---whether the system can finish a complex job successfully. Model cards and reports for Claude~3 \cite{anthropic_claude_2024}, GPT-4o \cite{hurst_gpt_2024}, and code models \cite{chen_evaluating_2021} publish these benchmark results as single numbers per model. Broader frameworks like HELM \cite{liang_holistic_2022} add other measures such as calibration, robustness, fairness, and efficiency, but still mainly give point values. In real-world use, Clio \cite{tamkin_clio_2024} collects statistics on what kinds of tasks people actually do with AI with a focus on privacy. Method studies like multilingual chain-of-thought \cite{shi_language_2022} and ReAct \cite{yao_react_2023} suggest ways to improve reasoning and interaction, judged by benchmark accuracy.

Recent studies widely acknowledge that benchmark-based evaluations for LLMs are valuable but inherently limited due to the new characteristics and challenges introduced by LLMs. Benchmarks facilitate progress tracking and model comparison, yet they capture only a partial view of LLM capabilities. Chang et al.~\cite{chang_surveyLLM_2024} note that traditional benchmarks (reusing fixed task sets for ranking models) have driven early advances but are now approaching their limits due to data repetition, task saturation\footnote{Over time, models are trained on similar data, inflating scores and diminishing benchmark usefulness.}, and limited assessment of reasoning or interactivity. Liu et al.~\cite{liu_trustworthy_2023} emphasize that standard benchmarks often overlook moral, safety, and social dimensions of model behavior. In domain-specific contexts, Croxford et al.~\cite{croxford_current_2025} show that medical benchmarks miss crucial aspects such as factual accuracy and clinical reasoning, highlighting the need for domain prior knowledge. Saleh et al.~\cite{saleh_evaluating_2025} observe that efficiency and performance benchmarks remain fragmented and inconsistent. From a safety perspective, Liu et al.~\cite{liu_scales_2025} warn that toxicity and robustness tests are often too narrow and easily gamed, recommending flexible, real-world evaluations. Finally, Ye et al.~\cite{ye_large_2025} introduce psychometric principles, arguing that current benchmarks frequently lack checks for reliability and fairness.

While we acknowledge the usefulness of LLM benchmarks for progress tracking of LLM versions, comparison and ranking of LLMs, they can be used to make claims on LLM reliability. Our HIP-LLM approach goes beyond benchmarks, by framing the reliability assessment of LLMs as statistical inference problems.

\subsection{Statistical Inference for Evaluating LLMs}
Benchmark-based descriptive statistics simply summarize the collected data, whereas statistical inference draws probabilistic conclusions about the \textit{underlying population or process} that generated those data, thereby enabling generalization, prediction and uncertainty quantification. In the AI/ML community, many studies have proposed statistically principled approaches to quantify uncertainty for \textit{individual inputs}, such as text prompts for LLMs, using methods like conformal prediction \cite{quachconformal,wang_2024_conu}. However, these approaches focus on instance-level uncertainty and do not capture uncertainty at the operational reliability level. To the best of our knowledge, only two existing works extend uncertainty quantification to this broader reliability perspective for LLMs.

In the Anthropic report \cite{miller_adding_2024}, Miller frames LLM evaluations as statistical experiments and argues for reporting uncertainty alongside benchmark scores. Using classical inference tools, the study constructs confidence intervals via the Central Limit Theorem, applies clustered standard errors for correlated items, recommends paired inference for model comparisons, and develops power analysis and variance reduction techniques. This work emphasizes that benchmark outcomes should be treated as data drawn from an underlying population, moving beyond descriptive statistics. 

HiBayES \cite{luettgau_hibayes_2025} applies hierarchical Bayesian generalized linear models that: (a) use a Bayesian statistical inference (multilevel Binomial/Poisson Generalized Linear Models) and captures the nested structure of LLM evaluations; (b) employ partial pooling to account for dependencies across levels; (c) provide a fully probabilistic framework, yielding full posterior distributions to quantify uncertainties, rather than relying on point-estimates.

Despite these advances, important limitations remain. Neither Miller’s frequentist framework nor HiBayES explicitly defines/models the OP (even though both implicitly assume that the data represent it). Moreover, neither approach can effectively incorporate (potentially imprecise) prior knowledge: Miller’s work is restricted to classical frequentist inference, while HiBayES relies on non-informative priors. Both also focus on the narrower notion of failure probability as reliability, overlooking the more general and standardized reliability definition in terms of the probability of failure-free runs over future operations. The latter requires predictive inference that accounts for the propagation of uncertainties into the future: a more demanding but practical reliability claim \cite{strigini_software_2013,bishop2022bootstrapping}. \rv{Miller accounts for clustered dependence through variance corrections, but does not introduce hierarchical latent variables or multi-level generative models as in HiBayES or HIP-LLM. Table~\ref{tab_comparison_Hip_Hi_Mill} summarizes the key features and provides a comparative overview of the three methods (HIP-LLM, HiBayES~\cite{luettgau_hibayes_2025}, and Miller~\cite{miller_adding_2024}).}

\begin{table}[t]
{\color{black}
\centering
\caption{Comparison of LLM evaluation frameworks highlighting differences in reliability definition, operational profile modeling, and uncertainty treatment}
\label{tab_comparison_Hip_Hi_Mill}

\renewcommand{\arraystretch}{1.25} 

\begin{tabular}{p{2.5cm} p{3.7cm} p{4.3cm} p{4.3cm}}
\toprule
\textbf{Aspect} 
& \textbf{Miller~\cite{miller_adding_2024}} 
& \textbf{HiBayES}~\cite{luettgau_hibayes_2025} 
& \textbf{HIP-LLM} \\
\midrule

\makecell[l]{Statistical\\ framework}
& Frequentist
& Bayesian 
& Bayesian  \\

\addlinespace[0.1cm]

Primary goal
& \makecell[l]{Uncertainty-aware\\ benchmark reporting} 
& \makecell[l]{Uncertainty-aware,\\hierarchical accuracy\\ estimation} 
& \makecell[l]{Uncertainty-aware\\reliability assessment\\ under OP with\\ embedded priors} \\

\addlinespace[0.2cm]

\makecell[l]{Metric under\\ estimation} 
& \makecell[l]{Failure probability\\ (per task)} 
& \makecell[l]{Failure probability\\ (per task)} 
& \makecell[l]{Reliability (probability\\ of failure-free\\ future tasks)} \\

\addlinespace[0.2cm]


\makecell[l]{OP}
& \makecell[l]{Implicit\\ (dataset assumed\\ representative)} 
& \makecell[l]{Implicit\\ (dataset assumed\\ representative)} 
& \makecell[l]{Explicit,\\ multi-level OPs} \\

\addlinespace[0.2cm]

\makecell[l]{Hierarchical\\ structure} 
& None 
& \makecell[l]{LLM, domains,\\ subdomains} 
& \makecell[l]{LLM, domains,\\ subdomains} \\

\addlinespace[0.2cm]

\makecell[l]{Dependence\\ handling} 
& \makecell[l]{Clustered standard\\ errors} 
& \makecell[l]{Partial pooling} 
& \makecell[l]{Partial pooling} \\

\addlinespace[0.2cm]

\makecell[l]{Prior\\ specification} 
& None 
& \makecell[l]{Uniform precise\\ priors} 
& \makecell[l]{Informative\\ imprecise priors\\ (credal sets)} \\

\addlinespace[0.2cm]

Output 
& \makecell[l]{Confidence\\ intervals} 
& \makecell[l]{Single posterior\\ distribution} 
& \makecell[l]{Posterior\\ envelopes} \\

\bottomrule
\end{tabular}
}
\end{table}



\subsection{Robust Bayesian Analysis}
Robust Bayesian analysis represents a general framework for investigating the sensitivity of posterior measures to uncertainties in the inputs of Bayesian inference \cite{Berger_1994, Insua2012robust}. Regarding uncertainties in priors, while several dedicated methods have been proposed \cite{bishop2010toward,strigini_software_2013,zhao_assessing_2020,Aghazadeh_impact_2024}, Imprecise Probability has emerged as one of the most widely adopted approaches \cite{augustin2014introduction,troffaes2007decision,utkin_imprecise_2018,imprecision_luck_walter_2009,zhao_interval_2020}. It addresses the problem of prior uncertainty by avoiding reliance on \textit{a single} prior distribution, rather representing \textit{a credal set} of plausible priors and deriving posterior bounds that reflect this epistemic uncertainty.

To illustrate the main idea of the Imprecise Probability framework, consider a simple coin-flipping problem where we wish to estimate the probability of heads $\theta$ with data $D$ ($n=10$, $k=3$ heads). As shown in Table~\ref{tab_ip_vs_classical}, instead of a single point estimate or a single Beta posterior distribution, the \textit{posterior envelope} (i.e., a set of posterior distributions) of Imprecise Probability\footnote{For simplicity and illustrative purpose, the example used here is based on the model of \textit{imprecise Linearly Updated Conjugate prior Knowledge (iLUCK)} \cite{imprecision_luck_walter_2009}) which leverages the conjugacy for analytical posterior results. Our HIP-LLM is not using this iLUCK model due to the non-linearity of the hierarchical probabilistic model proposed in Section \ref{sec_dependenSub_independenDom}.} reflects both the observed data and the epistemic uncertainty arising from imprecise prior knowledge (represented by the set of Beta priors). 

\begin{table}[H]
\centering
\caption{Classical Bayesian vs Imprecise Probability for a coin-flipping example.}
\label{tab_ip_vs_classical}
\resizebox{0.9\textwidth}{!}{
\begin{tabular}{lcc}
\toprule
\textbf{Feature} & \textbf{Classical Bayesian} & \textbf{Imprecise Probability} \\
\midrule
Prior & $\mathrm{Beta}(\alpha=2, \beta=2)$ &  $\mathrm{Beta}(\alpha, \beta), \; \alpha \in [1,3], \beta \in [1,3]$ \\
Posterior &  $\theta \mid D \sim \mathrm{Beta}(5,9)$ &  $\theta \mid D \sim \mathrm{Beta}(3+\alpha, 7+\beta), \; \alpha \in [1,3], \beta \in [1,3]$ \\
Poster. mean & $\mathbb{E}[\theta \mid D] = 0.36$ & $\mathbb{E}[\theta \mid D] \in [0.31, 0.38]$ \\
\bottomrule
\end{tabular}
}
\end{table}

Our hierarchical solution, HIP-LLM, is a robust Bayesian approach that addresses the aforementioned gaps. Similar to HiBayES, it adopts a hierarchical structure, but instead of yielding a single posterior, it reports posterior envelopes over imprecise priors and uncertain OPs (as variables).

\section{The Model: HIP-LLM}
\label{sec_dependenSub_independenDom}
The proposed method HIP-LLM stands for Hierarchical Imprecise Probability for Large Language Models reliability assessment. Before introducing the framework, we first formally define the reliability of LLMs.


LLMs are software; therefore, to ensure compatibility with the more general and standardized definition of software reliability (Def.~\ref{def_sw_reliability}), we define LLM reliability as follows:
\begin{definition}[LLM reliability]
\label{def_llm_reliability}
    The probability that an LLM produces failure-free responses over a specified number of future tasks (sequences of closely related queries), under specified (sub-)domains and operational environment.
\end{definition}
The definition of LLM reliability closely parallels the standardized definition of software reliability, as both emphasize \textit{probabilistic} reasoning and the requirement of \textit{failure-free performance} in future operations. The key distinctions arise from the nature of LLMs as ``on-demand'' software. Whereas classical software reliability is typically framed in terms of ``operational time'' that covers continuously operated systems like controllers over clock-time, LLM reliability is defined more explicitly in terms of a specified number of future discrete tasks (each consisting of sequences of related queries). Furthermore, while both definitions account for operational conditions, the LLM context requires explicit reference to domains and subdomains, reflecting the general-purpose and multi-domain design of LLMs. Finally, in traditional software, failure is typically defined as a deviation from the specification, whereas there is no specification for LLMs and thus ``failures of LLMs''\footnote{A complete and formal characterization of what constitutes a ``failure'' for LLMs remains an open research question, and out of the scope of this paper. Cf.~Section~\ref{sec_discussion} for discussions.} are often informally/implicitly characterized by factual errors (hallucinations) \cite{farquhar2024detecting,dahl2024hallucinating} or divergence from human-expert answers \cite{huang2024survey,zhang2024protip}.

\subsection{Problem Formulation}
\label{sec_problem_statement}

While we provide a table of notations in \ref{sec_notation_appendix}, according to Def.~\ref{def_llm_reliability}, LLM reliability can be formalised as follows:
\begin{definition}[Formalised LLM reliability]
\label{def_formal_llm_reliability}
Let $\mathcal{X}$ denote the input-space of all possible tasks for a given LLM, and let $\pi$ be the OP, i.e., a probability distribution over $\mathcal{X}$ that reflects the likelihood of encountering each task $x\in \mathcal{X}$ in practice. Consider a sequence of $n \geq 1$ tasks i.i.d. according to the OP $\pi$, let $\mathbb{I}(x_t) \in \{0,1\}$ indicate success (1) or failure (0) on the $t$-th task $x_t$, then the LLM reliability is:
\begin{equation}
\label{eq_llm_reliability}
R(n,\pi) = \Pr_{x_t \sim \pi}\left(\bigcap_{t=1}^n \{\mathbb{I}(x_t) = 1\}\right)
\end{equation}
\end{definition}
Intuitively, $R(n,\pi)$ represents the probability that the LLM will operate failure-free across the next $n$ i.i.d. tasks according to the OP $\pi$. Importantly, we make the following remarks:
\rv{
\begin{remark}[The i.i.d. assumption vs. contextual memory]
\label{rm_iid_vs_context}
    A key assumption of the above definition is the task failures/success are i.i.d. Bernoulli trials. Such modeling is not uncommon in reliability modeling, especially for critical on-demand systems. A typical justification is when demands are rare, and the states/memory of the software and its operational environment are effectively ``reset'' in-between \cite{strigini_guidelines_1997,salako2023unnecessity}. In the context of LLMs, we carefully define our reliability metric in terms of i.i.d ``tasks'', rather than individual prompts which are often contextually dependent. A task may consist of a sequence of related prompts aimed at achieving a single task goal. We acknowledge that modern LLMs (e.g., ChatGPT) typically retain chat history as contextual memory, which can violate the i.i.d. assumption between tasks. However, most LLMs provide the option to start a new chat session (for a new single task) without any memory and history from previous chat sessions, aligning with the assumption of ``resetting'' the LLM. We note such ``reset'' settings are commonly encountered in LLM offline evaluation like benchmarking \cite{chang_surveyLLM_2024}. In contrast, long-context or agentic workflows, where tasks are sequentially dependent through memory or tool use, violate this assumption and require alternative stateful reliability models which is out of the scope of this paper. 
\end{remark}
\begin{remark}[Failure probability vs. future reliability]
\label{rm_failure_prob_future_reliability}
While the reliability metric $R(n,\pi)$ is the general form of future reliability of processing $n$ tasks, the special case $1-R(1,\pi)$ (where $n=1$) represents the failure probability \cite{strigini_software_2013,bishop2022bootstrapping,zhao2020safety} that studied by, e.g., \cite{miller_adding_2024,luettgau_hibayes_2025}. Accordingly, throughout this paper, we use the term \emph{failure probability} to denote $1 - R(1,\pi)$ and \emph{failure-free runs (of $n$ future tasks)} to denote $R(n,\pi)$, in order to distinguish these two notions. The more general term \emph{reliability} is used to refer to either quantity when the intended meaning is clear from the surrounding context.
\end{remark}
\begin{remark}[General purpose reliability vs. domain-specific reliability]
\label{rm_gen_purp_local_domain}
Since the input-space $\mathcal{X}$ represents \textit{all} possible LLM tasks, so $R(n,\pi)$ is the \textit{general-purpose} reliability of the LLM under study. For (sub-)domain specific reliability, we need to partition $\mathcal{X}$ and derive ``local'' OPs; then the lower level  (sub-)domain specific reliability can be similarly derived like Eq.~\eqref{eq_llm_reliability}.
\end{remark}
\begin{remark}[Binary failure vs. non-binary scoring]
    \label{rm_non_binary_failure}
    In traditional software reliability engineering, failures are naturally defined as binary events, success or failure, according to an explicit system specification. Similarly, most LLM benchmarks introduce task-specific specifications, implemented via automated evaluators or human annotations, which yield binary outcomes for scoring each prompted task. Consistent with this established practice, we model LLM reliability based on binary failure in Def.~\ref{def_formal_llm_reliability}. However, unlike traditional software, what constitutes a ``failure'' for an LLM can be inherently subjective and domain dependent, and external scoring mechanisms may themselves be noisy or inconsistent. Accordingly, reliability in HIP-LLM in Def.~\ref{def_formal_llm_reliability} is formalized conditional on a specified failure definition and scoring process, rather than as an absolute, model-intrinsic property. Explicitly modeling scoring uncertainty or alternative failure definitions remains an important direction for future work.
\end{remark}
}

\rv{
To do statistical inference for the reliability metric defined in Eq.~\eqref{eq_llm_reliability} with assumptions in aforementioned remarks, a simplified ``textbook'' Bayesian model would be the Beta-Binominal one (which also used as one of the baselines in our experiments). This Beta-Binomial estimator applies to a single domain with
precise prior knowledge. Let $\theta := \Pr_{x\sim\pi}(I(x)=1)$ denote the (unknown) probability of success on a random task drawn from the OP $\pi$. Given $N$ i.i.d.\ evaluated tasks with $C$ successes and $N-C$ failures, assume a Binomial likelihood
\[
C, N \mid \theta \sim \mathrm{Binomial}(N,\theta)=\theta^{C}(1-\theta)^{N-C},
\]
and a prior $\Pr(\theta)$. By Bayes' rule, the posterior distribution of $\theta$ is
\begin{equation}
\label{eq_textbook_posterior_theta}
\Pr(\theta \mid C,N)
= \frac{\theta^{C}(1-\theta)^{N-C}\Pr(\theta)}
{\int_{0}^{1}\theta^{C}(1-\theta)^{N-C}\Pr(\theta)\,d\theta}.
\end{equation}
Similarly for the future reliability of passing $n_F$ tasks:
\begin{equation}
     \label{eq_text_book_beta_binomial}
    \Pr(R(n_F,\pi)\mid C, N)=\frac{\theta^{n_F}\theta^{C}(1-\theta)^{N-C}\Pr(\theta)}
    {\int_{0}^{1}\theta^{C}(1-\theta)^{N-C}\Pr(\theta)}
\end{equation}
If $\Pr(\theta)=\mathrm{Beta}(\alpha,\beta)$, conjugacy yields the closed-form posterior
\[
\theta \mid C, N \sim \mathrm{Beta}(\alpha + C,\; \beta + N - C).
\]
For future reliability, thanks to the conjugacy again, the posterior mean reliability
for $n_F$ future tasks is therefore:
\[
\mathbb{E}[R(n_F,\pi)\mid \mathcal{D}]
= \mathbb{E}[\theta^{n_F}\mid \mathcal{D}]
= \frac{B(\alpha + C + n_F,\; \beta + N - C)}
{B(\alpha + C,\; \beta + N - C)},
\]
where $B(\cdot,\cdot)$ denotes the Beta function. Similarly, the posterior PDF and CDF can also be derived and we omit them for brevity.
}

To more rigorously assess the formally defined LLM reliability, coping with the aforementioned Remarks and Gaps, the next subsection introduces our proposed solution HIP-LLM. It models the LLM as a hierarchical structure consisting of independent domains, each containing statistically dependent subdomains (cf.~Fig.~\ref{fig_general_structure}).
\rv{
\begin{remark}[The need of hierarchical modeling on (sub-)domain (in-)dependencies]
    \label{rm_domain_dependencies}
The modeled dependencies and independencies represent the epistemic structure of our hierarchical Bayesian model. That is, observing failures in one (sub-)domain may or may not update our beliefs about the reliability of other (sub-)domains. One possible example of justification\footnote{Cf.~Section~\ref{sec_discussion} for discussions on the validity of this hierarchal dependency assumption.} is: We model coding and law as independent because they rely on distinct competencies of the LLM---coding on formal, symbolic reasoning and syntax manipulation, and law on narrative understanding and normative interpretation. Since these skills draw from largely separate representations and training data, failures in one domain provide little information about failures in the other. In Bayesian terms, their failure probabilities can be treated as a \emph{priori independent parameters}, reflecting separate latent skill dimensions of the model. On the other hand, sub-domains failure probabilities are modeled as dependent parameters given their shared LLM competencies.
\end{remark}
}

Our goal is to infer the posterior distributions of future reliability (and its special case, the failure probability), at the subdomain, domain, and general-purpose LLM levels, based on observed correct responses from tasks within each subdomain.

\begin{figure}[H]
    \centering
    \includegraphics[width=1\linewidth]{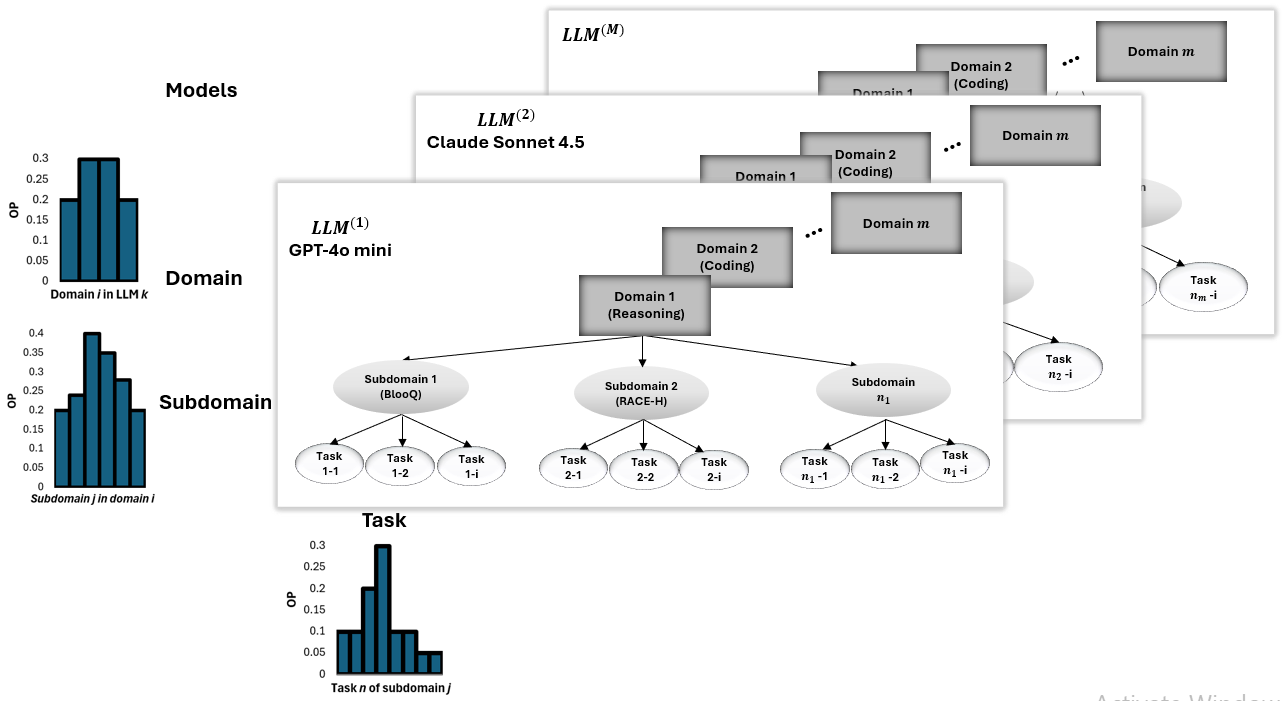}
    \caption{Schematic representation of the hierarchical LLM, domain, and subdomain structure for reliability estimation for M LLM models. \rv{Rectangles indicate independent components, while ovals indicate dependent components within the hierarchy.}}
    \label{fig_general_structure}
\end{figure}

\subsection{Proposed Solution}
\label{sec_solution_hierarchical_Bayes_model}
Consider a hierarchical structure of an LLM comprising independent domains 
$D_1, D_2, \dots,$ $ D_m$, where each domain $D_i$ contains statistically 
dependent subdomains $S_{i1}, \dots, S_{i n_i}$ (Fig.~\ref{fig_imprecise_dependentSub_independentDomain}). We wish to infer the posterior distributions over subdomain, domain and LLM level reliabilities by observing $C_{ij}$ correct responses out of $N_{ij}$ trials (tasks) in each subdomain.

Fig.~\ref{fig_imprecise_dependentSub_independentDomain} presents a detailed view of the hierarchical structure (subdomains $\rightarrow$ domains $\rightarrow$ LLM), with assumed priors and parameters.
We assume that subdomain reliabilities within a domain are \textit{dependent through a shared prior}, and that domain reliabilities are \textit{aggregated} from their subdomains according to task-specific OPs. Our goal is to construct a principled hierarchical Bayesian model that  supports information sharing across dependent subdomains through partial pooling (Sec.~\ref{sec_hierarchical_bayesian_framework}) and uncertainty quantification via Imprecise Probability (Sec.~\ref{sec_uncertainty_handling_imprecise_probability}).

Figures~\ref{fig_general_structure} and \ref{fig_imprecise_dependentSub_independentDomain} illustrate a general hierarchical structure comprising multiple LLM instances ($LLM^{(1)}, LLM^{(2)}, \ldots, LLM^{(M)}$). For clarity, however, Theorems~\ref{thm_subdomain_marginal}--\ref{thm_LLM_nF_dist} focus on the reliability assessment of a \emph{single} LLM system. Accordingly, we omit the superscript $(k)$ and use unindexed symbols ($D_i$, $S_{ij}$, $p_L$). Extending the framework to multiple LLMs is straightforward---apply it to each system independently and compare their posterior distributions.

\begin{figure}[H]
    \centering
    \includegraphics[width=1\linewidth]{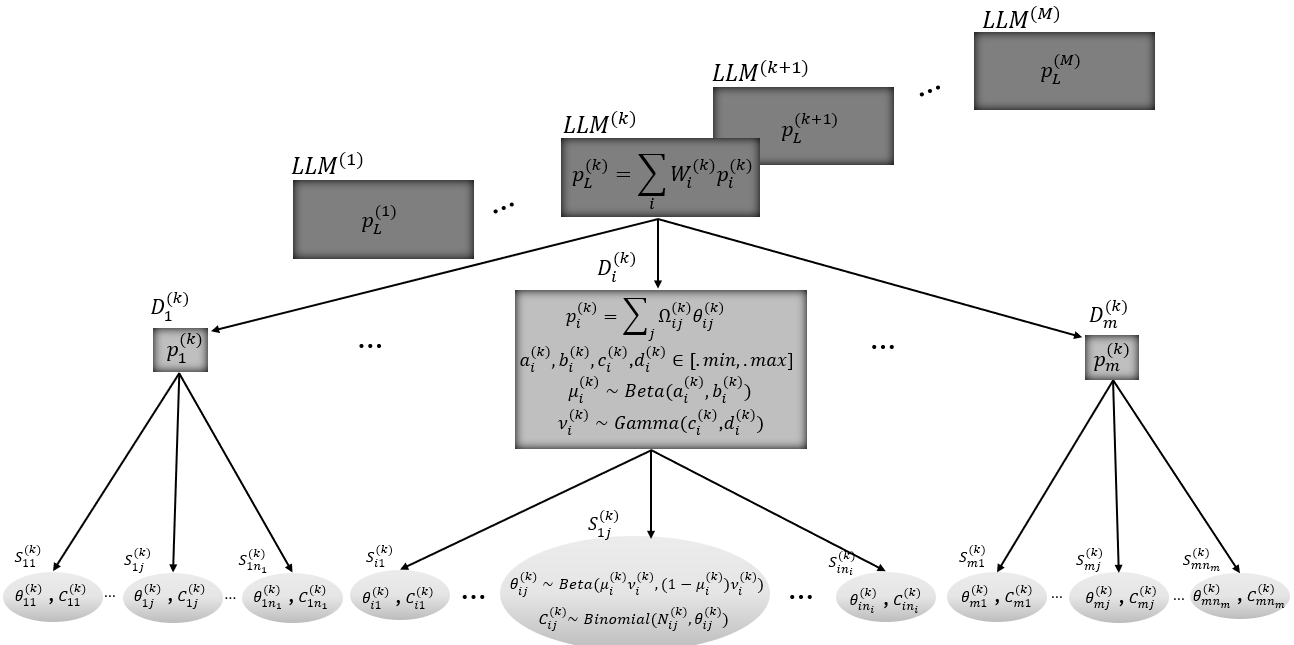}
    \caption{Hierarchical structure with independent domains and dependent subdomains. For readability, parameters and priors are shown only for one subdomain $S_{ij}$ and its parent domain $D_i$ under a representative $LLM_k$; the remaining subdomains, domains ($i=1, \dots, m$), and LLMs $(k=1, \dots ,M)$ are identical, differing only in their indices. \rv{Rectangles indicate independent components, while ovals indicate dependent components within the hierarchy.}}
\label{fig_imprecise_dependentSub_independentDomain}
\end{figure}

\subsubsection{Hierarchical Bayesian Framework for LLM Reliability Modeling}
\label{sec_hierarchical_bayesian_framework}
As seen in Fig.~\ref{fig_imprecise_dependentSub_independentDomain}, inference proceeds hierarchically from the bottom up. 
At the subdomain level, the observed data consist of the number of correct 
responses $C_{ij}$ out of total trials $N_{ij}$. These update the subdomain 
reliabilities $\theta_{ij}$, which are modeled with a Binomial likelihood\footnote{The assumption of independent Bernoulli trials with constant success probability $\theta_{ij}$ may not capture all real dependencies, but it serves as an effective approximation for modeling subdomain outcomes \cite{miller_adding_2024,luettgau_hibayes_2025,zhao_assessing_2020}.
Aggregating across $N_{ij}$ trials then yields the Binomial likelihood, 
whose support $\{0,\dots,N_{ij}\}$ matches the possible counts of correct 
responses observed in subdomain $S_{ij}$.
\[
C_{ij} \mid \theta_{ij}, N_{ij} \;\sim\; \mathrm{Binomial}(N_{ij}, \theta_{ij}).
\]} and a 
Beta prior\footnote{We place a Beta prior on subdomain reliability: 
\[
\theta_{ij} \sim \mathrm{Beta}(\alpha_i, \beta_i).
\] 
The Beta distribution is the most common choice for probabilities bounded in $[0,1]$, 
and it is conjugate to the Binomial likelihood, ensuring closed-form 
updates and computational stability.}. The Beta prior is parameterized by domain-level hyperparameters 
$(\mu_i,\nu_i)$\footnote{To make priors more intuitive and interpretable, we use a reparameterization. 
Instead of specifying the Beta prior directly in terms of $(\alpha_i, \beta_i)$, 
we express it as
\[
\theta_{ij}\mid \mu_i, \nu_i \;\sim\; \mathrm{Beta}(\mu_i\nu_i, (1-\mu_i)\nu_i),
\] 
where $\mu_i$ denotes the expected reliability (prior mean, 
$\mu_i = \mathbb{E}[\theta_{ij}\mid \mu_i,\nu_i]$) and $\nu_i$ denotes the prior 
strength, reflecting the confidence in $\mu_i$ (equivalent sample size or 
pseudo-counts, $\nu_i=\alpha_i+\beta_i$). This reparameterization makes prior 
beliefs easier to specify and justify.}, representing the expected reliability and the prior strength 
within domain $D_i$. These hyperparameters are in turn governed by 
domain-specific hyperpriors:
$\mu_i \sim \mathrm{Beta}(a_i,b_i)$\footnote{This choice reflects the idea that $\mu_i$ itself is a probability lying 
in $(0,1)$, and the Beta distribution provides a flexible family of shapes 
that can express different prior beliefs about domain reliability, ranging 
from diffuse to highly concentrated around particular values (e.g., favoring higher values, lower values, or balanced around 0.5).}, and $\nu_i \sim \mathrm{Gamma}\footnote{This treats $\nu_i$ as a positive random variable ($\nu_i>0$) reflecting how tightly 
subdomains within a domain are assumed to cluster around $\mu_i$. When $\nu_i$ is small,
the prior is diffuse and allows substantial variation across subdomains (weak
pooling). When $\nu_i$ is large, the prior concentrates mass near $\mu_i$ and
subdomains are tightly clustered (strong pooling).}(c_i, \text{rate}\footnote{The Gamma distribution is commonly written in two equivalent forms:
shape--rate and shape--scale. We use the \emph{shape--rate} form,
\(\nu_i \sim \mathrm{Gamma}(c_i,\mathrm{rate}=d_i)\), with
\(\mathbb{E}[\nu_i]=c_i/d_i\) and \(\mathrm{Var}[\nu_i]=c_i/d_i^{2}\).
If a library expects the \emph{shape--scale} form, set
\(\theta_i = 1/d_i\) and write
\(\nu_i \sim \mathrm{Gamma}(c_i,\mathrm{scale}=\theta_i)\);
the two parameterizations are mathematically identical under
\(\theta_i=1/d_i\). The only practical concern is clarity and
reproducibility—accidentally treating a rate as a scale (or vice versa)
would change the prior’s mean and variance.}=d_i)$.

Because the hyperparameters $(\mu_i, \nu_i)$ are shared across all subdomains in domain $D_i$, the information is pooled: subdomains with few observations are stabilized by drawing on evidence from other subdomains in the same domain, while subdomains with many observations are influenced mainly by their own data. This hierarchical setup achieves the desired dependence within domains, while domains remain independent (cf.~Section \ref{sec_discussion} for discussions on this setup).

Once posterior subdomain reliabilities $\theta_{ij}$ are inferred, they are 
aggregated into domain-level reliabilities 
$p_i = \sum_j \Omega_{ij}\theta_{ij}$ using OP weights 
$\Omega_{ij}$, which reflect the practical importance of subdomains. These 
domain-level posteriors are then further combined into the overall LLM reliability, 
$p_L = \sum_i W_i p_i$, using domain operational weights $W_i$. Note, OPs at different levels are represented by the variables $\Omega_{ij}$ and $W_i$, which, for a given LLM use case, can be instantiated either as fixed constants (when usage is certain) or as probability distributions that encode uncertainty about how the LLM will be used in practice \cite{bishop_deriving_2017,pietrantuono2020reliability,popov2025black,popov2025dynamic}.

\subsubsection{Uncertainty Handling via Imprecise Probability}
\label{sec_uncertainty_handling_imprecise_probability}
To address the epistemic uncertainty in prior specification, we adopt an Imprecise Probability approach:
we specify the hyper-hyper-parameters $a_i, b_i, c_i, d_i$ as intervals (Eq.~\ref{eq_hyper_hyper_intervals}) rather than point values. This produces posterior envelopes (lower and upper bounds) at the subdomain, domain, and overall LLM levels.
\begin{align}
\label{eq_hyper_hyper_intervals}
    & a_i \in [a_{i}^{\min}, a_{i}^{\max}], \quad b_i \in [b_{i}^{\min}, b_{i}^{\max}], \nonumber\\
    & c_i \in [c_{i}^{\min}, c_{i}^{\max}], \quad d_i \in [d_{i}^{\min}, d_{i}^{\max}] 
\end{align}

At the subdomain level, we compute posterior bounds for each task type by considering all admissible hyperparameter configurations:
\begin{align}
\underline{Pr}(\theta_{ij} \mid C_i) \leq Pr(\theta_{ij} \mid C_i) \leq \overline{Pr}(\theta_{ij} \mid C_i)
\end{align}
Note, $C_i$ denotes the data is domain $i$. Since we consider dependencies among subdomains, the posterior $\theta_{ij}$ is a function of all data $C_i$ in domain $i$.

Later (Theorem. \ref{thm_domain_posterior} and Theorem. \ref{thm_LLM_posterior}) we will discuss that the closed-form densities (like subdomain level) do not exist for domain and LLM level reliability. 
We represent the cumulative distribution function (CDF) envelopes for the domain and overall LLM levels.

At the domain level, uncertainty propagates upward through OP weights $p_i = \sum_j \Omega_{ij}\theta_{ij}$, 
producing domain reliability bounds:
\begin{align}
    \underline{F}_{p_i}(t \mid C_i) 
= \inf_{h_i \in \mathcal{A}_i} F_{p_i}(t \mid C_i, h_i),
\quad
\overline{F}_{p_i}(t \mid C_i) 
= \sup_{h_i \in \mathcal{A}_i} F_{p_i}(t \mid C_i, h_i),
\quad t \in [0,1] \nonumber
\end{align}
\rv{Here and in the following, $t \in [0,1]$ denotes a generic probability
threshold at which the CDF of the corresponding
(non-failure probability or reliability) random variable is evaluated.}

The domain level posterior $p_i$ only considers $C_i$, given our assumption on cross-domain independence.

At the LLM level, uncertainty aggregates across all $k$ domains through domain weights $p_L = \sum_{i=1}^k W_i p_i$,  while respecting cross-domain independence, resulting in system-level reliability bounds:
\begin{align}
    &\underline{F}_{p_L}(t \mid \text{data}) 
=  \inf_{\mathcal{H} \in \mathcal{A}_{\text{LLM}}} F_{p_L}(t \mid \text{data}, \mathcal{H}),
\nonumber \\
&\overline{F}_{p_L}(t \mid \text{data}) 
= \sup_{\mathcal{H} \in \mathcal{A}_{\text{LLM}}} F_{p_L}(t \mid \text{data}, \mathcal{H}) \nonumber
\end{align}
where, $t \in [0,1]$.

While in the next subsection, we develop Theorems 1--3 of deriving posterior sets for those non-failure probability variables $\theta_{ij}$s, $p_i$s and $p_L$, the posterior distribution sets for future reliability of passing next $n_F$ tasks at each level, e.g., for a domain $i$:
\begin{align}
\underline{Pr}((\sum_j \Omega_{ij}\theta_{ij})^{n_F} \mid C_i) \leq Pr((\sum_j \Omega_{ij}\theta_{ij})^{n_F} \mid C_i) \leq \overline{Pr}((\sum_j \Omega_{ij}\theta_{ij})^{n_F} \mid C_i)
\end{align}
can also be derived, as shown in our Theorems 4--6.

\subsubsection{Theorems}
\label{sec_theorems_one_future_opertaion}

The following theorems are the main mathematical results of HIP-LLM. Intuitively, given the imprecise prior knowledge encoded by the hyperparameters, the theorems derive the posterior distributions of non-failure probabilities at different hierarchical levels, as well as the future reliability, based on the probabilistic reasoning model illustrated in Fig.~\ref{fig_imprecise_dependentSub_independentDomain}, conditioned on the observed task failure data.

\begin{theorem}[Sub-domain level non-failure probability]
\label{thm_subdomain_marginal}
For subdomain $S_{ij}$ in domain $D_i$, let 
$C_i=\{(C_{ik},N_{ik})\}_{k=1}^{n_i}$ be the observed data. 
Let the admissible set of hyperparameters be
\[
\mathcal{A}_i \;=\;
[a_{i}^{\min},a_{i}^{\max}] \times 
[b_{i}^{\min},b_{i}^{\max}] \times 
[c_{i}^{\min},c_{i}^{\max}] \times 
[d_{i}^{\min},d_{i}^{\max}],
\]
and write $h_i=(a_i,b_i,c_i,d_i)$.
Then, \emph{for any} $h_i \in \mathcal{A}_i$, the marginal posterior density of $\theta_{ij}$ is
\[
Pr(\theta_{ij}\mid C_i,h_i) 
= \frac{f_{\mathrm{marg}}(\theta_{ij}, C_i; h_i)}{Z_{\mathrm{marg}}(h_i)},
\]
where $f_{\mathrm{marg}}$ (unnormalized posterior) and $Z_{\mathrm{marg}}$ (normalizing constant) are
\begin{align}
& \hspace{1cm}f_{\mathrm{marg}}(\theta_{ij}, C_i; h_i) 
=\nonumber \\
& \hspace{2cm} \int_0^1 \int_0^{\infty} 
\Bigg[\prod_{k \neq j} \int_0^1 d\theta_{ik}\Bigg] 
L(\boldsymbol{\theta}_i)\,
Pr(\boldsymbol{\theta}_i \mid \mu_i,\nu_i)\,
Pr(\mu_i,\nu_i\mid h_i)\,
d\mu_i\, d\nu_i \nonumber
\end{align}
\[
Z_{\mathrm{marg}}(h_i) 
= \int_0^1 \int_0^{\infty} 
Pr(C_i \mid \mu_i,\nu_i)\,Pr(\mu_i,\nu_i \mid h_i)\,
d\mu_i\, d\nu_i,
\]
with $L(\boldsymbol{\theta}_i)=Pr(C_i\mid \boldsymbol{\theta}_i)$.

The imprecise marginal posterior is characterized by the lower/upper envelopes
\[
\underline{Pr}(\theta_{ij}\mid C_i) 
= \inf_{h_i\in\mathcal{A}_i} Pr(\theta_{ij}\mid C_i,h_i),
\qquad
\overline{Pr}(\theta_{ij}\mid C_i) 
= \sup_{h_i\in\mathcal{A}_i} Pr(\theta_{ij}\mid C_i,h_i).
\]
\end{theorem}

The proof of Theorem~\ref{thm_subdomain_marginal} is presented at Appendix~\ref{sec_proof_subdomain_marginal_appendix}.
In this theorem, the subdomain posterior $Pr(\theta_{ij} \mid C_i, h_i)$
has a closed-form density because of conjugacy: the Beta prior combined with the Binomial likelihood yields a mixture of Beta distributions after marginalizing over the hyperparameters $(\mu_i, \nu_i)$, which can be expressed and evaluated as a proper probability density function.

However, for later Theorems 
closed-form densities do not exist because $p_i = \sum_{j} \Omega_{ij}\theta_{ij}$ and $p_L = \sum_{i} W_i p_i$ are weighted sums of dependent random variables---the distribution of a sum of Beta random variables has no analytical form except in trivial cases\footnote{Assuming the weights are constants. When the weights are modeled as random variables with their own probability distributions, the same problem persists (if not harder).}. Computing such densities would require intractable multi-dimensional integrals. The CDF formulation sidesteps this problem: 
$F_{p_i}(t \mid C_i, h_i) = \int_0^1 \int_0^\infty F_{p_i}(t \mid \mu_i, \nu_i, C_i) \, Pr(\mu_i, \nu_i \mid C_i, h_i) \, d\mu_i \, d\nu_i$ only requires a two-dimensional integral over $(\mu_i, \nu_i)$, where the conditional CDF can be computed via Monte Carlo sampling of independent Betas. Since CDFs provide all necessary information for practical reliability assessment (probabilities, quantiles, expectations), they are the natural representation when densities are unavailable. 



\begin{theorem}[Domain level posterior non-failure probability]
\label{thm_domain_posterior}
For domain $D_i$ with local OP weights $\Omega_{ij}$ 
(where $\sum_{j=1}^{n_i} \Omega_{ij} = 1$), 
let $p_i = \sum_{j=1}^{n_i} \Omega_{ij}\theta_{ij}$ be the domain-level non-failure probability. 
Define the admissible set of hyper-hyper-parameters
\[
\mathcal{A}_i =
[a_{i}^{\min},a_{i}^{\max}] \times 
[b_{i}^{\min},b_{i}^{\max}] \times 
[c_{i}^{\min},c_{i}^{\max}] \times 
[d_{i}^{\min},d_{i}^{\max}],
\]
and write $h_i=(a_i,b_i,c_i,d_i)$. 

Then, for any $h_i \in \mathcal{A}_i$, the posterior distribution of $p_i$ 
is characterized by its CDF:

\begin{align}
F_{p_i}(t \mid C_i, h_i)
  &= \Pr(p_i \le t \mid C_i, h_i) \nonumber\\
  & \hspace{2cm} = \int_0^1 \int_0^\infty F_{p_i}(t \mid \mu_i, \nu_i, C_i)\,
     Pr(\mu_i, \nu_i \mid C_i, h_i)\, d\mu_i\, d\nu_i \nonumber
\end{align}

where
\begin{itemize}
\item $F_{p_i}(t \mid \mu_i, \nu_i, C_i)$ is the conditional CDF of 
$p_i = \sum_{j=1}^{n_i} \Omega_{ij}\theta_{ij}$ given that 
$\theta_{ij} \mid \mu_i, \nu_i, C_i \stackrel{\text{ind.}}{\sim} 
\text{Beta}(C_{ij} + \mu_i\nu_i, N_{ij} - C_{ij} + (1-\mu_i)\nu_i)$ 
for $j=1,\ldots,n_i$,
\item $Pr(\mu_i, \nu_i \mid C_i, h_i)$ is the hyper-posterior obtained via 
Bayes' rule:
\begin{align}
    & Pr(\mu_i, \nu_i \mid C_i, h_i) = \nonumber \\
& \hspace{1cm}\frac{Pr(C_i \mid \mu_i, \nu_i) \, 
\text{Beta}(\mu_i \mid a_i, b_i) \, 
\text{Gamma}(\nu_i \mid c_i, \text{rate}=d_i)}
{\int_0^1 \int_0^\infty Pr(C_i \mid \mu, \nu) \, 
\text{Beta}(\mu \mid a_i, b_i) \, 
\text{Gamma}(\nu \mid c_i, \text{rate}=d_i) \, d\mu \, d\nu} \nonumber
\end{align}
\end{itemize}

The imprecise domain posterior is characterized by CDF envelopes:
\[
\underline{F}_{p_i}(t \mid C_i) 
= \inf_{h_i\in\mathcal{A}_i} F_{p_i}(t \mid C_i, h_i),
\qquad
\overline{F}_{p_i}(t \mid C_i) 
= \sup_{h_i\in\mathcal{A}_i} F_{p_i}(t \mid C_i, h_i).
\]
\end{theorem}
For the proof details, we refer reads to \ref{sec_proof_domain_appendix}.

\begin{theorem}[LLM-level posterior non-failure probability]
\label{thm_LLM_posterior}
For the LLM system with domain weights $W_i$ (where $\sum_{i=1}^{m} W_i = 1$), 
let $p_L = \sum_{i=1}^{m} W_i p_i$ be the LLM-level failure probability and 
$\text{data} = \{C_1, \ldots, C_m\}$ the observed data across all domains. 
Assume cross-domain independence.

Define the domain-level admissible sets
\[
\mathcal{A}_i = 
[a_{i}^{\min}, a_{i}^{\max}] \times 
[b_{i}^{\min}, b_{i}^{\max}] \times 
[c_{i}^{\min}, c_{i}^{\max}] \times 
[d_{i}^{\min}, d_{i}^{\max}]
\]
and write $h_i = (a_i, b_i, c_i, d_i)$ for $i = 1, \ldots, m$. 
Define the LLM-level admissible set as the Cartesian product
\[
\mathcal{A}_{\text{LLM}} = \mathcal{A}_1 \times \cdots \times \mathcal{A}_m,
\]
and collect the domain hyperparameters as 
$\mathcal{H} = (h_1, \ldots, h_m) \in \mathcal{A}_{\text{LLM}}$.

Then, for any $\mathcal{H} \in \mathcal{A}_{\text{LLM}}$, the posterior distribution of $p_L$ 
is characterized by its CDF:
\begin{align}
    &F_{p_L}(t \mid \text{data}, \mathcal{H}) = 
 Pr(p_L \leq t \mid \text{data}, \mathcal{H})
= \nonumber \\
&  \hspace{1.5cm} \int \cdots \int G\big(t \mid \{\mu_i, \nu_i\}_{i=1}^m, \text{data}\big) 
\prod_{i=1}^{m} Pr(\mu_i, \nu_i \mid C_i, h_i) \prod_{i=1}^{m} d\mu_i \, d\nu_i \nonumber
\end{align}
where
\begin{itemize}
\item $G(t \mid \{\mu_i, \nu_i\}_{i=1}^m, \text{data})$ is the conditional CDF of 
$p_L = \sum_{i=1}^{m} W_i p_i$ given all hyperparameters, defined as
\[
G\big(t \mid \{\mu_i, \nu_i\}_{i=1}^m, \text{data}\big) 
= \int_{\mathcal{R}_L(t)} \prod_{i=1}^{m} f_{p_i}(p_i \mid \mu_i, \nu_i, C_i) \, 
dp_1 \cdots dp_m,
\]
where $\mathcal{R}_L(t) := \{(p_1, \ldots, p_m) \in (0,1)^m : \sum_{i=1}^{m} W_i p_i \leq t\}$, 
and $f_{p_i}(\cdot \mid \mu_i, \nu_i, C_i)$ is the conditional density of 
$p_i = \sum_{j=1}^{n_i} \Omega_{ij}\theta_{ij}$ under 
$\theta_{ij} \mid \mu_i, \nu_i, C_i \stackrel{\text{ind.}}{\sim} 
\text{Beta}(C_{ij} + \mu_i\nu_i, N_{ij} - C_{ij} + (1-\mu_i)\nu_i)$,

\item $Pr(\mu_i, \nu_i \mid C_i, h_i)$ is the domain-level hyper-posterior for domain $i$:
\begin{align}
    & Pr(\mu_i, \nu_i \mid C_i, h_i) = \nonumber \\
& \hspace{1cm}\frac{Pr(C_i \mid \mu_i, \nu_i) \, 
\text{Beta}(\mu_i \mid a_i, b_i) \, 
\text{Gamma}(\nu_i \mid c_i, \text{rate}=d_i)}
{\int_0^1 \int_0^\infty Pr(C_i \mid \mu, \nu) \, 
\text{Beta}(\mu \mid a_i, b_i) \, 
\text{Gamma}(\nu \mid c_i, \text{rate}=d_i) \, d\mu \, d\nu} \nonumber
\end{align}

\item Cross-domain independence ensures 
$Pr(\{\mu_i, \nu_i\}_{i=1}^m \mid \text{data}, \mathcal{H}) = 
\prod_{i=1}^{m} Pr(\mu_i, \nu_i \mid C_i, h_i)$.
\end{itemize}

The imprecise LLM posterior is characterized by CDF envelopes:
\begin{align}
    &\underline{F}_{p_L}(t \mid \text{data}) 
= \inf_{\mathcal{H} \in \mathcal{A}_{\text{LLM}}} F_{p_L}(t \mid \text{data}, \mathcal{H}) \nonumber \\
&\overline{F}_{p_L}(t \mid \text{data}) 
= \sup_{\mathcal{H} \in \mathcal{A}_{\text{LLM}}} F_{p_L}(t \mid \text{data}, \mathcal{H}) \nonumber
\end{align}
For more details, see Section~\ref{sec_proof_LLM_appendix}.
\end{theorem}

Theorems~\ref{thm_subdomain_marginal}, \ref{thm_domain_posterior}, and \ref{thm_LLM_posterior} characterize the posterior distribution of non-failure probability $\theta_{ij}$ or aggregated $p_i$, $p_L$. However, in practice, we often care about the reliability over \emph{a specified number of consecutive future operations}, e.g., ``what is the probability that an LLM succeeds on the next 10 tasks in a row?'' or ``what is the probability that the LLM succeeds on the next 20 coding tasks?''.

The following set of theorems extends the aforementioned theorems to characterize the full posterior distribution of reliability for $n_F$ consecutive future operations, i.e., the probability that the LLM operates \(n_F\) consecutive failure-free tasks. Conditioning is on the observed evaluation data across subdomains \(S_{ij}\): \(C_{ij}\) correct generations out of \(N_{ij}\) prompts.

\begin{theorem}[Subdomain posterior reliability for $n_F$ future operations]
\label{thm_subdomain_nF_dist}
For subdomain $S_{ij}$ in domain $D_i$, let $C_i = \{(C_{ik}, N_{ik})\}_{k=1}^{n_i}$ denote all observed data in the domain. Let the admissible set of hyperparameters be
\begin{equation}
\mathcal{A}_i = [a_i^{\min}, a_i^{\max}] \times [b_i^{\min}, b_i^{\max}] \times [c_i^{\min}, c_i^{\max}] \times [d_i^{\min}, d_i^{\max}], \nonumber
\end{equation}
and write $h_i = (a_i, b_i, c_i, d_i)$. Define the reliability random variable\footnote{Comparing to the reliability definition in Def.~\ref{def_formal_llm_reliability}, the OPs are omitted as we assume they are fixed constants in these theorems.} $R_{ij}(n_F) =\theta_{ij}^{n_F}$, representing the probability of $n_F$ consecutive failure-free operations in subdomain $S_{ij}$.

For any $h_i \in \mathcal{A}_i$, the posterior CDF of $R_{ij}(n_F)$ is:
\begin{equation}
F_{R_{ij}(n_F)}(t \mid C_i, h_i) = \Pr(\theta_{ij}^{n_F} \leq t \mid C_i, h_i) = \int_0^{r^{1/n_F}} Pr(\theta_{ij} \mid C_i, h_i) \, d\theta_{ij}, \nonumber
\end{equation}
where $Pr(\theta_{ij} \mid C_i, h_i)$ is the marginal posterior density from Theorem~\ref{thm_subdomain_marginal}.

The imprecise posterior distribution is characterized by the CDF envelopes:
\begin{equation}
\underline{F}_{R_{ij}(n_F)}(t \mid C_i) = \inf_{h_i \in \mathcal{A}_i} F_{R_{ij}(n_F)}(t \mid C_i, h_i), \nonumber
\end{equation}
\begin{equation}
\overline{F}_{R_{ij}(n_F)}(t \mid C_i) = \sup_{h_i \in \mathcal{A}_i} F_{R_{ij}(n_F)}(t \mid C_i, h_i) \nonumber
\end{equation}
\end{theorem}
The transformation $\theta_{ij} \mapsto \theta_{ij}^{n_F}$ generally does not yield a closed-form density. We therefore characterize $R_{ij}(n_F)$ through its CDF. See~\ref{appendix_nF_reliabilit_appendix} computational methods.

\begin{theorem}[Domain posterior reliability for $n_F$ future operations]
\label{thm_domain_nF_dist}
For domain $D_i$ with local OP weights $\Omega_{ij}$ (where $\sum_{j=1}^{n_i} \Omega_{ij} = 1$), let $C_i = \{(C_{ik}, N_{ik})\}_{k=1}^{n_i}$ denote all observed data in the domain. Let the admissible set of hyperparameters be
\begin{equation}
\mathcal{A}_i = [a_i^{\min}, a_i^{\max}] \times [b_i^{\min}, b_i^{\max}] \times [c_i^{\min}, c_i^{\max}] \times [d_i^{\min}, d_i^{\max}], \nonumber
\end{equation}
and write $h_i = (a_i, b_i, c_i, d_i)$. Define the domain-level reliability
\begin{equation}
p_i = \sum_{j=1}^{n_i} \Omega_{ij} \theta_{ij}, \quad R_i(n_F) = p_i^{n_F}, \nonumber
\end{equation}
representing the probability of $n_F$ consecutive failure-free operations at the domain level.

For any $h_i \in \mathcal{A}_i$, the posterior distribution of $R_i(n_F)$ is characterized by its CDF:
\begin{equation}
F_{R_i(n_F)}(t \mid C_i, h_i) = \Pr\left(\left[\sum_{j=1}^{n_i} \Omega_{ij} \theta_{ij}\right]^{n_F} \leq t \,\bigg|\, C_i, h_i\right), \nonumber
\end{equation}
computed by integrating over the joint posterior $Pr(\boldsymbol{\theta}_i \mid C_i, h_i)$.

The imprecise posterior distribution is characterized by the CDF envelopes:
\begin{equation}
\underline{F}_{R_i(n_F)}(t \mid C_i) = \inf_{h_i \in \mathcal{A}_i} F_{R_i(n_F)}(t \mid C_i, h_i), \nonumber
\end{equation}
\begin{equation}
\overline{F}_{R_i(n_F)}(t \mid C_i) = \sup_{h_i \in \mathcal{A}_i} F_{R_i(n_F)}(t \mid C_i, h_i)\nonumber
\end{equation}
\end{theorem}
The CDF can be computed via numerical integration over the joint posterior $Pr(\boldsymbol{\theta}_i \mid C_i, h_i)$ and Monte Carlo sampling, cf.~\ref{appendix_nF_reliabilit_appendix}.

\begin{theorem}[LLM posterior reliability for $n_F$ future operations]
\label{thm_LLM_nF_dist}
For the LLM system with an OP of domain weights $W_i$ (where $\sum_{i=1}^m W_i = 1$), let $\text{data} = \{C_1, \ldots, C_m\}$ denote all observed data across domains. Assume cross-domain independence.

Define the domain-level admissible sets
\begin{equation}
\mathcal{A}_i = [a_i^{\min}, a_i^{\max}] \times [b_i^{\min}, b_i^{\max}] \times [c_i^{\min}, c_i^{\max}] \times [d_i^{\min}, d_i^{\max}] \nonumber
\end{equation}
and write $h_i = (a_i, b_i, c_i, d_i)$ for $i = 1, \ldots, m$. Define the LLM-level admissible set as the Cartesian product
\begin{equation}
\mathcal{A}_{\text{LLM}} = \mathcal{A}_1 \times \cdots \times \mathcal{A}_m, \nonumber
\end{equation}
and collect the domain hyperparameters as $\mathcal{H} = (h_1, \ldots, h_m) \in \mathcal{A}_{\text{LLM}}$.

Define the LLM-level reliability
\begin{equation}
p_L = \sum_{i=1}^m W_i p_i, \quad R_L(n_F) = p_L^{n_F}, \nonumber
\end{equation}
representing the probability of $n_F$ consecutive failure-free operations at the LLM level.

For any $\mathcal{H} \in \mathcal{A}_{\text{LLM}}$, the posterior distribution of $R_L(n_F)$ is characterized by its CDF:
\begin{equation}
F_{R_L(n_F)}(t \mid \text{data}, \mathcal{H}) = \Pr\left(\left[\sum_{i=1}^m W_i p_i\right]^{n_F} \leq t \,\bigg|\, \text{data}, \mathcal{H}\right), \nonumber
\end{equation}
computed by integrating over $\prod_{i=1}^m Pr(p_i \mid C_i, h_i)$.

The imprecise posterior distribution is characterized by the CDF envelopes:
\begin{equation}
\underline{F}_{R_L(n_F)}(t \mid \text{data}) = \inf_{\mathcal{H} \in \mathcal{A}_{\text{LLM}}} F_{R_L(n_F)}(t \mid \text{data}, \mathcal{H}), \nonumber
\end{equation}
\begin{equation}
\overline{F}_{R_L(n_F)}(t \mid \text{data}) = \sup_{\mathcal{H} \in \mathcal{A}_{\text{LLM}}} F_{R_L(n_F)}(t \mid \text{data}, \mathcal{H}) \nonumber
\end{equation}
\end{theorem}
Again, while the closed-form density is not available, the CDF can be computed via numerical integration and Monte Carlo sampling. See~\ref{appendix_nF_reliabilit_appendix} for details.

\rv{
\begin{remark}[Translate vague prior knowledge to hyperparameters]
    \label{rm_pk_narrative}
In HIP-LLM, hyperparameters are not tuning constants but variables used to represent uncertainty about prior knowledge. At the lowest level, $\theta_{ij}$ indicates ``how reliable'' (i.e., the probability that the model succeeds on a randomly drawn task from the subdomain) of $S_{ij}$. Subdomains within the same domain $D_i$ are assumed to be related, and this dependence is captured by assuming that their non-failure probabilities are drawn from a shared Beta distribution governed by two domain-level parameters: $\mu_i$ and $\nu_i$. Here, $\mu_i$ represents the expected reliability of the LLM in domain $D_i$, while $\nu_i$ represents how confident we are in this expectation, or equivalently, how strongly subdomain reliabilities are expected to cluster around $\mu_i$. Since assessors typically do not know the exact values of $\mu_i$ and $\nu_i$, HIP-LLM places distributions on them: $\mu_i \sim \mathrm{Beta}(a_i,b_i)$ and $\nu_i \sim \mathrm{Gamma}(c_i,d_i)$. The hyperparameters $(a_i,b_i,c_i,d_i)$ therefore encode high-level and (likely) vague/imperfect expert beliefs rather than performance data. For example, suppose an assessor believes that the LLM’s reasoning-domain reliability is likely between $0.8$ and $0.9$, but with only moderate confidence corresponding to roughly $30$--$100$ observations (either from previous-use/historical-data or structured expert elicitation experiments based on hypothetical usage scenarios). This belief can be expressed by choosing $a_i \in [8,18]$ and $b_i \in [2,4]$, which yields prior means $\mathbb{E}[\mu_i] \in [0.8,0.9]$, and by choosing\footnote{
\textcolor{black}{The numerical ranges are obtained using the standard interpretations of the Beta and Gamma distributions. 
For the domain-level mean reliability $\mu_i \sim \mathrm{Beta}(a_i,b_i)$, the prior mean is $\mathbb{E}[\mu_i]=a_i/(a_i+b_i)$. 
Choosing $a_i$ and $b_i$ such that this ratio lies between $0.8$ and $0.9$ ensures that the assessor’s belief about the expected reliability is respected (e.g., $8/(8+2)=0.8$ and $18/(18+2)=0.9$). The exact values of $(a_i,b_i)$ are not unique; the interval reflects uncertainty about the precise prior shape. For the confidence parameter $\nu_i \sim \mathrm{Gamma}(c_i,d_i)$ (shape--rate form), the prior mean is $\mathbb{E}[\nu_i]=c_i/d_i$, which can be interpreted as an equivalent number of pseudo-observations supporting the prior belief. 
Setting $d_i=0.1$ and letting $c_i$ vary between $3$ and $10$ yields $\mathbb{E}[\nu_i]\in[30,100]$, matching the assessor’s statement of moderate confidence.
}} $c_i \in [3,10]$ and $d_i = 0.1$, which gives $\mathbb{E}[\nu_i] \in [30,100]$. By allowing these hyperparameters to vary within intervals rather than fixing them to single values, HIP-LLM captures realistic epistemic uncertainty (following the imprecise probability framework \cite{augustin2014introduction,troffaes2007decision} designed to facilitate elicitation of imperfect/vague prior knowledge) and propagates this uncertainty through the hierarchy, yielding posterior reliability bounds rather than overconfident point estimates.
\end{remark}
}


\section{Evaluation}
\label{sec_numerical_example}
To demonstrate and evaluate our HIP-LLM, we empirically investigate five research questions (RQs) in this section. 

\subsection{Research Questions}\label{sec_research_question}

    \textbf{RQ1 (Effectiveness):} How effectively can HIP-LLM assess and compare posterior reliability distributions across different levels of the hierarchy, considering uncertainties propagated from subdomains to domains, and finally to general-purpose LLMs? In this RQ, we aim to demonstrate the use case of our HIP-LLM as a reliability assessment tool. 
    
    \textbf{RQ2 (Sensitivity to hyperparameters):} How sensitive are the posterior estimates of HIP-LLM to the hyperparameters ($a_i, b_i, c_i, d_i$)? These hyperparameters represent the assessors’ (imprecise) prior knowledge, thus understanding their sensitivity to the posteriors may provide insights on how prior knowledge can be elicited.
    
    \textbf{RQ3 (Sensitivity to OPs):} How sensitive are the posterior reliability estimates of HIP-LLM to variations in the OPs that characterize operational usage of LLMs at various levels? The delivered and perceived reliability of an LLM depends on how it will be used, i.e. the OP. We hypothesize that general-purpose LLMs may exhibit lower sensitivity to variations in OPs, whereas LLMs trained for specific (sub-)domains are likely to be more sensitive. We investigate and demonstrate how HIP-LLM can characterize and quantify such OP-dependent reliability variations.
    

    \textbf{RQ4 (Predictability):} How can HIP-LLM predicts future reliability of passing next $n_F$ tasks? While RQ1–RQ3 focus on failure probabilities, which is a special case of reliability, we additionally aim to demonstrate HIP-LLM’s capability to predict the probability of successfully completing $n_F$ future tasks and to quantify how this reliability varies as the reliability requirement $n_F$ changes.


\rv{\textbf{RQ5 (Comparison to baselines):} How does HIP-LLM compare to established and state-of-the-art Bayesian reliability estimators? To compare the accuracy of different models against a ``ground truth'' reliability, we conduct synthetic
simulation experiments where the ``ground truth'' OP and failure probabilities of sub-domains are assumed known.}

\rv{\textbf{RQ6 (Failure definitions):} How can HIP-LLM cope with alternative definitions of task success, such as different pass@k criteria? As per Remark~\ref{rm_non_binary_failure}, the definition of failures of a LLM may vary and subject to uncertainties. While we do not formally study LLM failure definitions, we demonstrate how different failure definitions can be incorporated in HIP-LLM. }

\rv{\textbf{RQ7 (Robustness to memory effects):} How robust are HIP-LLM’s posterior estimates to violations of the i.i.d. assumption caused by memory-induced dependence during LLM evaluation? While HIP-LLM is designed for reset, single-task scenarios, real-world evaluation data may not be strictly generated under this assumption. It is therefore critical to assess whether its posterior estimates remain meaningful when independence is only approximately satisfied or partially violated by memory-induced dependencies.}

\rv{\textbf{RQ8 (Scalability):} How does the computational time and hardware RAM cost of HIP-LLM scale with the number of domains, subdomains, hyperparameter configurations, and Monte Carlo samples? Understanding the scalability of HIP-LLM with these key parameters is essential to determine if it can be applied to large-scale evaluations in practice.}

\subsection{Experimental Setup}

\paragraph{OPs} 
To emphasize the role of OPs, \textit{we ``simulate'' the OPs by assigning probability distributions over these datasets for sampling and by specifying operational weights across (sub-)domains}. Without loss of generality, and consistent with our three-level hierarchy, we define the task-level OP as a uniform distribution and assign operational weights at the subdomain and domain levels proportionally to their dataset sizes. That said, assessors may use alternative distributions as the 3-level OPs when additional information is available, such as those approximated from historical usage data \cite{dong2023reliability,huang2023hierarchical} or user behavior reports \cite{chatterji2025people}. \rv{As demonstrated in RQ3, where we vary the operational weights.}

\rv{
\begin{remark}[Simulated OP from benchmarks vs. real-world OP]
\label{rm_OP_simulation}
Similar to many software reliability modeling studies~\cite{lyu_handbook_1996}, HIP-LLM assumes that the OP is specified and reflects real usage conditions. In the absence of real operational data, and following the same experimental practice as~\cite{luettgau_hibayes_2025,miller_adding_2024}, we leverage existing benchmark datasets to ``simulate'' OPs for fair comparison with baselines. Specifically, benchmark datasets are treated as sampling frames, and proportional dataset sizes are adopted as a proxy OP. This choice neither implies that OP acquisition is solved nor suggests that such proxies are suitable for real deployments; rather, it serves to demonstrate how HIP-LLM integrates OPs when they are provided and to enable comparison with existing benchmark-based methods. Estimating realistic and dynamically evolving OPs from operational data is an established problem in software reliability and is typically studied separately from reliability modeling~\cite{smidts2014software,musa1993operational}. While we believe existing OP estimation techniques may be applicable to LLMs, deriving OPs from real-world LLM usage data introduces new challenges and warrants dedicated future investigation. 
\end{remark}
}

\paragraph{Data} Specifically, we evaluate our hierarchical framework by simulating LLM operational data from four widely used benchmarks, structured into two domains with two subdomains each (same as HiBayEs \cite{luettgau_hibayes_2025}), instantiated by the following datasets:
\begin{itemize}
    \item \textbf{Domain 1 (Coding)}
        \begin{itemize}
             \item \textbf{$\text{Subdomain}_{11}$}: basic Python programming tasks with unit-test based evaluation (MBPP \cite{yu_humaneval_2024}).
             \item \textbf{$\text{Subdomain}_{12}$}: data-science oriented Python problems involving libraries such as \texttt{pandas} and \texttt{numpy} (DS-1000 \cite{lai_ds_2023}).
        \end{itemize}
    \item \textbf{Domain 2 (Reasoning)}:
        \begin{itemize}
             \item \textbf{$\text{Subdomain}_{21}$}: reading comprehension problems where the LLM must answer yes/no questions given short passages (BoolQ \cite{clark_boolq_2019}).
            \item \textbf{$\text{Subdomain}_{22}$}: high school level reading comprehension problems with multiple-choice answers (RACE-H \cite{lai_race_2017}).
        \end{itemize}
\end{itemize}
All experiments were conducted using publicly available APIs of GPT-4o, GPT-4o-mini (from OpenAI), Claude Sonnet 4.5, and Claude Haiku 3.5 
(from Anthropic)\footnote{For academic research purposes only which is permitted under both vendors' terms of service for research publications.}. Table~\ref{tab_eval_data_results} reports sample accuracies per subdomain, i.e., the proportions of correct responses under \texttt{Pass@1}\footnote{Produce 1 solution per task. Score as correct if that solution passes verification.}. To operationalize the i.i.d. Bernoulli trail assumption on tasks, all LLM evaluations were conducted with cleared context between tasks. Each task in MBPP, DS-1000, BoolQ, and RACE-H was processed as an independent API call with no conversational history, ensuring that outcomes are not influenced by previous interactions.

\begin{table}[H]
\caption{Evaluation results using Pass@1. Entries are per-subdomain accuracies ($C_{ij}/N_{ij}$).
Rightmost column shows the row mean across models.}
\centering
\begin{tabular}{llrrrrr}
\toprule
\textbf{Dom} & \textbf{Subdom \scriptsize (Data.)} &
\makecell{4o-mini} &
\makecell{4o} &
\makecell{sonnet-4.5} &
\makecell{haiku-3.5} &
\makecell{Mean}\\
\midrule
$\text{Dom}_1$ & $\text{Subdom}_{11}$ \scriptsize (MBPP)    & 0.440 & 0.471 & 0.450 & 0.447 & 0.452 \\
               & $\text{Subdom}_{12}$ \scriptsize (DS-1000) & 0.490 & 0.420 & 0.493 & 0.483 & 0.472 \\
               & Mean          & 0.465 & 0.446 & 0.472 & 0.465 & -- \\
\midrule
$\text{Dom}_2$ & $\text{Subdom}_{21}$ \scriptsize (BoolQ)   & 0.890 & 0.909 & 0.900 & 0.883 & 0.896 \\
               & $\text{Subdom}_{22}$ \scriptsize (RACE-H)  & 0.820 & 0.552 & 0.840 & 0.859 & 0.768 \\
               & Mean         & 0.855 & 0.731 & 0.87 & 0.871 & -- \\
\midrule
\multicolumn{2}{l}{\textbf{LLM   Mean}} &
0.661 & 0.585 & 0.671 & 0.668 & -- \\
\bottomrule
\end{tabular}
\label{tab_eval_data_results}
\end{table}

From the benchmark accuracy scores, all LLMs appear to perform better in the Reasoning domain than in the Coding domain. Their performances across most subdomains are similar, except for RACE-H, where model OpenAI-4o shows a clear weakness. The small variance between accuracy values for corresponding subdomains across models suggests that the performance differences are more domain-driven than model-driven.

\paragraph{Hyperparameters}
All the hyperparameters we are using for generating the set of figures in the next subsection is shown in \ref{sec_numerical_setting}. 

\paragraph{Interpretation of Posterior CDF Envelopes}
\rv{CDF envelopes are used to jointly represent data uncertainty and epistemic uncertainty arising from imprecise prior knowledge. Since closed-form densities are generally unavailable for weighted sums of dependent reliability parameters, posterior uncertainty is characterized numerically through CDF envelopes, which support probabilistic queries and risk-aware interpretation across all levels of the hierarchy. Unless stated otherwise, all figures in this section follow this convention.}

\rv{Throughout the experimental section, where appropriate, we visualize posterior uncertainty using CDF envelopes. These CDFs are constructed from observed benchmark evaluation data, consisting of $C_{ij}$ correct responses out of $N_{ij}$ tasks for each subdomain $S_{ij}$. Depending on the level of aggregation shown, a CDF represents the posterior distribution of a non-failure probability at the subdomain level ($\theta_{ij}$), the domain level ($p_i$), or the overall LLM level ($p_L$). Domain- and LLM-level CDFs are obtained by aggregating lower-level reliability parameters using the specified OP weights.}

\rv{Monte Carlo sampling is used because the hierarchical aggregation of dependent subdomain reliabilities and imprecise priors leads to posterior distributions that are analytically intractable but can be evaluated efficiently and accurately via numerical sampling.}

\paragraph{Baseline Experimental Configuration}
\rv{Unless otherwise stated, all empirical experiments are conducted using a fixed baseline configuration with
$m=2$ domains, $n=2$ subdomains per domain, $K=160$ hyperparameter configurations per domain,
$S=3000$ Monte Carlo samples per configuration, a $(\mu,\nu)$ integration grid of size $G=2000$,
a fixed CDF evaluation grid of size $T=201$, and a capped number of LLM-level configuration pairings
$K_{\text{total}} \le 512$.
Domain-level parameters $\mu_i$ and $\nu_i$ are treated as latent variables and numerically integrated over a grid, while the hyper-hyperparameters $(a_i,b_i,c_i,d_i)$ are sampled from fixed intervals,
with $a_i,b_i \in [1,12]$ and $c_i,d_i \in [1,25]$, to represent imprecise prior knowledge.
}

\subsection{Results and Analysis}
\label{sec_experiments}
This section reports and analyses the empirical results to the RQs.

\subsubsection{ \textbf{RQ1}(Effectiveness)}
\label{sec_RQ1}
Fig.~\ref{fig_subdomain_CDF_real_data_3Model} shows posterior CDF envelopes for the four subdomains (MBPP, DS--1000, BoolQ, RACE--H) across the four LLMs where a right–shifted CDF indicates higher reliability and a tighter band indicates greater certainty. 

\paragraph{Coding tasks} On MBPP (top left), the envelopes for GPT--4o--mini and Haiku~3.5 are nearly indistinguishable and lie to the left of GPT--4o, while Sonnet~4.5 is right–most but still overlapping with--4o. 
On DS--1000 (top right), GPT--4o--mini lies clearly to the left, while Haiku~3.5 overlaps almost entirely with GPT--4o---their envelopes coincide so closely that Haiku 3.5 is barely distinguishable. Sonnet~4.5 remains right–most, with only partial overlap with GPT--4o and Haiku 3.5.

\paragraph{Reasoning tasks}
On BoolQ (bottom left), GPT--4o--mini is almost identical to Haiku~3.5 on the left, while Sonnet~4.5 and GPT--4o nearly coincide on the right.
Hence the envelopes separate into two close pairs:
$\text{mini} \approx \text{Haiku 3.5} < \text{Sonnet 4.5} \approx 4\text{o}$.
On RACE--H (bottom right), Haiku~3.5 lies clearly left--most,
GPT--4o--mini overlaps partly with Haiku but extends rightward,
followed by GPT--4o, and finally Sonnet~4.5 on the far right.

 \begin{figure}[H]
    \centering
    \begin{subfigure}[t]{0.45\textwidth}
        \centering
        \includegraphics[width=\linewidth]{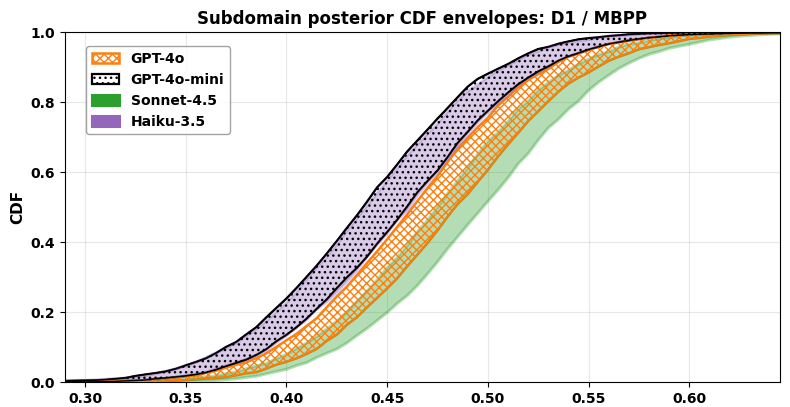}
        \caption{\rv{Posterior CDF envelopes of the non-failure probability for $\text{Subdom}_{11}$ (MBPP dataset in $\text{Dom}_1$) across four models.}}
        \label{fig_subdomain11_CDF}
    \end{subfigure}
    \hfill
    \begin{subfigure}[t]{0.45\textwidth}
        \centering
        \includegraphics[width=\linewidth]{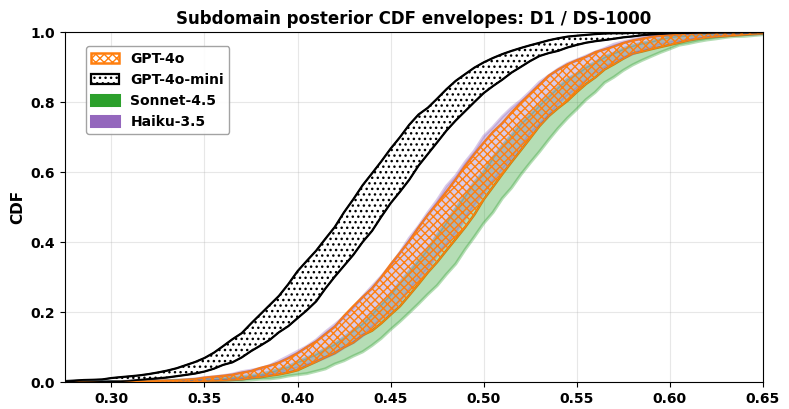}
        \caption{\rv{Posterior CDF envelopes of the non-failure probability for $\text{Subdom}_{12}$ (DS-1000 dataset in $\text{Dom}_1$) across four models.}}
        \label{fig_subdomain12_CDF}
    \end{subfigure}
     \hfill
    \begin{subfigure}[t]{0.45\textwidth}
        \centering
        \includegraphics[width=\linewidth]{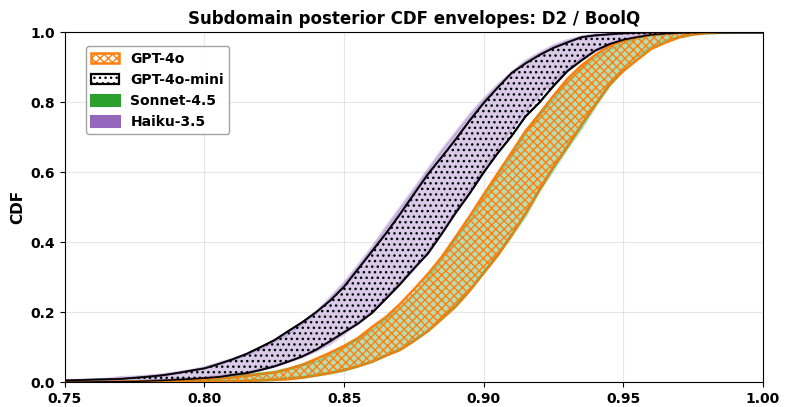}
        \caption{\rv{Posterior CDF envelopes of the non-failure probability for $\text{Subdom}_{21}$ (BlooQ dataset in $\text{Dom}_2$) across four models.}}
        \label{fig_subdomain21_CDF}
    \end{subfigure}
     \hfill
    \begin{subfigure}[t]{0.45\textwidth}
        \centering
        \includegraphics[width=\linewidth]{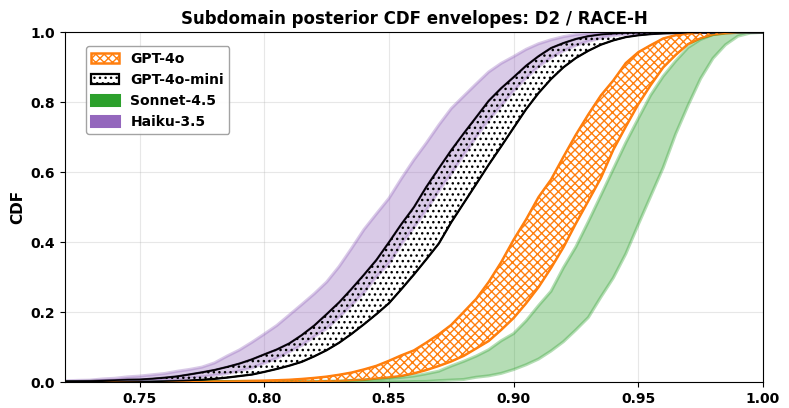}
        \caption{\rv{Posterior CDF envelopes of the non-failure probability for $\text{Subdom}_{22}$ (RACE-H dataset in $\text{Dom}_2$) across four models.}}
        \label{fig_subdomain22_CDF} 
    \end{subfigure}
    \caption{\rv{Posterior CDF envelopes of non-failure probability at the subdomain level. Rows correspond to domains (D$_1$ (coding): MBPP, DS-1000; D$_2$ (reasoning): BoolQ, RACE-H). Subdomains within the same domain are statistically dependent through shared domain-level hyperparameters $(\mu_i,\nu_i)$, where $\mu_i$ represents the average domain reliability and $\nu_i$ controls the strength of coupling (partial pooling) among subdomains. As a result, observations from one subdomain inform the inferred reliability of the others. Domains are assumed statistically independent. A question-
oriented interpretation of figure: After observing the evaluation data, how do the four models differ in terms
of their subdomain-level non-failure probability?
}
}
    \label{fig_subdomain_CDF_real_data_3Model}
\end{figure}

Fig.~\ref{fig_domain_CDF_real_data_3Models} aggregates subdomains within each domain via the operational weights ($\Omega_{ij}$) and reports the posterior CDFs of the domain reliabilities
($p_i$). Note, for simplicity, we just assign the operational weights ($\Omega_{ij}$) proportionally according to the dataset sizes of the two sub-domains in each domain $i$. Similarly to sub-domain results, we may observe and compare domain-level non-failure probabilities $p_i$. Again, comparing to accuracy scores (point estimates) and HiBayEs (single posterior distributions), our HIP-LLM yields distribution envelopes that considers more types of uncertainties. 



 \begin{figure}[H]
    \centering
    \begin{subfigure}[t]{0.46\textwidth}
        \centering
        \includegraphics[width=\linewidth]{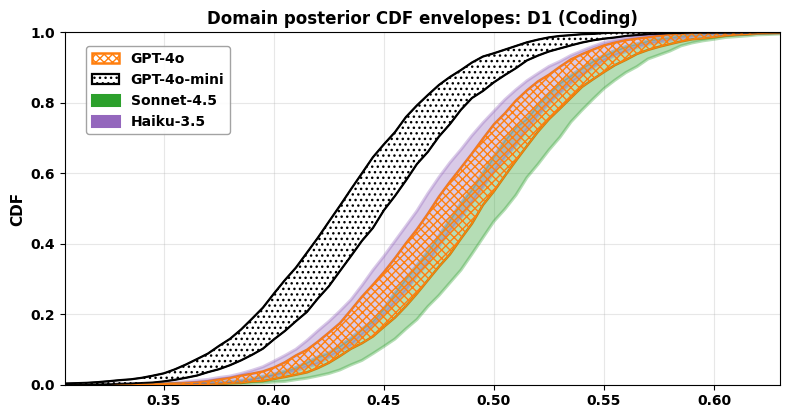}
        \caption{\rv{Posterior CDF envelopes of non-failure probability  $p_1 = \sum_j \Omega_{1j}\theta_{1j}$ for $\text{Dom}_1$ (Coding) across four LLMs,
with subdomain-level operational weights
$\Omega_{1\cdot} = (0.204,\,0.796)$.}}
        \label{fig_domain1_CDF}
    \end{subfigure}
    \hfill
    \begin{subfigure}[t]{0.45\textwidth}
        \centering
        \includegraphics[width=\linewidth]{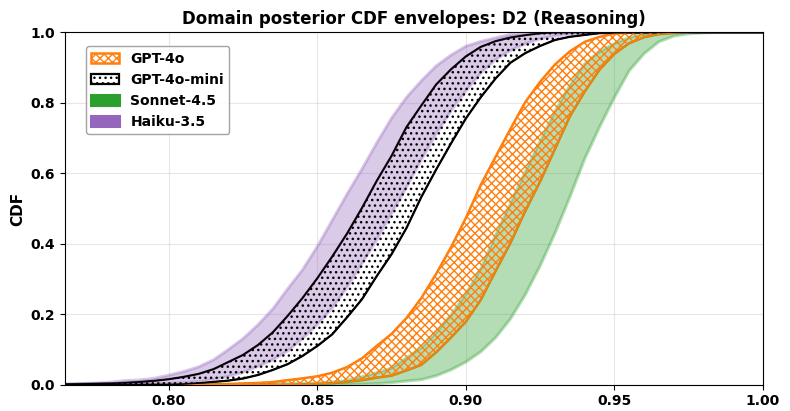}
        \caption{\rv{Posterior CDF envelopes of non-failure probability  $p_2=\sum_j\Omega_{2j}\theta_{2j}$ for $\text{Dom}_2$ (reasoning) across four models, with subdomain-level operational weights (0.483, 0.517).}}
        \label{fig_domain2_CDF}
    \end{subfigure}
    \caption{\rv{Posterior CDF envelopes of non-failure probability of domain-level ($p_i=\sum_j\Omega_{ij}\theta_{ij}$). A question-oriented interpretation of figure: After observing the evaluation data, how do the four models differ in terms of their domain level non-failure probability?}}
    \label{fig_domain_CDF_real_data_3Models}
\end{figure}

Fig.~\ref{fig_LLM_CDF_real_data_3Models} aggregates both domains into the overall LLM reliability $p_L$ using the cross–domain operational weights $W=[0.149,\,0.851]$ for
[Coding, Reasoning]. As before, the weights are simply assigned proportionally according to the dataset sizes of domains. The figure shows, 4o-mini and Haiku bands overlap almost completely, indicating near–equivalent reliability, while GPT--4o and Sonnet~4.5 partially overlap, reflecting moderate but consistent uncertainty between them. Overall, Sonnet~4.5 remains most reliable (which is consistent with Table~\ref{tab_eval_data_results}).

  \begin{figure}[H]
    \centering
    \includegraphics[width=0.56\linewidth]{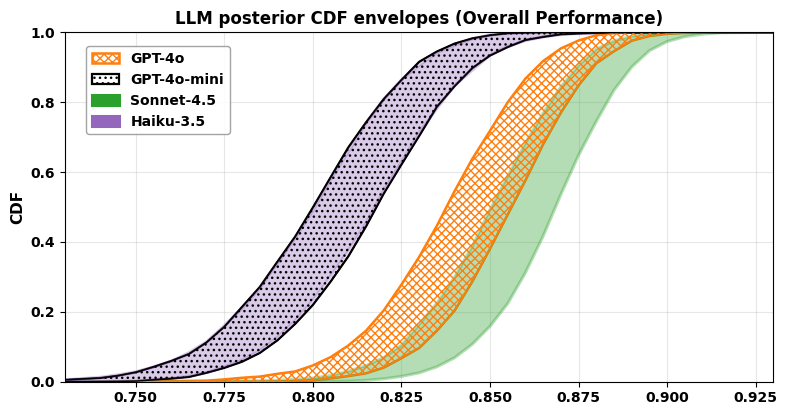}
    \caption{\rv{Posterior CDF envelope of non-failure probability of overall LLM level ($p_L=\sum_i W_i p_i$) with domain weights $W=(0.149,\,0.851)$ across four models. A question-oriented interpretation of figure: \textit{Engineering question:}
After observing the evaluation data, how do the four LLMs differ in terms of their LLM-level non-failure probability?
}}
    \label{fig_LLM_CDF_real_data_3Models}
    \end{figure}

\subsubsection{\textbf{RQ2} (Sensitivity to hyperparameters)}
\label{sec_RQ2}

         \begin{figure}[H]
    \centering
    \begin{subfigure}[t]{0.45\textwidth}
        \centering
        \includegraphics[width=\linewidth]{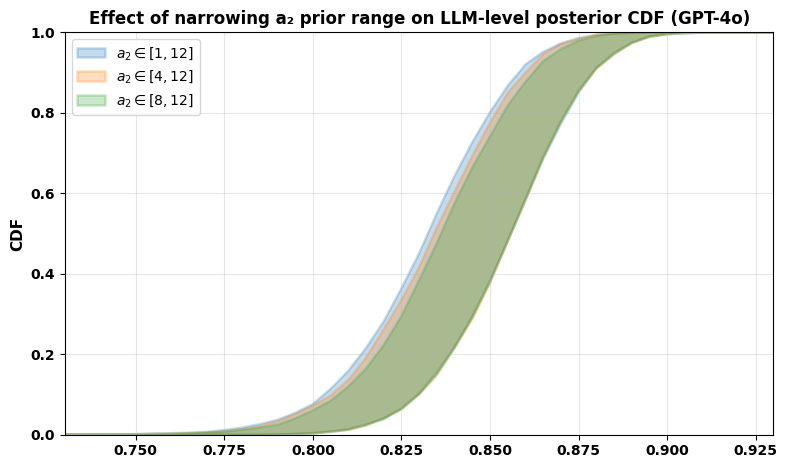}
        \caption{\rv{Sensitivity of posterior CDF envelops of non-failure probability ($p_L=\sum_{i=1}^m W_i p_i$) to $a$.}}
        \label{fig_cdf_ai_effect}
    \end{subfigure}
    \hfill
    \begin{subfigure}[t]{0.45\textwidth}
        \centering
        \includegraphics[width=\linewidth]{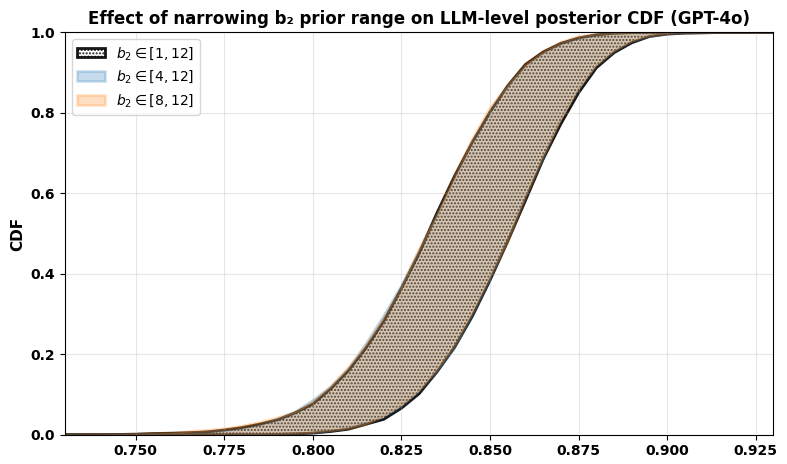}
        \caption{\rv{Sensitivity of posterior CDF envelops of non-failure probability of LLM ($p_L=\sum_{i=1}^m W_i p_i$) to $b$.}}
        \label{fig_cdf_bi_effect}
    \end{subfigure}
     \hfill
    \begin{subfigure}[t]{0.45\textwidth}
        \centering
        \includegraphics[width=\linewidth]{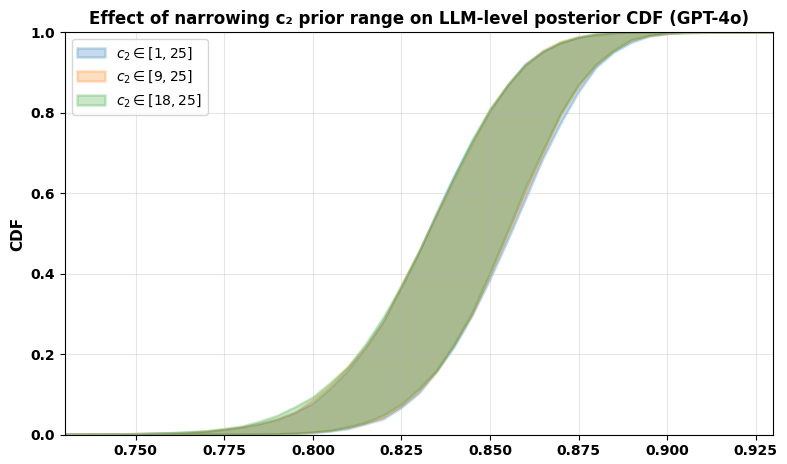}
        \caption{\rv{Sensitivity of posterior CDF envelops of non-failure probability of LLM ($p_L=\sum_{i=1}^m W_i p_i$) to $c$.}}
        \label{fig_cdf_ci_effect}
    \end{subfigure}
     \hfill
    \begin{subfigure}[t]{0.45\textwidth}
        \centering
        \includegraphics[width=\linewidth]{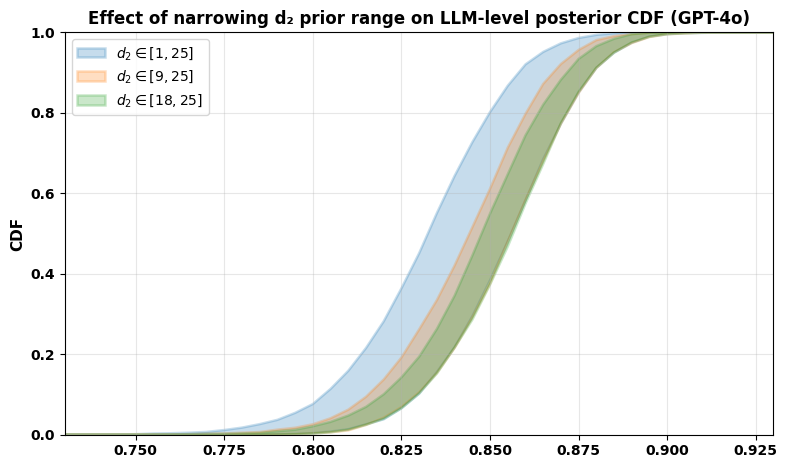}
        \caption{\rv{Sensitivity of posterior CDF envelops of non-failure probability of LLM ($p_L=\sum_{i=1}^m W_i p_i$) to $d$.}}
        \label{fig_cdf_di_effect} 
    \end{subfigure}
    \caption{\rv{Effect of variations in the hyperparameter values ($a$, $b$, $c$, and $d$ ) on the posterior CDF envelops of non-failure probability of LLM ($p_L=\sum_{i=1}^m W_i p_i$). A question-oriented interpretation of the figures: ``What are the LLM-level posterior CDF envelopes of the non-failure probability as domain-level hyperparameters vary within their admissible ranges, after observing the testing data?''}}
    \label{fig_cdf_HyperHyper_effect}
\end{figure}    


\rv{Corresponding to Remark~\ref{rm_pk_narrative}, the sensitivity analysis in Fig.~\ref{fig_cdf_HyperHyper_effect} highlights that different hyperparameters encode qualitatively different aspects of expert elicitation and therefore require different levels of care when specified. Variations in $a_i$ and $b_i$ reflect differences in expert expectations about the average domain-level reliability and mainly affect the location of the posterior reliability envelopes. 
As such, specifying these parameters requires careful consideration of how optimistic or pessimistic prior beliefs are justified by available evidence. In contrast, the hyperparameters $c_i$ and $d_i$ encode expert beliefs about confidence and pooling strength across subdomains, and directly control the width of the posterior envelopes.}


\rv{Fig.~\ref{fig_cdf_ai_effect} illustrates the sensitivity of the LLM-level posterior reliability to the hyperparameter $a_i$ in the prior $\mu_i \sim \mathrm{Beta}(a_i,b_i)$. For example, keeping $b_i=4$, an expert belief corresponding to $a_i=8$ yields a prior mean $\mu_i=0.67$, while a more optimistic belief expressed by $a_i=16$ increases the mean to $\mu_i=0.80$. When propagated through the hierarchical model and updated with the same evaluation data, these beliefs lead to a rightward shift of the posterior CDF envelope, indicating a higher central estimate of $p_L$, while the envelope width remains similar due to the constraining effect of the data.
Fig.~\ref{fig_cdf_bi_effect} shows the complementary effect of varying $b_i$. With $a_2=12$, allowing $b_2 \in [1,12]$ corresponds to prior means $\mu_2 \in [0.50,0.92]$, whereas narrowing the range to $b_2 \in [4,12]$ and $[8,12]$ restricts $\mu_2$ to $[0.50,0.75]$ and $[0.50,0.60]$, respectively. 
}


\rv{Fig.~\ref{fig_cdf_ci_effect} and Fig.~\ref{fig_cdf_di_effect} illustrate the sensitivity of the LLM-level posterior reliability distributions to the hyperparameters $c_i$ and $d_i$. 
Allowing wide ranges such as $c_2,d_2 \in [1,25]$ corresponds to admitting both weak and strong confidence scenarios, while progressively narrowing these ranges (e.g., $c_2 \in [9,25]$ or $[18,25]$, and similarly for $d_2$) excludes low- or high-confidence assumptions. When propagated through the hierarchical model and updated with the same evaluation data, these changes leave the central location of the posterior CDF largely unchanged but systematically tighten or widen the envelope, indicating that $c_i$ and $d_i$ primarily control epistemic uncertainty rather than the expected reliability level. The stronger effect of $d_i$ occurs because it directly reduces the pooling strength $\nu_i$, while $c_i$ mainly affects how concentrated this belief is. As a result, changing $d_i$ more strongly weakens pooling across subdomains and leads to larger changes in the width of the posterior reliability envelopes.
}

\subsubsection{ \textbf{RQ3} (Sensitivity to OPs)} 
\label{sec_RQ3}

 \begin{figure}[H]
    \centering
    \begin{subfigure}[t]{0.49\textwidth}
        \centering
        \includegraphics[width=\linewidth]{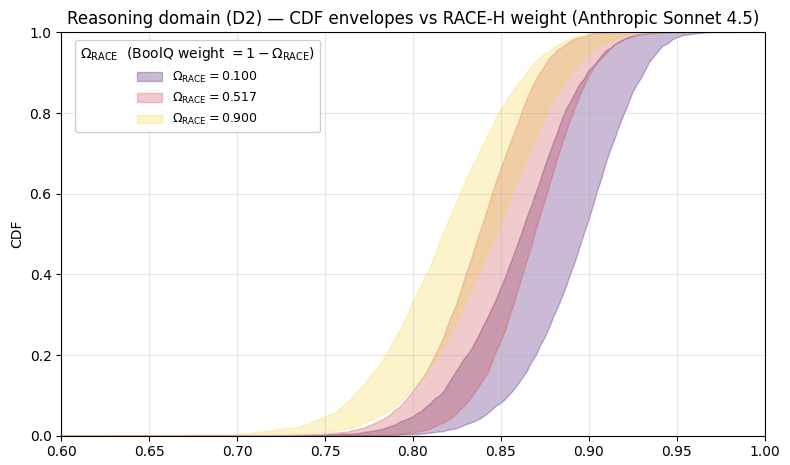}
        \caption{\rv{Posterior CDF envelopes of non-failure probability for the Reasoning domain ($\mathrm{D}_2$) under alternative operational weights on RACE--H, for Anthropic Sonnet~4.5.}}
        \label{fig_OP_effect_Sonnet}
    \end{subfigure}
    \hfill
    \begin{subfigure}[t]{0.49\textwidth}
        \centering
        \includegraphics[width=\linewidth]{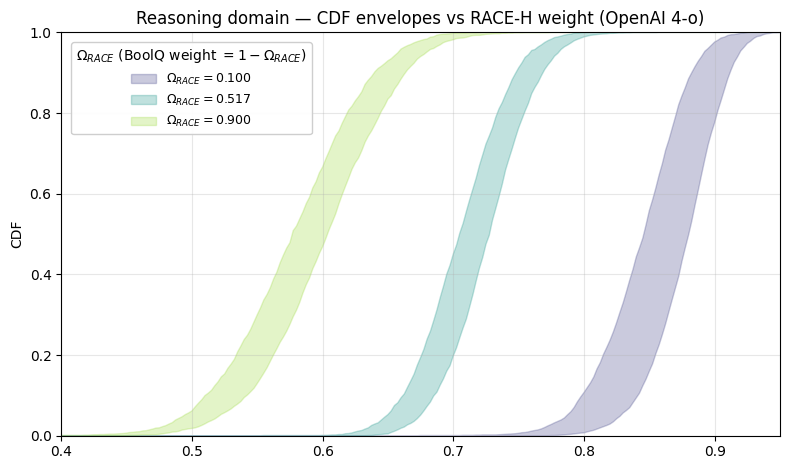}
        \caption{\rv{Posterior CDF envelopes of non-failure probability for the Reasoning domain ($\mathrm{D}_2$) under alternative operational weights on RACE--H, for OpenAI 4-o.}}
        \label{fig_OP_effect_4-o}
    \end{subfigure}
    \caption{\rv{Posterior CDF envelopes of non-failure probability for the Reasoning domain ($\mathrm{D}_2$) under alternative operational weights on RACE--H,
$\Omega_{\mathrm{RACE}} \in \{0.10,\,0.517,\,0.90\}$, with
$\Omega_{\mathrm{BoolQ}} = 1 - \Omega_{\mathrm{RACE}}$.
\textbf{(a)} Sonnet~4.5; \textbf{(b)} OpenAI~4-o. A question-oriented interpretation of figure: How does the inferred domain-level non-failure probability change as the operational profile shifts emphasis between RACE–H and BoolQ?}
}
    \label{fig_operational_profile_posterior}
\end{figure}

Fig.~\ref{fig_operational_profile_posterior} shows posterior CDF envelopes for the
Reasoning-domain reliability $p_2$ under three OPs,
$\Omega_{\mathrm{RACE}}\in\{0.10,\,0.517,\,0.90\}$ with
$\Omega_{\mathrm{BoolQ}}=1-\Omega_{\mathrm{RACE}}$.
Fig.~\ref{fig_OP_effect_Sonnet} reports Anthropic Sonnet~4.5;
Fig.~\ref{fig_OP_effect_4-o} reports OpenAI 4-o.

\textit{(a) Sonnet~4.5.} As $\Omega_{\mathrm{RACE}}$ increases the CDF shifts slightly
left (lower delivered $p_2$), with strong overlap among bands. This weak sensitivity
to the OP is consistent with the subdomain accuracies being both high
and similar ($\theta_{\text{BoolQ}}\approx0.90$ vs.\ $\theta_{\text{RACE-H}}\approx0.84$):
the mixture $p_2=\Omega_{\mathrm{BoolQ}}\theta_{\text{BoolQ}}+\Omega_{\mathrm{RACE}}\theta_{\text{RACE-H}}$
changes only modestly as weight moves from BoolQ to RACE--H. The mild widening of the
envelope at higher $\Omega_{\mathrm{RACE}}$ reflects slightly larger posterior
uncertainty on RACE--H relative to BoolQ.

\textit{(b) OpenAI 4-o.} In contrast, the envelopes are well separated and move
substantially left as $\Omega_{\mathrm{RACE}}$ increases. This strong dependence on
the operational weighting arises because 4-o’s subdomain accuracies differ widely
($\theta_{\text{BoolQ}}\approx0.91$ vs.\ $\theta_{\text{RACE-H}}\approx0.55$), so the
mixture $p_2$ is dominated by whichever subdomain receives more weight. 

Overal, the posterior CDF envelopes for OpenAI 4-o exhibit a strong
dependence on the BoolQ/RACE--H weighting, whereas Sonnet~4.5 is comparatively
invariant to $\Omega_{\mathrm{RACE}}$. This arises because GPT-4-o’s subdomain accuracies differ widely ($0.91$ vs $0.55$), producing a hierarchical posterior with pronounced weight sensitivity. These results suggest that GPT-4-o’s reasoning performance is more uneven across benchmark types, while Sonnet-4.5 demonstrates domain-level robustness.

\rv{From a practical point of view, these results show that changing the operational weights can meaningfully change the reliability experienced by users, especially for OpenAI~4-o. The operational weights represent how often different types of tasks occur in real use, so changing them corresponds to a change in user behavior or application context. For 4-o, different task mixes lead to clearly different reliability levels, meaning that the same model can appear reliable in one application but much less so in another. In contrast, Sonnet~4.5 shows relatively stable reliability across different task mixes, indicating more uniform performance. This highlights that understanding the expected task distribution is important when deploying LLMs, particularly for models whose performance varies strongly across subdomains.
}

\subsubsection{\textbf{RQ4} (Predictability)}
\label{sec_RQ4}
Fig.~\ref{fig_reliability_vs_nF} presents the expected LLM-level reliability $\mathbb{E}[R_L(n_F)]$ as a function of the operational horizon $n_F$, where $n_F$ denotes the number of consecutive operations in the future. The horizontal axis displays $n_F$ on a logarithmic scale, while the vertical axis shows the expected reliability---the posterior mean probability that the LLM successfully completes all next $n_F$ consecutive tasks without failure. 

\begin{figure}[H]
    \centering
    \begin{subfigure}[t]{0.45\textwidth}
        \centering
        \includegraphics[width=\linewidth]{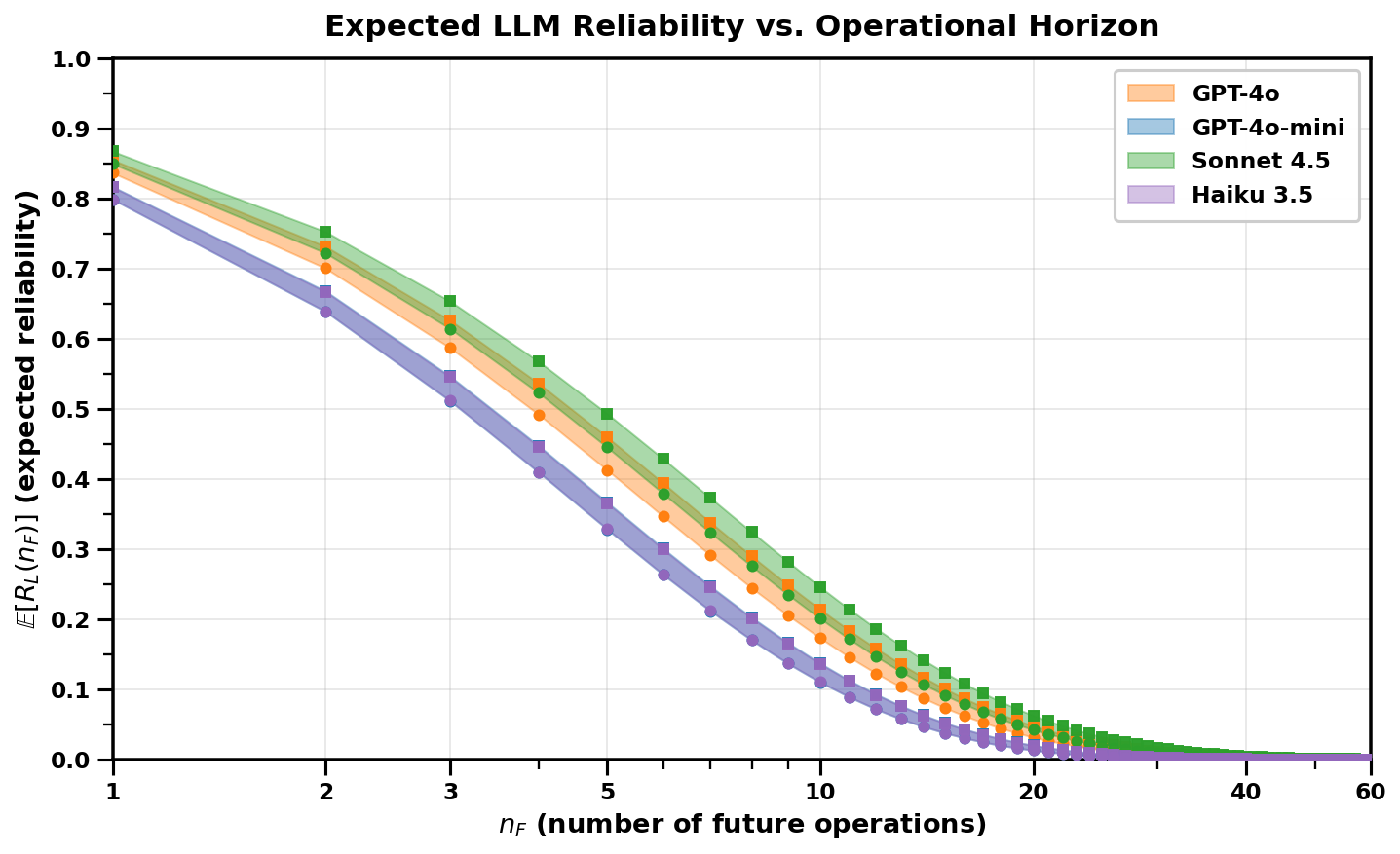}
        \caption{ 
        \rv{Imprecise posterior envelopes of the expected LLM reliability ($\mathbb{E}[R_L(n_F)]$) versus the required number of failure-free future tasks $n_F$ for reliability claims.  Imprecise posterior envelopes show decay in reliability over consecutive tasks across four models.}
}
        \label{fig_reliability_vs_nF}
    \end{subfigure}
    \hfill
    \begin{subfigure}[t]{0.45\textwidth}
        \centering
        \includegraphics[width=\linewidth]{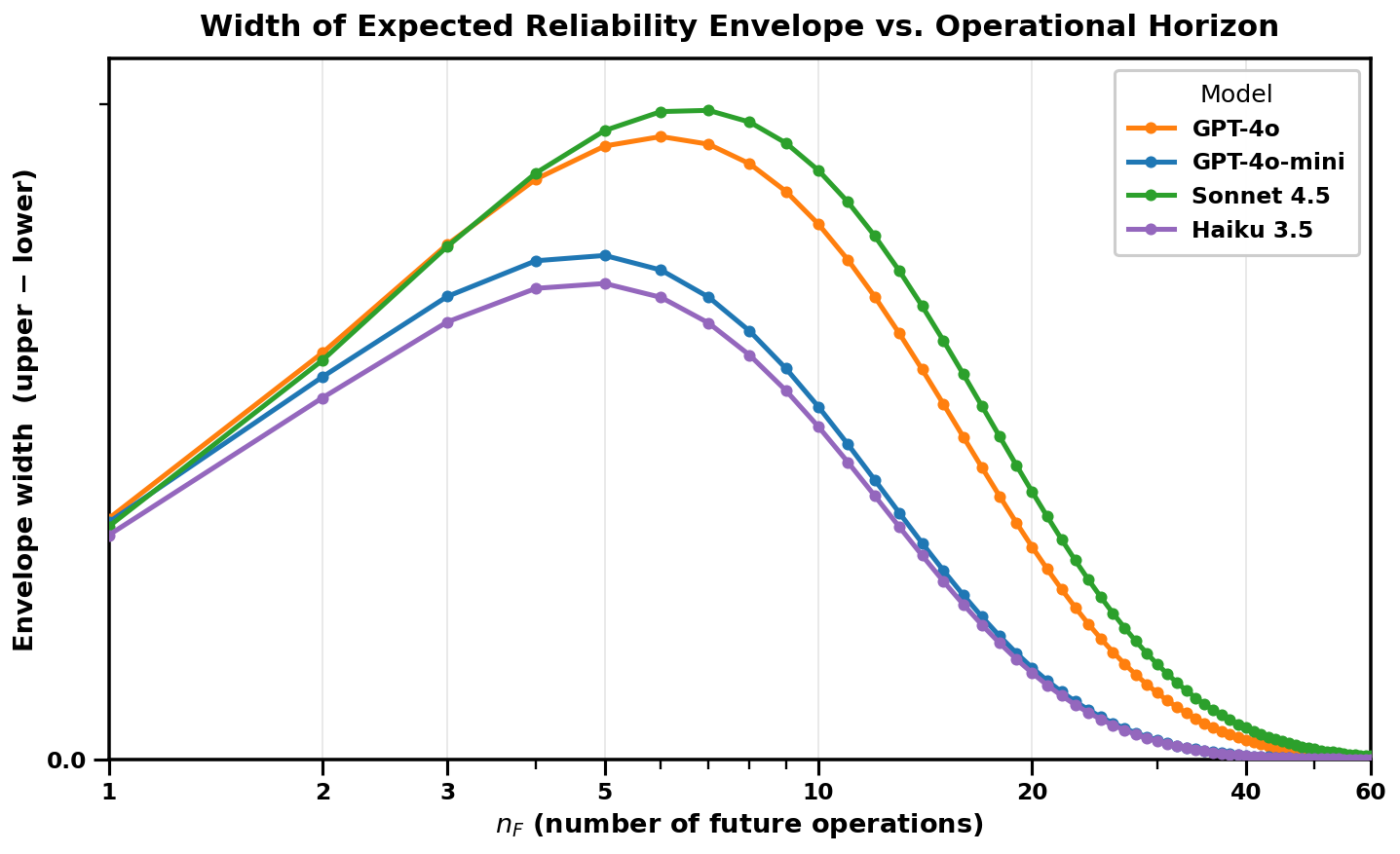}
        \caption{
        \rv{Envelope widths (upper–lower) of expected reliability $\mathbb{E}[R_L(n_F)]$. Each curve represents the spread (upper–lower bound) of expected reliability as the number of future operations increases. The bell-shaped profile indicates that uncertainty grows with operational horizon before decaying as overall reliability converges toward zero.
}
        }
        \label{fig_Width_Expected_reliability}
    \end{subfigure}
    \caption{\rv{Expected reliability envelopes and their width as functions of number of failure-free future tasks required by assessors (i.e., $\mathbb{E}[R_L(n_F)]$ and $[\text{upper bound of envelop} - \text{lower bound of envelop}]$).  Question-oriented interpretation of the figure: ``what is the posterior  expected reliability of passing next $n_F$ future tasks, after seeing testing data?''}
}
    \label{fig_reliability}
\end{figure}


At $n_F=1$ all models exhibit high expected reliability,
$\mathbb{E}[R_L(1)] \approx 0.80\text{--}0.87$, i.e., a high
probability of producing a correct answer on the very next task
(the non‐failure probability).
As the horizon length $n_F$ increases, expected reliability drops
sharply for every model. 
Across all horizons the ordering is stable: Sonnet-4.5 achieves the highest reliability, GPT-4o follows closely,
while Haiku-3.5 and GPT-4o-mini form the lower
pair (with GPT-4o-mini typically the lowest). However, as shown in As Fig.~\ref{fig_Width_Expected_reliability}, the uncertainty (measured by the interval width) of Sonnet-4.5's reliability prediction is also the highest.

As Fig.~\ref{fig_Width_Expected_reliability} shows, at $n_F=1$ the envelopes are quite tight, indicating low uncertainty. The bands widen modestly for small–to–medium horizons (i.e., when $N_F$ increases from 2 to $\sim6$), as a result of compounded prediction uncertainties. As $n_F$ increases, all curves converge toward zero, indicating reduced uncertainty in the reliability prediction. Intuitively, this suggests greater confidence that the LLMs are unlikely to pass all $n_F$ future tasks.

\rv{From a user perspective, Fig.~\ref{fig_reliability} illustrates how reliability claims vary with the required number of future failure-free tasks $n_F$: as $n_F$ increases, the claim becomes more stringent, leading to a corresponding decrease in confidence. 
This provides practical guidance for selecting models based on intended usage. For example, in healthcare decision support, an assessor may require a very high reliability of failure-free performance for a small number of critical diagnostic queries, favoring models with stronger short-horizon reliability guarantees (i.e., a higher but more steeply declining curve). In legal document analysis, a compliance team may comprise reliability a bit over a longer sequence of tasks (i.e., a lower but flatter curve), such as cross-checking clauses or summarizing precedents. In this way, Fig.~\ref{fig_reliability} illustrates how model reliability scales with the required number of future failure-free tasks and supports informed deployment decisions across different application scenarios.
}

{\color{black}
\subsubsection{\textbf{RQ5} (Comparison to baselines)}
\label{sec_RQ5}
We evaluate the proposed \textsc{HIP-LLM} method against three Bayesian baselines
using controlled synthetic experiments in which the \emph{ground-truth system
reliability} (GT) is known. The system consists of two domains, each comprising two
subdomains. Each subdomain is characterized by a reliability
$\theta_{ij}$ and an OP weight $\text{OP}_{ij}$. The overall
system reliability is defined as the OP-weighted average of subdomain reliabilities\footnote{In the paper, the OP is defined hierarchically, with $W_i$
representing the importance of each domain and $\Omega_{ij}$ representing the
relative importance of subdomain $j$ within domain $i$, so that the system
reliability is given by $p_L = \sum_i W_i \sum_j \Omega_{ij}\,\theta_{ij}$. In the
experiments, this structure is written in an equivalent but simpler form by
specifying directly the combined subdomain weights
$\mathrm{OP}_{ij} = W_i\,\Omega_{ij}$. The system reliability is then computed as
$p_L = \sum_{i,j} \mathrm{OP}_{ij}\,\theta_{ij}$, which is algebraically identical
to the hierarchical expression. Under this representation, domain weights are
not fixed or ignored: they are implicitly given by $W_i = \sum_j \mathrm{OP}_{ij}$,
and the within-domain weights are recovered as
$\Omega_{ij} = \mathrm{OP}_{ij}/W_i$. This choice simplifies the experimental
setup by defining the OP at the subdomain level, while remaining
fully consistent with the hierarchical formulation used in the paper.
},
\begin{equation}
p_L = \sum_{i,j} \text{OP}_{ij}\,\theta_{ij}.
\end{equation}

In Bayesian inference, the posterior distribution reflects a combination of prior beliefs and observed data, with their relative influence governed by the amount of available evidence. As the sample size increases, the posterior becomes increasingly dominated by the data rather than the prior\footnote{In the limiting cases, when no data are observed the posterior coincides with the prior, whereas with an infinite amount of data the posterior converges to ``frequentist'' estimates driven solely by the data.}. Thus, to illustrate HIP-LLM’s ability to incorporate informative prior and to examine the influence of priors on posterior estimates, we consider two distinct sample-size regimes:
\begin{itemize}
    \item \textbf{Small-$N$:}
    $[100,\,500,\,1000,\,300]$ samples per subdomain, where priors may
    have more influence on the posterior.
    \item \textbf{Large-$N$:}
    $[1000,\,5000,\,10000,\,3000]$ samples per subdomain, where data are expected to dominate more than the priors.
\end{itemize}
For each regime, synthetic binomial observations are generated from the same ground-truth subdomain reliabilities. Inference is performed under three OP scenarios:
\begin{itemize}
    \item \textbf{Dataset-based OP ($\text{OP}^{\text{data}}$):} OP weights are set proportional to dataset sizes, introducing a structural mismatch with respect to the ground truth OP. This setting represents a naive scenario in which the OP is ignored and assessors rely solely on (benchmark) datasets for reliability assessment.
    \item \textbf{Approximated OP ($\text{OP}^\text{approx}$):} The ground truth OP is perturbed by $\pm 20\%$ noise and renormalized, representing a scenario in which assessors recognize the importance of OPs and apply existing OP estimation methods to approximate them, rather than relying solely on evaluation datasets (i.e., $\text{OP}^{\text{data}}$), while still incurring estimation errors.
    \item \textbf{Ground-truth OP ($\text{OP}^{\text{GT}}$)}: The ground truth OP is used directly, representing an idealized oracle scenario in which the assessors know the OP for certain.
\end{itemize}

We compare the following inference methods:
\begin{itemize}
    \item \textbf{BB-UnInf}: Independent Beta--Binomial model with non-informative
    priors (Beta(1,1)).
    \item \textbf{BB-Inf}: Independent Beta--Binomial model with informative
    priors centered on the GT reliabilities.
    \item \textbf{HiBayES (\cite{luettgau_hibayes_2025}):} A hierarchical Bayesian model with partial pooling
    across subdomains.
    \item \textbf{HIP-LLM:} The proposed imprecise-probability approach, producing
    an envelope over posterior distributions induced by a family of priors.
\end{itemize}


The criteria for comparison are as follows:
\begin{itemize}
    \item \textbf{Median ($\hat{p}_L^{\mathrm{med}}$):} For methods producing a single posterior distribution (BB-UnInf, BB-Inf, HiBayES), the median is defined as the empirical median of the posterior samples of $p_L$. For HIP-LLM, which induces a \emph{set of posterior distributions} indexed by hyperparameter configurations $h \in \mathcal{H}$, we first compute the median of each posterior,
    \[
    m_h = Q_{0.50}\big(p_L^{(h)}\big),
    \]
    and then report the envelope of posterior medians,
    \[
    \hat{p}_L^{\mathrm{med}} \in \left[ \min_{h \in \mathcal{H}} m_h,\; \max_{h \in \mathcal{H}} m_h \right].
    \]
    \item \textbf{Error:} For single-posterior methods, the error is defined as the absolute deviation between the posterior median and the ground-truth reliability,
    \[ 
    \text{Error} = \left| \hat{p}_L^{\mathrm{med}} - p_L^{\mathrm{GT}} \right|.
    \]
    For HIP-LLM, the error is computed \emph{per posterior envelope} as
    \[
    e_h = \left| m_h - p_L^{\mathrm{GT}} \right|,
    \]
    and the reported error corresponds to the envelope of admissible errors,
    \[
    \text{Error} \in 
    \left[ \min_{h \in \mathcal{H}} e_h,\; \max_{h \in \mathcal{H}} e_h \right].
    \]

    \item \textbf{Credible Interval (CI):} For single-posterior methods, a Bayesian $90\%$ credible interval is defined as
    \[
    \text{CI}_{0.90}
    =
    \left[
    Q_{0.05}\big(p_L\big),\;
    Q_{0.95}\big(p_L\big)
    \right],
    \]
    where $Q_q(\cdot)$ denotes the empirical quantile.

    For HIP-LLM, which produces a family of posteriors $\{p_L^{(h)}\}_{h \in \mathcal{H}}$, we report the widest credible interval compatible with the assumed prior uncertainty,
    \[
    \left[
    \min_{h \in \mathcal{H}} Q_{0.05}\big(p_L^{(h)}\big),\;
    \max_{h \in \mathcal{H}} Q_{0.95}\big(p_L^{(h)}\big)
    \right].
    \]

\end{itemize}

Tables~\ref{tab_smallN_single} and \ref{tab_largeN_single} illustrate how system
reliability estimates are influenced by the OP and prior
assumptions. In the Small-$N$ regime, where data are
limited, all methods slightly underestimate the ground-truth reliability
($p_L^{\mathrm{GT}} = 0.5860$), with posterior medians typically between $0.57$
and $0.58$. In this setting, both OP mismatch and prior choice matter: using an
approximate or dataset-based OP increases error compared to the ground-truth OP,
and informative priors can reduce point error when they are correctly aligned
with the GT subdomain reliabilities. This reduction, however, is due to favorable
prior specification rather than inherent robustness and would not persist if the
prior were misspecified.

\begin{table}[H]
\centering
\caption{Posterior estimates under the Small-N ($[100,\,500,\,1000,\,300]$) regime
(\(p_L^{\text{GT}} = 0.5860\)).}
\label{tab_smallN_single}
\begin{tabular}{llccc}
\toprule
OP & Method &  \makecell{Median \\ value/ bound} & \makecell{Error \\ value/bound} & 90\% Interval  \\
\midrule
\multirow{4}{*}{$\text{OP}^{\text{data}}$}
 & BB-UnInf  & 0.5733 & 0.0127 & [0.5549, 0.5917] \\
 & BB-Inf & 0.5749 & 0.0111 & [0.5569, 0.5924] \\
 & HiBayES    & 0.5732 & 0.0128 & [0.5552, 0.5915] \\
 & HIP-LLM    & [0.5713, 0.5751] & [0.0109, 0.0147] & [0.5528, 0.5937] \\
\midrule
\multirow{4}{*}{$\text{OP}^{\text{approx}}$}
 & BB-UnInf  & 0.5740 & 0.0120 & [0.5460, 0.6009] \\
 & BB-Inf & 0.5756 & 0.0104 & [0.5509, 0.5995] \\
 & HiBayES    & 0.5712 & 0.0148 & [0.5438, 0.5966] \\
 & HIP-LLM    & [0.5698, 0.5774] & [0.0086, 0.0162] & [0.5421, 0.6048] \\
\midrule
\multirow{4}{*}{$\text{OP}^\text{GT}$}
 & BB-UnInf  & 0.5782 & 0.0078 & [0.5486, 0.6062] \\
 & BB-Inf & 0.5796 & 0.0064 & [0.5538, 0.6042] \\
 & HiBayES    & 0.5750 & 0.0110 & [0.5459, 0.6011] \\
 & HIP-LLM    & [0.5738, 0.5816] & [0.0044, 0.0122] & [0.5446, 0.6100] \\
\bottomrule
\end{tabular}
\end{table}

HIP-LLM behaves differently in the Small-$N$ regime by reporting intervals rather
than single reliability values. For example, it produces ranges of plausible
medians such as $[0.5713,\,0.5751]$ under the dataset-based OP and
$[0.5738,\,0.5816]$ under the ground-truth OP. These wider bounds reflect
epistemic uncertainty about prior assumptions and make explicit that multiple
reliability estimates are consistent with the data. Hierarchical modeling through
HiBayES also smooths estimates across subdomains, but can introduce additional
bias when data are scarce.

In the Large-$N$ regime, the effect of priors and modeling choices largely
disappears. All methods produce nearly identical posterior medians across OP
scenarios, and uncertainty intervals become very narrow. For HIP-LLM, the median
interval collapses to a tight range (e.g., $[0.5860,\,0.5866]$ under the
dataset-based OP), showing that prior uncertainty no longer plays a significant
role. 

When comparing methods, HIP-LLM intervals should be interpreted as
robustness bounds rather than standard credible intervals; for conservative
decision making, the lower bound of the HIP-LLM median interval provides a
worst-case reliability estimate consistent with the available data.

\begin{table}[t]
\centering
\caption{Posterior estimates under the Large-N ($[1000,\,5000,\,10000,\,3000]$) regime
(\(p_L^{\text{GT}} = 0.5860\)).}
\label{tab_largeN_single}
\begin{tabular}{llccc}
\toprule
OP & Method &  \makecell{Median \\ value/ bound} & \makecell{Error \\ value/bound} & 90\% Interval  \\
\midrule
\multirow{4}{*}{$\text{OP}^{\text{data}}$}
 & BB-UnInf  & 0.5907 & 0.0047 & [0.5849, 0.5965] \\
 & BB-Inf & 0.5908 & 0.0048 & [0.5849, 0.5965] \\
 & HiBayES    & 0.5909 & 0.0049 & [0.5846, 0.5969] \\
 & HIP-LLM    & [0.5860, 0.5866] & [0.0000, 0.0006] & [0.5846, 0.5968] \\
\midrule
\multirow{4}{*}{$\text{OP}^{\text{approx}}$}
 & BB-UnInf  & 0.5887 & 0.0027 & [0.5801, 0.5973] \\
 & BB-Inf & 0.5885 & 0.0025 & [0.5799, 0.5970] \\
 & HiBayES    & 0.5887 & 0.0027 & [0.5800, 0.5980] \\
 & HIP-LLM    & [0.5752, 0.5764] & [0.0096, 0.0108] & [0.5794, 0.5978] \\
\midrule
\multirow{4}{*}{$\text{OP}^\text{GT}$}
 & BB-UnInf  & 0.5928 & 0.0068 & [0.5836, 0.6017] \\
 & BB-Inf & 0.5926 & 0.0066 & [0.5837, 0.6015] \\
 & HiBayES    & 0.5929 & 0.0069 & [0.5838, 0.6025] \\
 & HIP-LLM    & [0.5784, 0.5797] & [0.0063, 0.0076] & [0.5832, 0.6024] \\
\bottomrule
\end{tabular}
\end{table}

\subsubsection{\textbf{RQ6} (Sensitivity to the definition of success/failure)}
\label{sec_RQ6}

A key modeling assumption in HIP-LLM is the use of binary success/failure labels at the subdomain level, summarized by the observed success counts $C_{ij}$ and total trials $N_{ij}$. In practice, however, the notion of ``success'' is not unique and may depend on the chosen evaluation criterion.

To assess the robustness of HIP-LLM to alternative success definitions, we conduct a controlled experiment on a single dataset (MBPP) and a single model (Claude Sonnet~4.5). We compare two common criteria: \emph{Pass@1}, where only the first generated solution is evaluated, and \emph{Pass@3}, where a task is considered successful if at least one out of three independent generations passes all tests. All other experimental conditions are held fixed.

Figure~\ref{fig_pass1_pass3} shows the subdomain-level posterior CDF envelopes for the true success probability 
$\frac{C}{N}$ under Pass@1 and Pass@3. As expected, the more permissive Pass@3 criterion yields a higher observed success count and a right-shifted posterior envelope. However, the two envelopes exhibit substantial overlap, indicating that the inferred reliability is only moderately sensitive to the choice of success definition.

\begin{figure}
    \centering
    \includegraphics[width=0.68\linewidth]{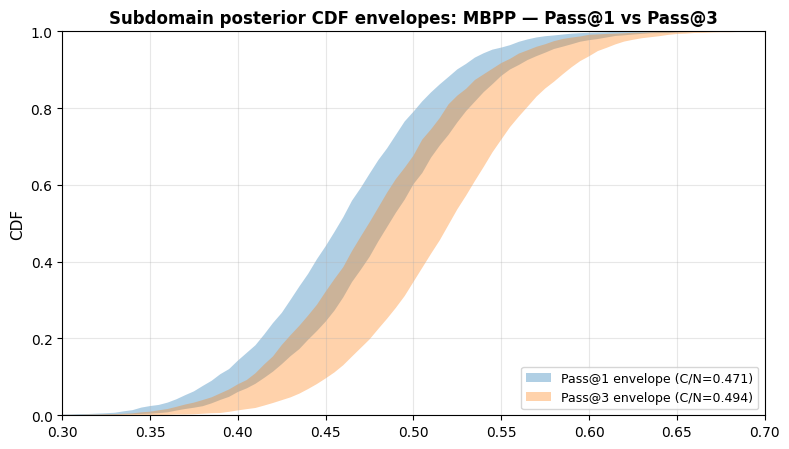}
    \caption{Subdomain-level posterior CDF envelopes of non-failure probability for MBPP under alternative success definitions.
Results are shown for Claude Sonnet~4.5 on MBPP ($N=257$), using identical tasks, prompts, and tests.
Pass@1 yields 
$\frac{C}{N}=0.471$
, while Pass@3 yields $\frac{C}{N}=0.494$.
Envelopes reflect imprecise hierarchical inference with identical hyper-hyperparameter intervals; only the success definition differs. A question-oriented interpretation of figure: How does the inferred subdomain-level non-failure probability change when alternative success (failure) definitions, such as Pass@1 versus Pass@3, are used on the same evaluation data?}
    \label{fig_pass1_pass3}
\end{figure}

\subsubsection{\textbf{RQ7} (Robustness to Memory Effects)}
\label{sec_RQ7}
Our formal reliability definition models task outcomes as i.i.d.\ Bernoulli trials under a fixed OP.
In practice, however, LLM usage may violate independence because the model can condition on previously observed conversational context (i.e., \emph{memory}). To assess the practical impact of such dependence on HIP-LLM, we conducted two complementary experiments: (i) an empirical evaluation to quantify memory growth when memory is not reset after each task and its relation to dependence strength; and (ii) a targeted sensitivity analysis that injects memory-induced dependence into one of subdomains ($\text{Subdom}_{21}$-BoolQ) and examines how the resulting dependence shifts posterior envelopes across hierarchy levels.

We evaluated one subdomain ($\text{Subdom}_{21}$-BoolQ)\footnote{We conducted this experiment on 300 tasks from the BoolQ dataset. To preserve the OP proportions across subdomains, we kept the same relative weights for the other datasets as in the main evaluation.} in a single continuous session in which the entire conversation history was retained in every API call. For each query, we recorded the API-reported 
memory.
In this setting, memory grows monotonically with the number of questions and provides an operational proxy for the strength of cross-item dependence introduced by persistent context (Fig.~\ref{fig_Memory_relation}). After 299 questions, the retained context reached 
$\approx 2.06\times 10^5$ bytes, i.e., $\approx 201$ KB, illustrating that long-context evaluation can induce non-negligible memory accumulation even within a single benchmark run.

Using the above proxy, we observed that the dependence strength induced by memory increases approximately linearly with the retained context size. This motivates a parsimonious sensitivity model in which the dependence parameter is a linear function of memory. 
The linear mapping is used only as a stress-test mechanism to quantify how departures from independence could propagate into the inferred reliability distributions; it does not change the definition of reliability itself, but probes the robustness of posterior conclusions when the i.i.d.\ approximation is imperfect.

We then injected memory-induced dependence into BoolQ according to the linear dependence model and compared posterior CDF envelopes under three BoolQ $\theta$ settings,
\[
\theta_{\text{BoolQ}} \in \{0.915,\ 0.940,\ 0.945\},
\]
where $\theta_{\text{BoolQ}}=0.915$ corresponds to the \emph{i.i.d. assumption} (no memory-induced dependence), while $\theta_{\text{BoolQ}}=0.940$ and $0.945$ represent non-i.i.d. regimes induced by increasing retained context. The remaining subdomain accuracies (MBPP, DS-1000, RACE-H) were kept fixed, and the same hierarchical structure and operational-profile aggregation were preserved (subdomain $\rightarrow$ domain via $\Omega_{ij}$, and domain $\rightarrow$ LLM via $W_i$). For each $\theta_{\text{BoolQ}}$, we propagated uncertainty through the hierarchy under the same imprecise hyperparameter intervals and computed CDF envelopes at three levels:
(i) BoolQ subdomain,
(ii) Reasoning domain ($D2$),
and (iii) overall LLM level.

 \begin{figure}
    \centering
    \begin{subfigure}[t]{0.42\textwidth}
        \centering
        \includegraphics[width=\linewidth]{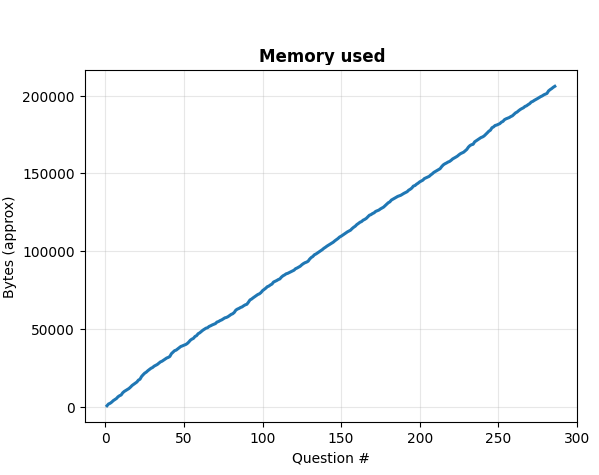}
        \caption{Growth of retained conversational context as a function of the number of BoolQ tasks evaluated in a single continuous session. Memory size (in bytes) increases monotonically and serves as an operational proxy for the strength of memory-induced dependence.}
        \label{fig_Memory_relation}
    \end{subfigure}
    \hfill
    \begin{subfigure}[t]{0.495\textwidth}
        \centering
        \includegraphics[width=\linewidth]{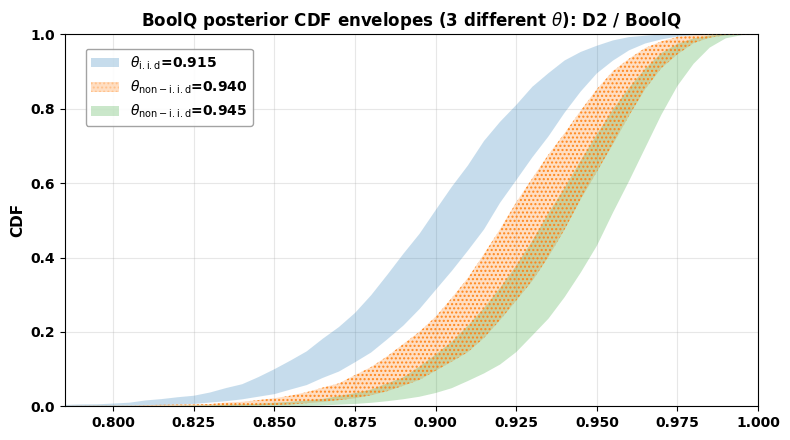}
        \caption{Posterior CDF envelopes of non-failure probability for the BoolQ subdomain under three accuracy settings: $\theta_{\text{BoolQ}}=0.915$ (i.i.d.), and $\theta_{\text{BoolQ}}=0.940, 0.945$ (non-i.i.d.). Increasing memory-induced dependence produces smooth rightward shifts of the envelope.}
        \label{fig_non-iid_CDF_subdomain}
    \end{subfigure}
    \hfill
    \begin{subfigure}[t]{0.49\textwidth}
        \centering
        \includegraphics[width=\linewidth]{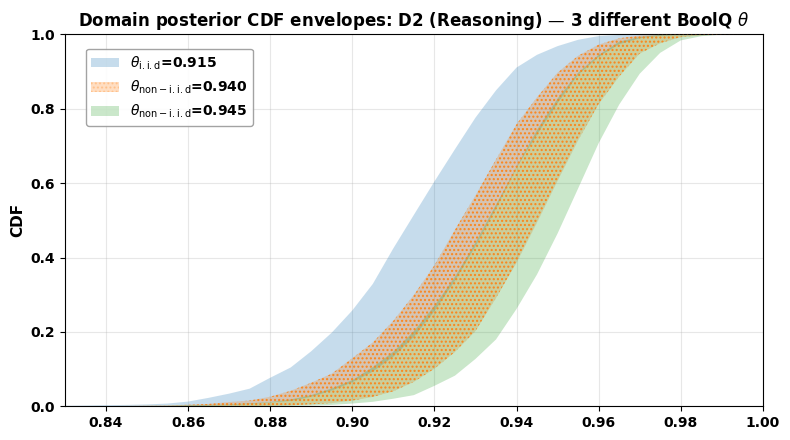}
        \caption{Posterior CDF envelopes of non-failure probability at the Reasoning domain level ($\text{Dom}2$) resulting from propagation of BoolQ dependence through subdomain aggregation with weights $\Omega_{ij}$.}
        \label{fig_non-iid_CDF_domain}
    \end{subfigure}
    \hfill
    \begin{subfigure}[t]{0.49\textwidth}
        \centering
        \includegraphics[width=\linewidth]{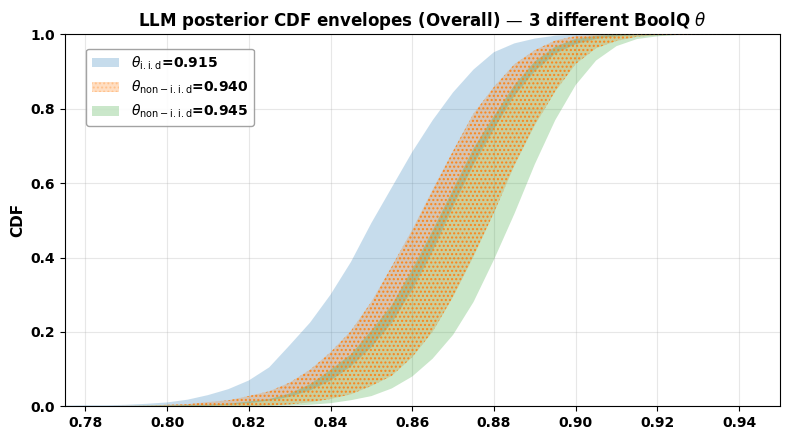}
        \caption{Posterior CDF envelopes of non-failure probability at the LLM level after domain aggregation with operational-profile weights $W_i$, illustrating the attenuated but consistent impact of subdomain-level dependence.}
        \label{fig_non-iid_CDF_LLM}
    \end{subfigure}
    \caption{Sensitivity of posterior CDF envelopes of non-failure probability to memory-induced non-i.i.d.\ dependence. The top-left panel (\ref{fig_Memory_relation}) quantifies memory accumulation during same-session evaluation, motivating a linear dependence proxy. The remaining panels show how increasing dependence in the BoolQ subdomain propagates through the hierarchical reliability model, yielding smooth and controlled shifts of envelopes at subdomain, domain, and LLM levels.
    }
    \label{fig_non-iid_CDF}
\end{figure}

Across all three levels, the envelopes shift smoothly with $\theta_{\text{BoolQ}}$: increasing $\theta_{\text{BoolQ}}$ yields a consistent rightward shift (stochastically larger reliability), while decreasing $\theta_{\text{BoolQ}}$ shifts the envelopes leftward.

These experiments support two practical conclusions: (1) persistent context can induce measurable dependence during benchmark-style evaluation, and the effect scales with memory; and (2) within a realistic range around the observed reliability levels, our posterior envelope conclusions remain stable, with dependence primarily producing controlled shifts rather than qualitative changes in domain- and LLM-level reliability characterization.

\subsubsection{\textbf{RQ8} (Scalability)}
\label{sec_RQ8}
This section characterizes the computational cost of the proposed HIP-LLM inference pipeline. The analysis is performed at the level of the algorithmic procedure (rather than for any single experimental figure) and is structured to (i) derive theoretical time and memory scaling laws from the inference steps, and (ii) empirically validate these scaling laws by controlled parameter sweeps.

Let $m$ denote the number of domains and $\bar{n}$ the average number of subdomains per domain\footnote{Although different domains may have different numbers of subdomains, the computational cost is mainly determined by the typical workload of a domain. For this reason, we describe the scaling in terms of the average number of subdomains $\bar{n}$. This keeps the formulas simple, reflects the observed runtime in practice, and avoids overly pessimistic worst-case estimates.}. For each domain $i$, HIP-LLM represents hyperparameters on a discretized grid over $(\mu_i,\nu_i)$ of size $G=n_\mu n_\nu$. Imprecision at the hyperparameter level is handled by evaluating a finite set of $K$ hyperparameter configurations (or candidates) over the grid. To propagate uncertainty from subdomain to domain and from domain to system level, the method uses Monte Carlo sampling with $S$ samples per configuration. Posterior envelopes are reported on a fixed evaluation grid (e.g., a CDF grid) of size $T$.

The inference is decomposed into three hierarchical stages:
\begin{enumerate}
    \item \textbf{Subdomain level:} For each domain and each hyperparameter configuration, HIP-LLM evaluates subdomain posteriors conditional on $(\mu_i,\nu_i)$ on the grid and constructs the corresponding imprecise posterior set.
    \item \textbf{Domain level:} For each configuration, HIP-LLM propagates subdomain uncertainty to the domain level via Monte Carlo sampling by drawing subdomain reliability variables and aggregating them using OP weights $\Omega_{ij}$.
    \item \textbf{LLM level:} HIP-LLM aggregates domain-level reliability to the system level using domain weights $W_i$, while preserving imprecision by computing lower/upper envelopes across the explored hyperparameter configurations.
\end{enumerate}

Theoretical time complexity is obtained by counting the dominant operations required by the above stages as functions of $(m,\bar{n},G,K,S,T)$. Subdomain-level cost is driven by grid-based hyperposterior evaluation and any intra-domain coupling required to compute subdomain posteriors under shared hyperparameters. Domain-level cost is driven by Monte Carlo propagation, which scales with the number of subdomains drawn per sample and the number of samples. LLM-level cost is driven by repeating aggregation across hyperparameter configurations and computing envelope statistics on the evaluation grid. Memory complexity is characterized by the storage required for (i) grid-based posterior arrays per domain/configuration, (ii) Monte Carlo samples (if stored) or streaming statistics (if not stored), and (iii) envelope representations on the evaluation grid.

To validate the theoretical scaling, we use controlled experiments in which all numerical parameters are fixed to baseline values and then varied one at a time. Specifically, we run repeated executions while sweeping $m$, $\bar{n}$, $K$, $S$, and $G$ independently, and record wall-clock time and peak memory usage for each run. This isolates the marginal effect of each parameter and enables direct comparison with the predicted scaling laws. To quantify scaling exponents, we fit a power-law model of the form $\text{time}=c\,x^\alpha$ (and analogously for memory), where $x$ is the swept parameter\footnote{The parameter we change (e.g. number of subdomains, samples, domains, etc.)} and $\alpha$\footnote{The parameter that shows how fast time grows when 
$x$ increases. for instance, $\alpha=1$ corresponds to linear growth, while $\alpha=2$ corresponds to quadratic growth, and so on.} is the inferred scaling exponent. Here $c$ is a constant which sets the scale.

Many steps of HIP-LLM can be computed independently across domains and hyperparameter configurations, making the method naturally parallelizable. Parallel execution reduces wall-clock runtime in practice, while the underlying asymptotic computational complexity remains unchanged.

The following experiment empirically validates the theoretical computational complexity analysis of the HIP-LLM inference pipeline. We report controlled scalability experiments that measure wall-clock runtime and peak memory usage while varying key algorithmic parameters one at a time. The results are used to confirm the predicted asymptotic scaling laws\footnote{Do the measured runtimes grow the way the theory says they should?} and to identify the dominant computational stages in practice.

\paragraph{Scalability plots} (Fig.~\ref{fig_scaling_G}- Fig.~\ref{fig_scaling_s} ) showing runtime as a function of each swept parameter ($m$, $\bar n$, $K$, $S$, and $G$), together with fitted scaling laws. Across all sweeps, the empirical results closely match the theoretical scaling predictions and confirm that subdomain-level inference dominates the computational cost of HIP-LLM in practice.

     \begin{figure}
    \centering
    \begin{subfigure}[t]{0.49\textwidth}
        \centering
        \includegraphics[width=\linewidth]{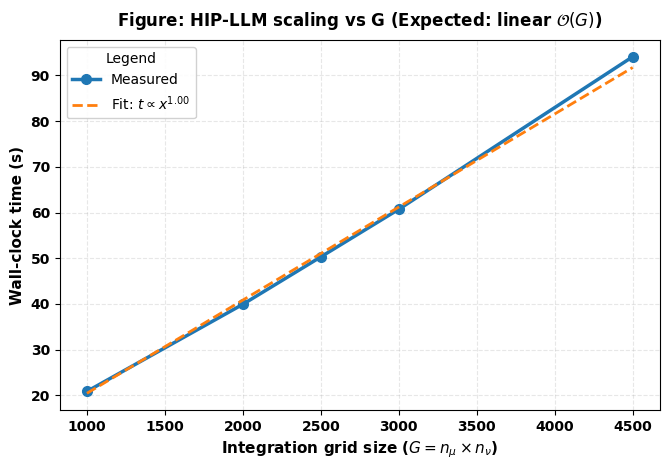}
        \caption{Scaling with the $(\mu,\nu)$ integration grid size $G=n_\mu n_\nu$, showing linear growth in wall-clock time, consistent with the theoretical $\mathcal{O}(G)$ complexity.}

        \label{fig_scaling_G}
    \end{subfigure}
    \hfill
    \begin{subfigure}[t]{0.49\textwidth}
        \centering
        \includegraphics[width=\linewidth]{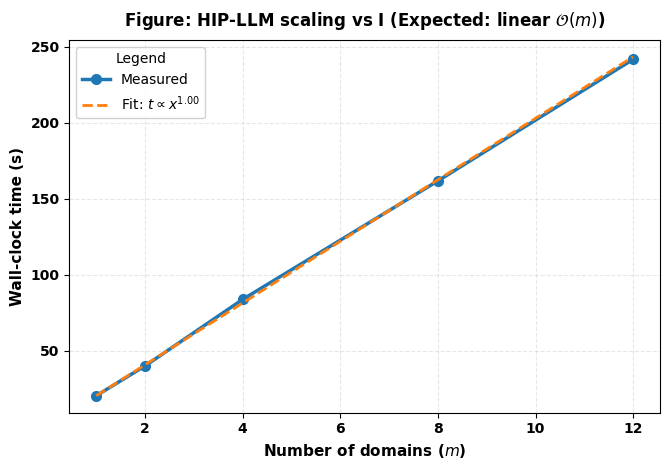}
        \caption{Scaling with the number of domains $m$, demonstrating linear runtime growth due to independent per-domain inference, in agreement with the expected $\mathcal{O}(m)$ behavior.}
        \label{fig_scaling_I}
    \end{subfigure}
     \hfill
    \begin{subfigure}[t]{0.49\textwidth}
        \centering
        \includegraphics[width=\linewidth]{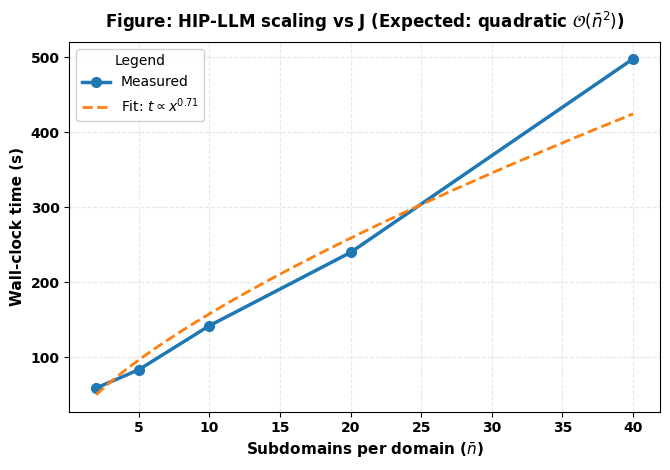}
        \caption{Scaling with the average number of subdomains per domain $\bar n$, exhibiting sub-quadratic growth over the tested range, reflecting practical efficiencies relative to the worst-case $\mathcal{O}(\bar n^2)$ bound.}

        \label{fig_scaling_J}
    \end{subfigure}
     \hfill
    \begin{subfigure}[t]{0.49\textwidth}
        \centering
        \includegraphics[width=\linewidth]{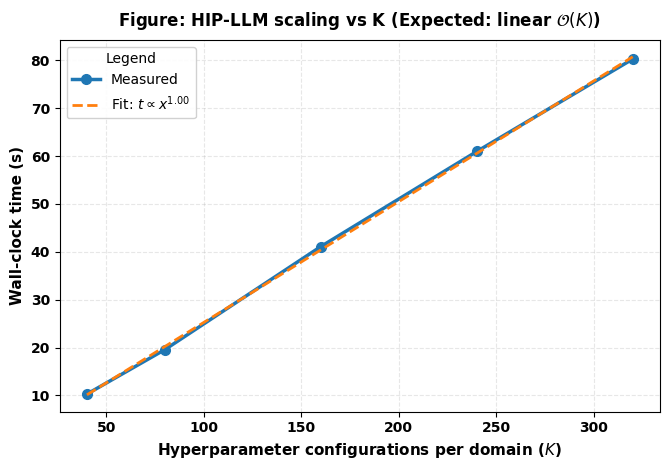}
        \caption{Scaling with the number of hyperparameter configurations per domain $K$, confirming linear complexity due to independent evaluation across configurations.}

        \label{fig_scaling_K}
    \end{subfigure}
     \hfill
    \begin{subfigure}[t]{0.49\textwidth}
        \centering
        \includegraphics[width=\linewidth]{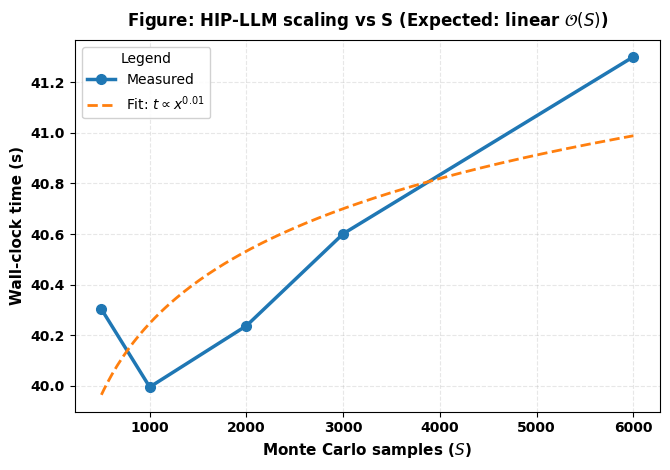}
        \caption{Scaling with the number of Monte Carlo samples $S$, showing weak dependence on $S$ and indicating that Monte Carlo propagation is not the dominant computational cost in the tested regime.}

        \label{fig_scaling_s}
    \end{subfigure}
     \begin{subfigure}[t]{0.49\textwidth}
        \centering
        \includegraphics[width=\linewidth]{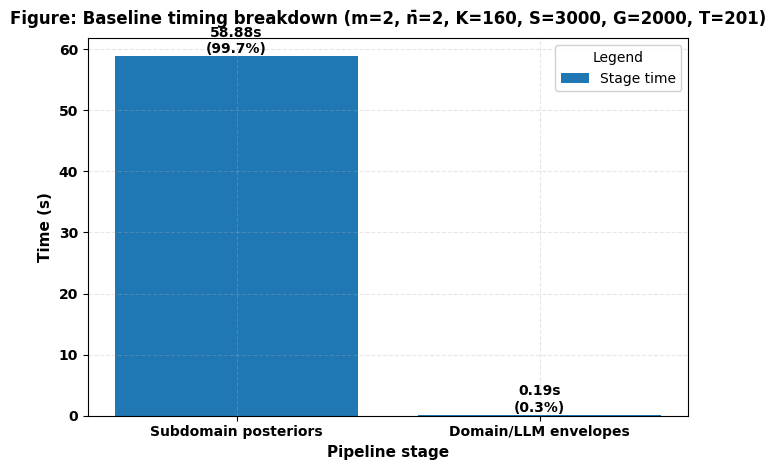}
        \caption{Baseline timing breakdown for the full HIP-LLM pipeline, showing that subdomain posterior computation dominates the total runtime, while domain- and LLM-level envelope aggregation contributes a negligible fraction.}

        \label{fig_base_time_break}
    \end{subfigure}
    \caption{Empirical scalability of the HIP-LLM inference pipeline under controlled parameter sweeps. Each subfigure varies a single parameter while keeping all others fixed at the baseline configuration ($m=2$, $\bar n=2$, $K=160$, $S=3000$, $G=2000$, $T=201$, $K_{\text{total}}\le512$). Measured runtimes are compared against power-law fits to validate theoretical complexity predictions and to identify dominant computational stages.}
    \label{fig_scaling}
\end{figure} 

\paragraph{Baseline timing breakdown} (Fig.~\ref{fig_base_time_break}) reporting the relative contribution of each inference stage under the reference configuration. The results show that subdomain-level posterior computation dominates runtime, accounting for over $99\%$ of total execution time, while domain- and LLM-level aggregation contributes a negligible fraction. This explains the weak empirical dependence on $S$ and highlights the subdomain inference stage as the primary computational bottleneck.

\paragraph{Memory Complexity}
According to Tab.~\ref{tab_memory_scaling} empirical results confirm the expected memory behavior of HIP-LLM. 
Peak memory usage scales approximately linearly with the number of domains $m$, 
the average number of subdomains $\bar n$, the number of hyperparameter configurations $K$, 
and the number of Monte Carlo samples $S$, due to the storage of posterior quantities and sampled reliabilities. 
In contrast, memory usage is effectively independent of the integration grid size $G$, 
indicating that the $(\mu,\nu)$ grid is reused across configurations rather than stored per sample.

\begin{table}[t]
\centering
\caption{Peak memory usage under controlled scalability experiments.
Each row reports the measured peak memory while sweeping one parameter and
keeping all others fixed at the baseline configuration.}
\label{tab_memory_scaling}
\begin{tabular}{lcc}
\toprule
Swept parameter & Range & Peak memory (MB) \\
\midrule
Number of domains $m$ & $1 \rightarrow 12$ & $24.6 \rightarrow 146.7$ \\
Subdomains per domain $\bar n$ & $2 \rightarrow 40$ & $35.8 \rightarrow 314.2$ \\
Monte Carlo samples $S$ & $500 \rightarrow 6000$ & $7.6 \rightarrow 69.4$ \\
Hyperparameter configs $K$ & $40 \rightarrow 320$ & $19.1 \rightarrow 57.8$ \\
Grid size $G$ & $1000 \rightarrow 4500$ & $\approx 35.7$ (constant) \\
\bottomrule
\end{tabular}
\end{table}

}

\section{Discussion}
\label{sec_discussion}

In this section, we discuss the assumptions and limitations of the proposed HIP-LLM.

\paragraph{Assumptions on hierarchical dependencies of failure probabilities}
Inspired by HiBayES \cite{luettgau_hibayes_2025}, we introduce a hierarchical dependency structure where dependencies are defined at the level of failure probabilities across related task groups. Specifically, high-level domains (e.g., coding vs. law) are modeled as independent because their underlying competencies and training data sources are substantially distinct. In contrast, subdomains within a domain (e.g., coding in C++ vs. coding in Python) are modeled as having dependent failure probabilities, reflecting shared latent skills, representations, or reasoning patterns that influence performance across related tasks. This abstraction allows tractable modeling of correlated reliability while acknowledging that the true dependencies among knowledge areas are more complex and fluid. While we do not prescribe specific taxonomies of domain/sub-domains for LLMs in this paper, our formulation is intended to be without loss of generality: users of HIP-LLM can redefine their own ``domains'' to reflect their own assumptions about dependence. For example, if users believe that two domains (say, A and B) are not fully independent, they may treat them as two subdomains under a newly introduced higher-level synthetic domain. Such flexibility allows HIP-LLM to adapt to different operational interpretations or empirical evidence of dependency, reflecting the assessor's domain knowledge.

Note, such (in)dependencies at the parameter level (i.e., correlations or independence between the failure probabilities of domains or subdomains) are different from task-level outcome dependencies, which we discuss next.

\paragraph{Assumptions of i.i.d. Bernoulli Trails for Task Outcomes} 
It is important to distinguish the assumption of independent task outcomes from the (in)dependencies among failure probabilities discussed earlier. The latter concern parameter-level relationships, i.e., how the long-run failure rates across domains or subdomains may co-vary due to shared competencies, whereas the former concerns instance-level independence of individual task results, given those parameters. Even if the failure probabilities of two sub-domains are correlated random variables, we can still assume that, conditional on those probabilities, each task outcome within a sub-domain is an independent Bernoulli trial.

As discussed earlier in Def.~\ref{def_formal_llm_reliability} and Remark~\ref{rm_iid_vs_context}, the i.i.d. assumption is pragmatically justified in the context of LLM evaluations when tasks are executed in isolation, such as \textit{initiating new chat sessions without contextual memory}. \rv{Nevertheless, caution is warranted: in practice, the way tasks are designed, sequenced, or batched can introduce subtle dependencies that violate this idealization, e.g., in long-context, agentic workflows. Assessors should therefore ensure testing procedures (e.g., resetting each task session) align with the independence assumption, a common setup in offline benchmarking.}

Indeed, the i.i.d. Bernoulli trial assumption remains a simplification, motivated by mathematical convenience and interpretability. \rv{However, this assumption is not introduced uniquely by our HIP-LLM, rather that models with the i.i.d. assumption have a long history of being used in software reliability assessment. Thayer \emph{et al.}'s~\cite{ThayerEtAl_1975} model was one of the earliest, used in early works on random testing \cite{DuranEtal_1984_EvaluationofRandomTesting}. Regulatory bodies also have recommended using the i.i.d. assumption in reliability assessments when appropriate~\cite{atwood2003handbook}. Having said that, weakening the i.i.d. assumption has been a long-standing research topic in the software reliability community. Numerous models have been proposed to capture dependence among successive tests, including a binary Markov chain model~\cite{chen_binary_1996}, a Markov renewal process~\cite{goseva_popstojanova_failure_2000,Xie_2005_ModellingCorrelatedFailures}, the Markov-based model with benign-failure states~\cite{bondavalli_dependability_1995}, and Bayesian models examining how sensitive reliability claims are to potential dependencies of test outcomes \cite{salako2023unnecessity,salako2024demonstrating}.}

\rv{A road-map for extending HIP-LLM beyond i.i.d. task outcomes can build directly on aforementioned, established software reliability models. More concretely, a first step is treating a multi-turn interaction as the unit of analysis and allowing within-session dependence while modeling across-session dynamics via latent states. One can introduce Markov/hidden-state structures, e.g., hierarchical hidden Markov model or Markov renewal formulations where latent ``interaction states'' evolve over time and govern task-level failure probabilities.}


\paragraph{What constitutes a failure for LLMs}
Reliability, by definition, concerns ``failures''. In traditional software systems, failures are relatively straightforward to define, as they can be explicitly linked to violations of formal specifications or requirements. However, for AI systems, and particularly for LLMs, no such explicit specification exists. As a result, a complete and formal characterization of what constitutes a failure for LLMs remains an open research challenge. The notion of ``failure'' can vary across applications and stakeholders. For instance, factual errors, inconsistency\footnote{Most existing LLM evaluation studies focus on one-shot accuracy metrics (e.g., Pass@1), whereas real-world usage often involves iterative prompting. Inconsistencies across multiple outputs can undermine user trust and thus regarded as a failure.}, harmful or biased content, or divergence from human judgment may all be considered failures depending on the context \cite{chang_surveyLLM_2024}. 

\rv{The choice of scoring criteria for defining LLM failures is inherently application dependent and should align with domain-specific specifications (if any) and evaluation norms. For example, in code generation, failures are typically defined in terms of functional correctness, such as compilation success or passing predefined test cases. In mathematical problem solving, correctness is often determined by agreement with a unique ground-truth answer or a formally verifiable solution. In contrast, in legal analysis, failure definitions may rely on expert judgment regarding factual accuracy, consistency with legal principles, or omission of critical considerations. For more open-ended tasks such as email drafting or creative writing, failures are often defined relative to stylistic, pragmatic, or preference-based criteria, frequently assessed through human evaluation. These examples illustrate that failure definitions for LLMs are not universal but must be specified in accordance with the intended use and domain requirements, reinforcing the need to interpret reliability estimates as conditional on the chosen scoring criteria.}

That said, in this work, our focus is not on prescribing a specific definition of LLM failure but on \textit{providing a reliability assessment modeling framework that is agnostic to failure definitions}. \rv{For example, the reliability definition in Def.~\ref{def_formal_llm_reliability} can be readily extended from binary failures to non-binary scores by replacing the indicator function with a continuous scoring function (e.g., taking values in $[0,1]$ through probabilistic or multi-rater failure models), thereby capturing scoring noise or subjectivity.}
HIP-LLM can accommodate any well-defined notion of failure, enabling users to instantiate the model with the criteria most relevant to their domain or evaluation objectives.

\paragraph{How can OPs be specified in practice}

\rv{OPs are fundamental to software reliability assessment, and the same principle applies to LLM reliability. HIP-LLM provides explicit hierarchical ``interfaces'' for integrating OPs; however, in the absence of access to real operational data, this study relies on simulated OPs derived from existing benchmark datasets. Specifically, benchmark datasets are treated as sampling frames that approximate the unknown operational data distribution (cf.~Remark~\ref{rm_OP_simulation}), with task-level distributions assumed within each subdomain and operational weights assigned across subdomains and domains. This experimental choice is made solely to demonstrate how HIP-LLM incorporates OPs and to enable fair comparison with existing benchmark-based methods, and \textit{should not be interpreted as a solution to OP acquisition}. Accurately estimating OPs from real usage data is beyond the scope of this reliability modeling study and is a well-recognized, separate research problem in software reliability engineering, for which numerous methodologies have been proposed (e.g., \cite{musa1993operational,musa_adjusting_1994,lyu_handbook_1996}). The concept of OP has also been discussed and measured in AI software contexts~\cite{dong2023reliability,huang2023hierarchical,zhao2021assessing,zhao_statistical_2025}.}

\rv{Importantly, the challenge of modeling LLM OPs does not undermine the proposed framework; rather, it highlights why OPs are essential to reliability assessment. In established software reliability practice, OPs are progressively refined as new information becomes available (e.g., execution logs, usage telemetry, or expert-informed adjustments) without requiring changes to the underlying reliability model~\cite{musa1993operational,musa_adjusting_1994,lyu_handbook_1996,smidts2014software}. As more representative datasets, user interaction logs, or task-distribution estimates for LLMs become available~\cite{chatterji2025people}, OPs can be updated accordingly, allowing reliability estimates to be recalibrated toward actual field usage. Such separation and ``decoupling design'' of HIP-LLM between evidence (testing results) and usage assumptions (OPs) enables HIP-LLM to evolve naturally with deployment conditions, improving fidelity without \textit{redesigning} the model.}

\rv{While we acknowledge that our present treatment of OP uncertainty remains simplified (for demonstration/comparison), more complex modeling on OP uncertainties, e.g., dynamic OPs, will form our important future works. Notably, a range of reliability models in traditional software engineering explicitly account for OP uncertainty and dynamics, including Bayesian updating, adaptive OP adjustment, and black-box or time-varying usage models~\cite{bishop_deriving_2017,musa_adjusting_1994,pietrantuono2020reliability,popov2025black,popov2025dynamic}. These established techniques provide a foundation that can be retrofitted into HIP-LLM to support uncertain and dynamical OPs in LLM deployments.}

\paragraph{How to embed prior knowledge}
HIP-LLM models uncertainty at three levels: the subdomain success rate ($\theta_{ij}$), how much domains share information ($\mu_i,\nu_i$), and the top-level parameters ($a_i,b_i,c_i,d_i$) that control the importance of those priors.

As a Bayesian framework, HIP-LLM requires priors to initiate inference, and these priors may substantially influence the resulting posteriors, particularly when data are limited \cite{bishop_toward_2011,zhao2019assessing}. The challenge of eliciting meaningful priors is \textit{not unique} to HIP-LLM but is common to all Bayesian models. Unlike HiBayES \cite{luettgau_hibayes_2025}, which employs noninformative priors, HIP-LLM follows the principles of robust Bayesian analysis \cite{Berger_1994,Insua2012robust,walter_bayesian_2017} by adopting the Imprecise Probability framework. This approach enables the representation of uncertainty about prior beliefs through sets of priors rather than committing to a single distribution. Prior research has shown that Imprecise Probability can facilitate expert elicitation by supporting structured ``thought experiments'' \cite{imprecision_luck_walter_2009,zhao_probabilistic_2019,zhao2024bayesian} allowing experts to express vague, partial, or interval-valued prior beliefs. Similarly, HIP-LLM’s prior layer can be informed by such thought experiments and expert judgment. While we acknowledge that eliciting priors in practice remains challenging, HIP-LLM at least offers a practical ``what-if'' analysis tool by allowing assessors to explore how different prior assumptions might influence posterior reliability claims, for example: if my prior belief is A, then after observing evidence B, the posterior reliability would be C.

\rv{\paragraph{On the practicality of HIP-LLM}
HIP-LLM is designed to replace current LLM assessment practices that rely primarily on benchmark scores. As highlighted by the 5 gaps identified in the introduction, benchmark-based evaluations often ignore OPs, overlook dependencies between related domains, and cannot incorporate prior knowledge, limiting their relevance to real-world use. In contrast, HIP-LLM produces decision-oriented reliability assessments by integrating OPs, hierarchical domain dependencies, and epistemic uncertainty and priors. Rather than reporting isolated scores (which, strictly speaking, cannot be regarded as ``reliability'' under standardized definitions such as ANSI \cite{ansi_1991}), HIP-LLM enables rigorous reliability claims such as ``\textit{the confidence that a model will complete a specified number of future tasks without failure under a given usage profile}''. This allows stakeholders to compare LLMs within specific domains, aggregate reliability across domains, or assess how reliability changes across different user groups or usage patterns. Such capabilities support informed offline decision-making, including model selection, risk assessment, and regulatory certification prior to deployment. While this work focuses on offline assessment, HIP-LLM could potentially be extended to online reliability monitoring with appropriate engineering optimizations for efficiency (e.g., caching and sliding data windows), which we leave for future work.
}

\section{Conclusion}
\label{sec_conclusion}

This paper presents HIP-LLM, a hierarchical Bayesian framework with imprecise probability for evaluating the reliability of LLMs. While HIP-LLM acknowledges the general-purpose nature of contemporary LLMs, it also adheres to the standardized definition of software reliability which inherently depends on specific applications and operational contexts. We reconcile these seemingly conflicting aspects by introducing a multi-level hierarchical structure of evaluation tasks that explicitly models dependencies among subdomains, integrates information through partial pooling, and incorporates OPs at multiple levels of abstraction. Moreover, as a Bayesian inference method, HIP-LLM respects the difficulty of selecting priors. Instead of simply using non-informative priors, HIP-LLM provides the mechanism to embed imprecise prior knowledge and reports posterior envelopes (rather than single distributions). Our experiments on common LLMs demonstrate that these features enable a more nuanced and standardized reliability estimation compared to existing benchmarks and state-of-the-art methods. 

While we believe HIP-LLM represents an important first step toward a principled reliability assessment framework for LLMs (and more broadly, for multimodal generative AI) several key limitations remain (cf.~Section~\ref{sec_discussion}). These include the challenge of formally defining what constitutes a failure for an LLM, accurately estimating OPs, and systematically eliciting imprecise prior knowledge. \rv{Moreover, the applicability of HIP-LLM is limited to reset or single-task usage scenarios conforming to the i.i.d. assumption (cf.~Remark \ref{rm_iid_vs_context}); and should not be overgeneralized to long-context or agentic LLM workflows, which require dependency-aware reliability models---an important future work.
More future extensions should aim to develop methods for handling uncertain and dynamically evolving OPs, and integrating more refined mechanisms for prior elicitation.}

\section*{Acknowledgment}
SK and XZ have received funding from the European Union's EU Framework Program for Research and Innovation Europe Horizon (grant agreement No 101202457). XZ's contribution is also supported by the UK EPSRC New Investigator Award [EP/Z536568/1]. SK's contribution is supported by the UKRI Future Leaders Fellowship Grant [MR/S035176/1].

Views and opinions expressed are those of the authors only and do not necessarily reflect those of the European Union or European Research Executive Agency (REA). Neither the European Union nor the granting authority can be held responsible for them.

\bibliographystyle{elsarticle-num} 
\bibliography{ref}

\newpage

\appendix

\section{Mathematical Derivations}
\label{sec_mathematical_derivations}

This appendix contains the proofs and supplementary information for the theorems presented in the main text.

\subsection{Common Framework and Notation}
\label{sec_common_framework_appendix}

In a LLM, for a domain $D_i$ with subdomains $S_{i1},\dots,S_{i n_i}$. 
For each subdomain $j=1,\dots,n_i$:

\begin{itemize}
\item \textbf{Parameter:} $\boldsymbol{\theta}_i = (\theta_{i1}, \ldots, \theta_{in_i}) \in (0,1)^{n_i}$ (subdomain reliability)
\item \textbf{Parameter:} $n_F$ = number of trials in the future
\item \textbf{Data:} $C_i = \{(C_{ik}, N_{ik})\}_{k=1}^{n_i}$ (observed correct responses and trial counts)
\item \textbf{Hierarchical prior.}
\[
\begin{aligned}
C_{ij}\mid \theta_{ij},N_{ij} &\sim \text{Binomial}(N_{ij},\theta_{ij}),\\
\theta_{ij}\mid \mu_i,\nu_i &\stackrel{\text{iid}}{\sim} \text{Beta}(\mu_i\nu_i,\,(1-\mu_i)\nu_i),\\
\mu_i &\sim \text{Beta}(a_i,b_i),\qquad
\nu_i \sim \text{Gamma}(c_i,\text{rate}=d_i),
\end{aligned}
\]
with $\mu_i\in(0,1)$, $\nu_i>0$ (domain-level parameters), and hyper- hyperparameters $h_i=(a_i,b_i,c_i,d_i)$.

\item \textbf{Admissible hyperparameters (imprecise prior).}
\[
\mathcal A_i=
[a_i^{\min},a_i^{\max}]\times[b_i^{\min},b_i^{\max}]
\times[c_i^{\min},c_i^{\max}]\times[d_i^{\min},d_i^{\max}].
\]

\item \textbf{Common Likelihood:}
\begin{align}
L(\boldsymbol{\theta}_i) = Pr(C_i \mid \boldsymbol{\theta}_i) = \prod_{j=1}^{n_i} \binom{N_{ij}}{C_{ij}} \theta_{ij}^{C_{ij}} (1-\theta_{ij})^{N_{ij}-C_{ij}} \nonumber
\label{eq_likelihood_appendix}
\end{align}
\end{itemize}


\subsection{Proof of Theorem \ref{thm_subdomain_marginal}}
\label{sec_proof_subdomain_marginal_appendix}

\textbf{Theorem 1}
For subdomain $S_{ij}$ in domain $D_i$, let 
$C_i=\{(C_{ik},N_{ik})\}_{k=1}^{n_i}$ be the observed data. 
Let the admissible set of hyperparameters be
\[
\mathcal{A}_i \;=\;
[a_{i}^{\min},a_{i}^{\max}] \times 
[b_{i}^{\min},b_{i}^{\max}] \times 
[c_{i}^{\min},c_{i}^{\max}] \times 
[d_{i}^{\min},d_{i}^{\max}],
\]
and write $h_i=(a_i,b_i,c_i,d_i)$.
Then, \emph{for any} $h_i \in \mathcal{A}_i$, the marginal posterior density of $\theta_{ij}$ is
\[
Pr(\theta_{ij}\mid C_i,h_i) 
= \frac{f_{\mathrm{marg}}(\theta_{ij}, C_i; h_i)}{Z_{\mathrm{marg}}(h_i)},
\]
where $f_{\mathrm{marg}}$ (unnormalized posterior) and $Z_{\mathrm{marg}}$ (normalizing constant) are
\begin{align}
& \hspace{1cm}f_{\mathrm{marg}}(\theta_{ij}, C_i; h_i) 
=\nonumber \\
& \hspace{2cm} \int_0^1 \int_0^{\infty} 
\Bigg[\prod_{k \neq j} \int_0^1 d\theta_{ik}\Bigg] 
L(\boldsymbol{\theta}_i)\,
Pr(\boldsymbol{\theta}_i \mid \mu_i,\nu_i)\,
Pr(\mu_i,\nu_i\mid h_i)\,
d\mu_i\, d\nu_i \nonumber
\end{align}
\[
Z_{\mathrm{marg}}(h_i) 
= \int_0^1 \int_0^{\infty} 
Pr(C_i \mid \mu_i,\nu_i)\,Pr(\mu_i,\nu_i \mid h_i)\,
d\mu_i\, d\nu_i,
\]
with $L(\boldsymbol{\theta}_i)=Pr(C_i\mid \boldsymbol{\theta}_i)$.

The imprecise marginal posterior is characterized by the lower/upper envelopes
\[
\underline{Pr}(\theta_{ij}\mid C_i) 
= \inf_{h_i\in\mathcal{A}_i} Pr(\theta_{ij}\mid C_i,h_i),
\qquad
\overline{Pr}(\theta_{ij}\mid C_i) 
= \sup_{h_i\in\mathcal{A}_i} Pr(\theta_{ij}\mid C_i,h_i).
\]

\begin{proof}[Proof of Theorm~\ref{thm_subdomain_marginal}]
We proceed step by step.

\textbf{Step 1 — Likelihood and hierarchical prior densities.}  
The joint data likelihood factorizes over subdomains:
\begin{equation}
\label{eq_likelihood_appendix}
L(\boldsymbol{\theta_i})=Pr(C_i\mid \theta_i)=\prod_{k=1}^{n_i}\binom{N_{ik}}{C_{ik}}
\theta_{ik}^{C_{ik}}(1-\theta_{ik})^{N_{ik}-C_{ik}}.
\end{equation}

Given $(\mu_i,\nu_i)$, the prior on $\theta_i$ factorizes as
\[
Pr(\boldsymbol{\theta_i}\mid \mu_i,\nu_i)=\prod_{k=1}^{n_i}
\frac{\Gamma(\nu_i)}{\Gamma(\mu_i\nu_i)\Gamma((1-\mu_i)\nu_i)}
\,\theta_{ik}^{\mu_i\nu_i-1}(1-\theta_{ik})^{(1-\mu_i)\nu_i-1}.
\]

Given the domain-level parameters $(\mu_i,\nu_i)$, the subdomain reliabilities are conditionally independent and 
identically distributed:
\[
\theta_{ij} \mid \mu_i,\nu_i \;\stackrel{\mathrm{iid}}{\sim}\; 
\mathrm{Beta}\!\big(\mu_i \nu_i,\,(1-\mu_i)\nu_i\big),
\qquad j=1,\dots,n_i.
\]
Here $\mu_i\in(0,1)$ is the prior mean and $\nu_i>0$ is the prior strength (concentration).

For each $k$, the Beta pdf (Single-coordinate density)  is:
\[
Pr(\theta_{ik}\mid \mu_i,\nu_i)
= \frac{\Gamma(\nu_i)}{\Gamma(\mu_i\nu_i)\,\Gamma((1-\mu_i)\nu_i)}\,
\theta_{ik}^{\,\mu_i\nu_i-1}\,(1-\theta_{ik})^{\,(1-\mu_i)\nu_i-1},
\qquad \theta_{ik}\in(0,1).
\]

Since $\{\theta_{ik}\}_{k=1}^{n_i}$ are conditionally independent given $(\mu_i,\nu_i)$,

\begin{equation}
\label{eq_marginal_prior_given_mu_nu_appendix}
\begin{aligned}
Pr(\boldsymbol{\theta}_i \mid \mu_i,\nu_i)
&= \prod_{k=1}^{n_i} Pr(\theta_{ik}\mid \mu_i,\nu_i) \\
&= \prod_{k=1}^{n_i} \mathrm{Beta}\!\left(\theta_{ik}\mid \mu_i\nu_i,(1-\mu_i)\nu_i\right) \\
&= \left(\frac{\Gamma(\nu_i)}{\Gamma(\mu_i\nu_i)\,\Gamma((1-\mu_i)\nu_i)}\right)^{\!n_i}
   \prod_{k=1}^{n_i}
   \theta_{ik}^{\,\mu_i\nu_i-1}\,(1-\theta_{ik})^{\,(1-\mu_i)\nu_i-1} 
\end{aligned}
\end{equation}

Given hyperparameters $h_i=(a_i,b_i,c_i,d_i)$, we take $\mu_i$ and $\nu_i$ to be independent:
\[
Pr(\mu_i,\nu_i \mid h_i) \;=\; Pr(\mu_i \mid h_i)\,Pr(\nu_i \mid h_i).
\]

We assume:
\[
\mu_i \sim \mathrm{Beta}(a_i,b_i), \quad \mu_i\in(0,1),
\qquad
\nu_i \sim \mathrm{Gamma}(c_i,\text{rate}=d_i), \quad \nu_i>0.
\]
Explicitly,
\[
Pr(\mu_i\mid h_i)=\frac{\Gamma(a_i+b_i)}{\Gamma(a_i)\Gamma(b_i)}\,
\mu_i^{\,a_i-1}(1-\mu_i)^{\,b_i-1},
\qquad
Pr(\nu_i\mid h_i)=\frac{d_i^{\,c_i}}{\Gamma(c_i)}\,\nu_i^{\,c_i-1}\,e^{-d_i \nu_i}
\]
Thus,

\begin{equation}
\label{eq_mu_nu_given_hi_appendix}
\begin{aligned}
Pr(\mu_i,\nu_i \mid h_i)
&= \mathrm{Beta}(\mu_i\mid a_i,b_i)\;\mathrm{Gamma}(\nu_i\mid c_i,\text{rate}=d_i) \\
&= \frac{\Gamma(a_i+b_i)}{\Gamma(a_i)\,\Gamma(b_i)}\;
   \frac{d_i^{\,c_i}}{\Gamma(c_i)}\;
   \mu_i^{\,a_i-1}(1-\mu_i)^{\,b_i-1}\;
   \nu_i^{\,c_i-1}e^{-d_i \nu_i} \\
&= \frac{1}{\mathrm{B}(a_i,b_i)}\;\frac{d_i^{\,c_i}}{\Gamma(c_i)}\;
   \mu_i^{\,a_i-1}(1-\mu_i)^{\,b_i-1}\;
   \nu_i^{\,c_i-1}e^{-d_i \nu_i}
\end{aligned}
\end{equation}

\noindent\textit{where } $\mathrm{B}(a_i,b_i)=\dfrac{\Gamma(a_i)\Gamma(b_i)}{\Gamma(a_i+b_i)}$.

The hyperprior factorizes:
\[
Pr(\mu_i,\nu_i\mid h_i)
=\frac{\Gamma(a_i+b_i)}{\Gamma(a_i)\Gamma(b_i)} \mu_i^{a_i-1}(1-\mu_i)^{b_i-1}
\cdot \frac{d_i^{\,c_i}}{\Gamma(c_i)} \nu_i^{c_i-1}e^{-d_i\nu_i}.
\]

\textbf{Step 2 — Normalizing constant.}  
By definition,
\[
Z_{\mathrm{marg}}(C_i, h_i)=\int_0^1 Pr(\theta_{ij}\mid C_i,h_i)\,d\theta_{ij}
=\int_0^1\!\!\int_0^\infty Pr(C_i\mid \mu_i,\nu_i)\,Pr(\mu_i,\nu_i\mid h_i)\,d\mu_i\,d\nu_i,
\]
where
\[
Pr(C_i\mid \mu_i,\nu_i)=\int_{(0,1)^{n_i}}L(\theta_i)\,Pr(\theta_i\mid \mu_i,\nu_i)\,d\theta_i.
\]

Start from the definition:
\begin{equation}
\label{eq:def-posterior}
Pr(\theta_{ij}\mid C_i,h_i)
=\frac{f_{\mathrm{marg}}(\theta_{ij},C_i;h_i)}{Z_{\mathrm{marg}}(C_i,h_i)} \nonumber
\end{equation}
where
\begin{align}
f_{\mathrm{marg}}(\theta_{ij},C_i;h_i)
&= \nonumber\\
&\int_0^1\!\!\int_0^\infty
\left[\prod_{k\neq j}\int_0^1 d\theta_{ik}\right]\,
L(\boldsymbol{\theta}_i)\,Pr(\boldsymbol{\theta}_i\mid \mu_i,\nu_i)\,Pr(\mu_i,\nu_i\mid h_i)\,d\mu_i\,d\nu_i \nonumber
\end{align}

The posterior integrates to 1:
\begin{align}
    \int_0^1 Pr(\theta_{ij}\mid C_i,h_i)\,d\theta_{ij}=1
\;\;\Longrightarrow\;\;
\frac{1}{Z_{\mathrm{marg}}(C_i,h_i)}\int_0^1 f_{\mathrm{marg}}(\theta_{ij},C_i;h_i)\,d\theta_{ij}=1 \nonumber
\end{align}

Hence
\begin{equation}
\,Z_{\mathrm{marg}}(C_i,h_i)=\int_0^1 f_{\mathrm{marg}}(\theta_{ij},C_i;h_i)\,d\theta_{ij}\, \nonumber
\end{equation}

By expanding the integral:
\begin{align}
Z_{\mathrm{marg}}(C_i,h_i)
&=\int_0^1\!\!\int_0^1\!\!\int_0^\infty
\left[\prod_{k\neq j}\int_0^1 d\theta_{ik}\right]\,
L(\boldsymbol{\theta}_i)\,Pr(\boldsymbol{\theta}_i\mid \mu_i,\nu_i)\,Pr(\mu_i,\nu_i\mid h_i)\,
d\mu_i\,d\nu_i\,d\theta_{ij}\nonumber\\
&=\int_0^1\!\!\int_0^\infty
\left[\int_{(0,1)^{n_i}} L(\boldsymbol{\theta}_i)\,Pr(\boldsymbol{\theta}_i\mid \mu_i,\nu_i)\,d\boldsymbol{\theta}_i\right]\,
Pr(\mu_i,\nu_i\mid h_i)\,d\mu_i\,d\nu_i. \label{eq_Z-expanded_appendix}
\end{align}
By substituting Eqs.~\ref{eq_likelihood_appendix}, \ref{eq_marginal_prior_given_mu_nu_appendix}, \ref{eq_mu_nu_given_hi_appendix} into \ref{eq_Z-expanded_appendix} we have:

\begin{equation}
\label{eq_Z-marg-gamma_form_appendix}
\begin{aligned}
Z_{\mathrm{marg}}(C_i,h_i)
=& \left[\prod_{j=1}^{n_i}\binom{N_{ij}}{C_{ij}}\right]\;
   \frac{d_i^{\,c_i}}{\mathrm{B}(a_i,b_i)\,\Gamma(c_i)}
   \int_0^1\!\!\int_0^\infty
   \mu_i^{\,a_i-1}(1-\mu_i)^{\,b_i-1}\;
   \nu_i^{\,c_i-1}e^{-d_i\nu_i} \\
&\hspace{0 cm}\times
   \left(\frac{\Gamma(\nu_i)}{\Gamma(\alpha_i)\,\Gamma(\beta_i)}\right)^{\!n_i}
   \prod_{j=1}^{n_i}
   \frac{\Gamma\!\big(C_{ij}+\alpha_i\big)\;
         \Gamma\!\big(N_{ij}-C_{ij}+\beta_i\big)}
        {\Gamma\!\big(N_{ij}+\nu_i\big)}
   \; d\nu_i\, d\mu_i
\end{aligned}
\end{equation}

\textbf{Step 3 — Imprecise posterior.}
For any fixed $h_i=(a_i,b_i,c_i,d_i)\in\mathcal A_i$, 
Steps 1–2 give the precise (pointwise) posterior density
\[
Pr(\theta_{ij}\mid C_i,h_i)
=\frac{f_{\mathrm{marg}}(\theta_{ij},C_i;h_i)}{Z_{\mathrm{marg}}(C_i,h_i)},
\] 

The hyperparameter vector $h_i$ is not known exactly but only to lie
within the admissible hyperrectangle
\[
\mathcal A_i=
[a_i^{\min},a_i^{\max}]\times
[b_i^{\min},b_i^{\max}]\times
[c_i^{\min},c_i^{\max}]\times
[d_i^{\min},d_i^{\max}].
\]
Hence the family 
$\mathcal F_i=\{Pr(\theta_{ij}\mid C_i,h_i):h_i\in\mathcal A_i\}$
represents all posterior densities consistent with our prior ignorance
about $h_i$.  

Following the theory of imprecise probabilities, we define the
\emph{lower} and \emph{upper} posterior densities (envelopes) as the
pointwise infimum and supremum over this family:
\[
\underline{Pr}(\theta_{ij}\mid C_i)
=\inf_{h_i\in\mathcal A_i}Pr(\theta_{ij}\mid C_i,h_i),
\qquad
\overline{Pr}(\theta_{ij}\mid C_i)
=\sup_{h_i\in\mathcal A_i}Pr(\theta_{ij}\mid C_i,h_i).
\]
Operationally, these envelopes are obtained by evaluating the closed-form
posterior $Pr(\theta_{ij}\mid C_i,h_i)$ at the extremal corners of
$\mathcal A_i$ or through numerical optimization if the extrema occur in
the interior.  The resulting pair
$(\underline{Pr},\overline{Pr})$ bounds all admissible precise posteriors
and fully characterizes the imprecise posterior belief about
$\theta_{ij}$ given the uncertainty in $(a_i,b_i,c_i,d_i)$.

\end{proof}


\subsection{Proof of Theorem \ref{thm_domain_posterior}}
\label{sec_proof_domain_appendix}
\textbf{Theorem 2}
For domain $D_i$ with local OP weights $\Omega_{ij}$ 
(where $\sum_{j=1}^{n_i} \Omega_{ij} = 1$), 
let $p_i = \sum_{j=1}^{n_i} \Omega_{ij}\theta_{ij}$ be the domain-level non-failure probability. 
Define the admissible set of hyper-hyper-parameters
\[
\mathcal{A}_i =
[a_{i}^{\min},a_{i}^{\max}] \times 
[b_{i}^{\min},b_{i}^{\max}] \times 
[c_{i}^{\min},c_{i}^{\max}] \times 
[d_{i}^{\min},d_{i}^{\max}],
\]
and write $h_i=(a_i,b_i,c_i,d_i)$. 

Then, for any $h_i \in \mathcal{A}_i$, the posterior distribution of $p_i$ 
is characterized by its CDF:

\begin{align}
F_{p_i}(t \mid C_i, h_i)
  &= \Pr(p_i \le t \mid C_i, h_i) \nonumber\\
  & \hspace{2cm} = \int_0^1 \int_0^\infty F_{p_i}(t \mid \mu_i, \nu_i, C_i)\,
     Pr(\mu_i, \nu_i \mid C_i, h_i)\, d\mu_i\, d\nu_i \nonumber
\end{align}

where
\begin{itemize}
\item $F_{p_i}(t \mid \mu_i, \nu_i, C_i)$ is the conditional CDF of 
$p_i = \sum_{j=1}^{n_i} \Omega_{ij}\theta_{ij}$ given that 
$\theta_{ij} \mid \mu_i, \nu_i, C_i \stackrel{\text{ind.}}{\sim} 
\text{Beta}(C_{ij} + \mu_i\nu_i, N_{ij} - C_{ij} + (1-\mu_i)\nu_i)$ 
for $j=1,\ldots,n_i$,
\item $Pr(\mu_i, \nu_i \mid C_i, h_i)$ is the hyper-posterior obtained via 
Bayes' rule:
\begin{align}
    & Pr(\mu_i, \nu_i \mid C_i, h_i) = \nonumber \\
& \hspace{1cm}\frac{Pr(C_i \mid \mu_i, \nu_i) \, 
\text{Beta}(\mu_i \mid a_i, b_i) \, 
\text{Gamma}(\nu_i \mid c_i, \text{rate}=d_i)}
{\int_0^1 \int_0^\infty Pr(C_i \mid \mu, \nu) \, 
\text{Beta}(\mu \mid a_i, b_i) \, 
\text{Gamma}(\nu \mid c_i, \text{rate}=d_i) \, d\mu \, d\nu} \nonumber
\end{align}
\end{itemize}

The imprecise domain posterior is characterized by CDF envelopes:
\[
\underline{F}_{p_i}(t \mid C_i) 
= \inf_{h_i\in\mathcal{A}_i} F_{p_i}(t \mid C_i, h_i),
\qquad
\overline{F}_{p_i}(t \mid C_i) 
= \sup_{h_i\in\mathcal{A}_i} F_{p_i}(t \mid C_i, h_i).
\]

The domain reliability $p_i = \sum_{j=1}^{n_i} \Omega_{ij}\theta_{ij}$ is a weighted sum of subdomain reliabilities, whose joint posterior distribution $Pr(\boldsymbol{\theta}_i \mid C_i, h_i)$ has no tractable closed form. However, the hierarchical structure of our model provides a natural decomposition: we can express this joint posterior as a mixture over the shared hyperparameters $(\mu_i, \nu_i)$, where conditionally on these hyperparameters, the subdomains become independent with Beta posteriors. This decomposition, formalized in Lemma~\ref{lem:hierarchical_decomposition}, is the key to deriving the domain-level CDF. The proof proceeds by (1) expressing the domain CDF as an integral of the joint posterior over a weighted-sum constraint region, (2) applying the hierarchical decomposition from the lemma, (3) exchanging the order of integration via Fubini's theorem, and (4) recognizing the inner integral as a conditional CDF, yielding a tractable mixture representation.

\begin{lemma}[Hierarchical decomposition of subdomain joint posterior]
\label{lem:hierarchical_decomposition}
Under the hierarchical model with hyperparameters $h_i = (a_i, b_i, c_i, d_i)$, 
the joint posterior of subdomain reliabilities $\boldsymbol{\theta}_i = (\theta_{i1}, \ldots, \theta_{in_i})$ 
admits the decomposition
\begin{equation}
\label{eq:hierarchical_decomp}
Pr(\boldsymbol{\theta}_i \mid C_i, h_i) 
= \int_0^1 \int_0^\infty Pr(\boldsymbol{\theta}_i \mid \mu_i, \nu_i, C_i) \, 
Pr(\mu_i, \nu_i \mid C_i, h_i) \, d\mu_i \, d\nu_i,
\end{equation}
where:
\begin{enumerate}
\item The conditional distribution factorizes into independent Betas:
\[
Pr(\boldsymbol{\theta}_i \mid \mu_i, \nu_i, C_i) 
= \prod_{j=1}^{n_i} \text{Beta}(\theta_{ij} \mid C_{ij} + \mu_i\nu_i, N_{ij} - C_{ij} + (1-\mu_i)\nu_i).
\]

\item The hyper-posterior is obtained via Bayes' rule:
\begin{align}
Pr(\mu_i, \nu_i \mid C_i, h_i) 
&= \frac{Pr(C_i \mid \mu_i, \nu_i) \, Pr(\mu_i, \nu_i \mid h_i)}{Pr(C_i \mid h_i)}, \label{eq:hyper_posterior} \\
Pr(\mu_i, \nu_i \mid h_i) 
&= \text{Beta}(\mu_i \mid a_i, b_i) \times \text{Gamma}(\nu_i \mid c_i, \text{rate}=d_i), \nonumber
\end{align}
where the marginal likelihood is
\begin{equation}
\label{eq:marginal_likelihood}
Pr(C_i \mid \mu_i, \nu_i) 
= \prod_{j=1}^{n_i} \binom{N_{ij}}{C_{ij}} 
\frac{B(C_{ij} + \mu_i\nu_i, N_{ij} - C_{ij} + (1-\mu_i)\nu_i)}{B(\mu_i\nu_i, (1-\mu_i)\nu_i)},
\end{equation}
and the evidence is
\begin{equation}
\label{eq:evidence}
Pr(C_i \mid h_i) 
= \int_0^1 \int_0^\infty Pr(C_i \mid \mu, \nu) \, 
\text{Beta}(\mu \mid a_i, b_i) \, \text{Gamma}(\nu \mid c_i, \text{rate}=d_i) \, d\mu \, d\nu.
\end{equation}
\end{enumerate}
\end{lemma}

\begin{proof}[Proof of Lemma~\ref{lem:hierarchical_decomposition}]
The decomposition follows directly from the law of total probability and the hierarchical prior structure. By the tower property of conditional expectation applied to densities:
\begin{align}
    &Pr(\boldsymbol{\theta}_i \mid C_i, h_i) 
=\nonumber \\
& \hspace{1cm}\int_0^1 \int_0^\infty Pr(\boldsymbol{\theta}_i, \mu_i, \nu_i \mid C_i, h_i) \, d\mu_i \, d\nu_i
= \nonumber \\
& \hspace{1cm} \int_0^1 \int_0^\infty Pr(\boldsymbol{\theta}_i \mid \mu_i, \nu_i, C_i) \, Pr(\mu_i, \nu_i \mid C_i, h_i) \, d\mu_i \, d\nu_i \nonumber
\end{align}

The conditional independence of $\{\theta_{ij}\}$ given $(\mu_i, \nu_i, C_i)$ follows from the hierarchical prior:
\[
\theta_{ij} \mid \mu_i, \nu_i \stackrel{\text{iid}}{\sim} \text{Beta}(\mu_i\nu_i, (1-\mu_i)\nu_i), 
\quad j=1,\ldots,n_i,
\]
combined with independent Binomial likelihoods $C_{ij} \mid \theta_{ij}, N_{ij} \sim \text{Binomial}(N_{ij}, \theta_{ij})$. 
By conjugacy, the Beta-Binomial update yields:
\[
\theta_{ij} \mid \mu_i, \nu_i, C_i \sim \text{Beta}(C_{ij} + \mu_i\nu_i, N_{ij} - C_{ij} + (1-\mu_i)\nu_i).
\]

The hyper-posterior~\eqref{eq:hyper_posterior} follows from Bayes' rule. The marginal likelihood~\eqref{eq:marginal_likelihood} 
is the Beta-Binomial distribution for each subdomain, obtained by integrating the Binomial likelihood against the Beta prior. 
The evidence~\eqref{eq:evidence} is the normalization constant ensuring the hyper-posterior integrates to one.
\end{proof}

\begin{proof}[Proof of Theorem~\ref{thm_domain_posterior}]
We derive the posterior CDF of the domain reliability $p_i = \sum_{j=1}^{n_i} \Omega_{ij}\theta_{ij}$ 
using the hierarchical decomposition from Lemma~\ref{lem:hierarchical_decomposition}.

\medskip
\noindent\textbf{Step 1: CDF definition.}
For any $t \in [0,1]$, the posterior CDF of $p_i$ is
\[
F_{p_i}(t \mid C_i, h_i) 
= \Pr(p_i \leq t \mid C_i, h_i)
= \Pr\left(\sum_{j=1}^{n_i} \Omega_{ij}\theta_{ij} \leq t \,\Big|\, C_i, h_i\right).
\]
Define the constraint region
\[
\mathcal{R}_i(t) := \left\{\boldsymbol{\theta}_i \in (0,1)^{n_i} : 
\sum_{j=1}^{n_i} \Omega_{ij}\theta_{ij} \leq t\right\}.
\]
Then
\[
F_{p_i}(t \mid C_i, h_i) 
= \int_{\mathcal{R}_i(t)} Pr(\boldsymbol{\theta}_i \mid C_i, h_i) \, d\boldsymbol{\theta}_i.
\]

\medskip
\noindent\textbf{Step 2: Apply hierarchical decomposition.}
Substituting the decomposition from Lemma~\ref{lem:hierarchical_decomposition} (equation~\eqref{eq:hierarchical_decomp}):
\begin{align*}
F_{p_i}(t \mid C_i, h_i)
&= \int_{\mathcal{R}_i(t)} 
\left[\int_0^1 \int_0^\infty Pr(\boldsymbol{\theta}_i \mid \mu_i, \nu_i, C_i) \, 
Pr(\mu_i, \nu_i \mid C_i, h_i) \, d\mu_i \, d\nu_i\right] d\boldsymbol{\theta}_i.
\end{align*}

\medskip
\noindent\textbf{Step 3: Change of integration order.}
By Fubini's theorem (applicable since all integrands are non-negative and integrate to finite values):
\begin{align*}
F_{p_i}(t \mid C_i, h_i)
&= \int_0^1 \int_0^\infty 
\left[\int_{\mathcal{R}_i(t)} Pr(\boldsymbol{\theta}_i \mid \mu_i, \nu_i, C_i) \, 
d\boldsymbol{\theta}_i\right] 
Pr(\mu_i, \nu_i \mid C_i, h_i) \, d\mu_i \, d\nu_i.
\end{align*}

\medskip
\noindent\textbf{Step 4: Conditional CDF.}
The inner integral is the conditional CDF of $p_i$ given $(\mu_i, \nu_i)$ and $C_i$:
\[
F_{p_i}(t \mid \mu_i, \nu_i, C_i) 
= \int_{\mathcal{R}_i(t)} Pr(\boldsymbol{\theta}_i \mid \mu_i, \nu_i, C_i) \, 
d\boldsymbol{\theta}_i
= \Pr\left(\sum_{j=1}^{n_i} \Omega_{ij}\theta_{ij} \leq t \,\Big|\, 
\mu_i, \nu_i, C_i\right),
\]
where $\theta_{ij} \mid \mu_i, \nu_i, C_i$ are independent Beta random variables 
as specified in Lemma~\ref{lem:hierarchical_decomposition}.

\medskip
\noindent\textbf{Step 5: Mixture representation.}
Combining Steps 3 and 4:
\[
F_{p_i}(t \mid C_i, h_i) 
= \int_0^1 \int_0^\infty F_{p_i}(t \mid \mu_i, \nu_i, C_i) \, 
Pr(\mu_i, \nu_i \mid C_i, h_i) \, d\mu_i \, d\nu_i.
\]

where
\begin{itemize}
\item $F_{p_i}(t \mid \mu_i, \nu_i, C_i)$ is the conditional CDF of 
$p_i = \sum_{j=1}^{n_i} \Omega_{ij}\theta_{ij}$ given that 
$\theta_{ij} \mid \mu_i, \nu_i, C_i \stackrel{\text{ind.}}{\sim} 
\text{Beta}(C_{ij} + \mu_i\nu_i, N_{ij} - C_{ij} + (1-\mu_i)\nu_i)$ 
for $j=1,\ldots,n_i$,
\item $Pr(\mu_i, \nu_i \mid C_i, h_i)$ is the hyper-posterior obtained via 
Bayes' rule:
\begin{align}
    & Pr(\mu_i, \nu_i \mid C_i, h_i) = \nonumber \\
& \hspace{1cm}\frac{Pr(C_i \mid \mu_i, \nu_i) \, 
\text{Beta}(\mu_i \mid a_i, b_i) \, 
\text{Gamma}(\nu_i \mid c_i, \text{rate}=d_i)}
{\int_0^1 \int_0^\infty Pr(C_i \mid \mu, \nu) \, 
\text{Beta}(\mu \mid a_i, b_i) \, 
\text{Gamma}(\nu \mid c_i, \text{rate}=d_i) \, d\mu \, d\nu} \nonumber
\end{align}
\end{itemize}

This expresses the domain posterior CDF as a weighted average of conditional CDFs, 
where the weights are given by the hyper-posterior $Pr(\mu_i, \nu_i \mid C_i, h_i)$.

\medskip
\noindent\textbf{Step 6: Imprecise probability bounds.}
For each $h_i \in \mathcal{A}_i$, the mixture formula defines a valid CDF 
$F_{p_i}(\cdot \mid C_i, h_i)$. The imprecise posterior is characterized by 
pointwise envelopes over the admissible set:
\begin{align}
    \underline{F}_{p_i}(t \mid C_i) 
= \inf_{h_i \in \mathcal{A}_i} F_{p_i}(t \mid C_i, h_i),
\quad
\overline{F}_{p_i}(t \mid C_i) 
= \sup_{h_i \in \mathcal{A}_i} F_{p_i}(t \mid C_i, h_i),
\quad t \in [0,1] \nonumber
\end{align}
\end{proof}


\subsection{Proof of Theorem \ref{thm_LLM_posterior}}
\label{sec_proof_LLM_appendix}

\textbf{Theorem 3} For the LLM system with domain weights $W_i$ (where $\sum_{i=1}^{m} W_i = 1$), 
let $p_L = \sum_{i=1}^{m} W_i p_i$ be the LLM-level failure probability and 
$\text{data} = \{C_1, \ldots, C_m\}$ the observed data across all domains. 
Assume cross-domain independence.

Define the domain-level admissible sets
\[
\mathcal{A}_i = 
[a_{i}^{\min}, a_{i}^{\max}] \times 
[b_{i}^{\min}, b_{i}^{\max}] \times 
[c_{i}^{\min}, c_{i}^{\max}] \times 
[d_{i}^{\min}, d_{i}^{\max}]
\]
and write $h_i = (a_i, b_i, c_i, d_i)$ for $i = 1, \ldots, m$. 
Define the LLM-level admissible set as the Cartesian product
\[
\mathcal{A}_{\text{LLM}} = \mathcal{A}_1 \times \cdots \times \mathcal{A}_m,
\]
and collect the domain hyperparameters as 
$\mathcal{H} = (h_1, \ldots, h_m) \in \mathcal{A}_{\text{LLM}}$.

Then, for any $\mathcal{H} \in \mathcal{A}_{\text{LLM}}$, the posterior distribution of $p_L$ 
is characterized by its CDF:
\begin{align}
    &F_{p_L}(t \mid \text{data}, \mathcal{H}) = 
 Pr(p_L \leq t \mid \text{data}, \mathcal{H})
= \nonumber \\
&  \hspace{1.5cm} \int \cdots \int G\big(t \mid \{\mu_i, \nu_i\}_{i=1}^m, \text{data}\big) 
\prod_{i=1}^{m} Pr(\mu_i, \nu_i \mid C_i, h_i) \prod_{i=1}^{m} d\mu_i \, d\nu_i \nonumber
\end{align}
where
\begin{itemize}
\item $G(t \mid \{\mu_i, \nu_i\}_{i=1}^m, \text{data})$ is the conditional CDF of 
$p_L = \sum_{i=1}^{m} W_i p_i$ given all hyperparameters, defined as
\[
G\big(t \mid \{\mu_i, \nu_i\}_{i=1}^m, \text{data}\big) 
= \int_{\mathcal{R}_L(t)} \prod_{i=1}^{m} f_{p_i}(p_i \mid \mu_i, \nu_i, C_i) \, 
dp_1 \cdots dp_m,
\]
where $\mathcal{R}_L(t) := \{(p_1, \ldots, p_m) \in (0,1)^m : \sum_{i=1}^{m} W_i p_i \leq t\}$, 
and $f_{p_i}(\cdot \mid \mu_i, \nu_i, C_i)$ is the conditional density of 
$p_i = \sum_{j=1}^{n_i} \Omega_{ij}\theta_{ij}$ under 
$\theta_{ij} \mid \mu_i, \nu_i, C_i \stackrel{\text{ind.}}{\sim} 
\text{Beta}(C_{ij} + \mu_i\nu_i, N_{ij} - C_{ij} + (1-\mu_i)\nu_i)$,

\item $Pr(\mu_i, \nu_i \mid C_i, h_i)$ is the domain-level hyper-posterior for domain $i$:
\begin{align}
    & Pr(\mu_i, \nu_i \mid C_i, h_i) = \nonumber \\
& \hspace{1cm}\frac{Pr(C_i \mid \mu_i, \nu_i) \, 
\text{Beta}(\mu_i \mid a_i, b_i) \, 
\text{Gamma}(\nu_i \mid c_i, \text{rate}=d_i)}
{\int_0^1 \int_0^\infty Pr(C_i \mid \mu, \nu) \, 
\text{Beta}(\mu \mid a_i, b_i) \, 
\text{Gamma}(\nu \mid c_i, \text{rate}=d_i) \, d\mu \, d\nu} \nonumber
\end{align}

\item Cross-domain independence ensures 
$Pr(\{\mu_i, \nu_i\}_{i=1}^m \mid \text{data}, \mathcal{H}) = 
\prod_{i=1}^{m} Pr(\mu_i, \nu_i \mid C_i, h_i)$.
\end{itemize}

The imprecise LLM posterior is characterized by CDF envelopes:
\begin{align}
    &\underline{F}_{p_L}(t \mid \text{data}) 
= \inf_{\mathcal{H} \in \mathcal{A}_{\text{LLM}}} F_{p_L}(t \mid \text{data}, \mathcal{H}) \nonumber \\
&\overline{F}_{p_L}(t \mid \text{data}) 
= \sup_{\mathcal{H} \in \mathcal{A}_{\text{LLM}}} F_{p_L}(t \mid \text{data}, \mathcal{H}) \nonumber
\end{align}

\begin{proof}[Proof of Theorem~\ref{thm_LLM_posterior}]
We derive the posterior distribution of the LLM reliability 
$p_L = \sum_{i=1}^{m} W_i p_i$ by characterizing its CDF, 
building on the domain-level results from Theorem~\ref{thm_domain_posterior}.

\medskip
\noindent\textbf{Step 1: CDF definition.}
For any $t \in [0,1]$, the posterior CDF of $p_L$ is
\[
F_{p_L}(t \mid \text{data}, \mathcal{H}) 
= \Pr(p_L \leq t \mid \text{data}, \mathcal{H})
= \Pr\left(\sum_{i=1}^{m} W_i p_i \leq t \,\Big|\, \text{data}, \mathcal{H}\right).
\]
Define the region
\[
\mathcal{R}_L(t) := \left\{(p_1, \ldots, p_m) \in (0,1)^m : 
\sum_{i=1}^{m} W_i p_i \leq t\right\}.
\]
Then
\[
F_{p_L}(t \mid \text{data}, \mathcal{H}) 
= \int_{\mathcal{R}_L(t)} \prod_{i=1}^{m} Pr(p_i \mid C_i, h_i) \, dp_1 \cdots dp_m,
\]
where cross-domain independence ensures the joint density factorizes as
\[
Pr(p_1, \ldots, p_m \mid \text{data}, \mathcal{H}) 
= \prod_{i=1}^{m} Pr(p_i \mid C_i, h_i).
\]

\medskip
\noindent\textbf{Step 2: Domain-level mixture representation.}
From Theorem~\ref{thm_domain_posterior}, each domain posterior can be written as
\[
Pr(p_i \mid C_i, h_i) 
= \int_0^1 \int_0^\infty f_{p_i}(p_i \mid \mu_i, \nu_i, C_i) \, 
Pr(\mu_i, \nu_i \mid C_i, h_i) \, d\mu_i \, d\nu_i,
\]
where $f_{p_i}(\cdot \mid \mu_i, \nu_i, C_i)$ is the conditional density of 
$p_i = \sum_{j=1}^{n_i} \Omega_{ij}\theta_{ij}$ given that
\[
\theta_{ij} \mid \mu_i, \nu_i, C_i \stackrel{\text{ind.}}{\sim} 
\text{Beta}(C_{ij} + \mu_i\nu_i, N_{ij} - C_{ij} + (1-\mu_i)\nu_i), 
\quad j=1,\ldots,n_i.
\]

\medskip
\noindent\textbf{Step 3: Hierarchical decomposition across domains.}
Substituting the domain-level mixture into the LLM-level CDF:
\begin{align*}
&F_{p_L}(t \mid \text{data}, \mathcal{H})
= \nonumber \\ & \hspace{1cm}\int_{\mathcal{R}_L(t)} \prod_{i=1}^{m} 
\left[\int_0^1 \int_0^\infty f_{p_i}(p_i \mid \mu_i, \nu_i, C_i) \, 
Pr(\mu_i, \nu_i \mid C_i, h_i) \, d\mu_i \, d\nu_i\right] dp_1 \cdots dp_m \nonumber
\end{align*}

\medskip
\noindent\textbf{Step 4: Change of integration order.}
By Fubini's theorem (applicable because all integrands are non-negative and integrate to finite values), we can exchange the order of integration:
\begin{align*}
F_{p_L}(t \mid \text{data}, \mathcal{H})
&= \int_0^1 \cdots \int_0^1 \int_0^\infty \cdots \int_0^\infty 
\left[\int_{\mathcal{R}_L(t)} \prod_{i=1}^{m} f_{p_i}(p_i \mid \mu_i, \nu_i, C_i) \, 
dp_1 \cdots dp_m\right] \\
&\qquad\qquad\qquad\qquad\qquad\qquad \times 
\prod_{i=1}^{m} Pr(\mu_i, \nu_i \mid C_i, h_i) \prod_{i=1}^{m} d\mu_i \, d\nu_i.
\end{align*}

\medskip
\noindent\textbf{Step 5: Conditional CDF.}
The inner integral is the conditional CDF of $p_L$ given all hyperparameters:
\begin{align*}
G\big(t \mid \{\mu_i, \nu_i\}_{i=1}^m, \text{data}\big)
&:= \int_{\mathcal{R}_L(t)} \prod_{i=1}^{m} f_{p_i}(p_i \mid \mu_i, \nu_i, C_i) \, 
dp_1 \cdots dp_m \\
&= \Pr\left(\sum_{i=1}^{m} W_i p_i \leq t \,\Big|\, 
\{\mu_i, \nu_i\}_{i=1}^m, \text{data}\right),
\end{align*}
where $p_i \mid \mu_i, \nu_i, C_i$ are independent across domains, each distributed according to the weighted sum of independent Betas as specified in Step~2.

\medskip
\noindent\textbf{Step 6: Mixture representation.}
Combining Steps 4 and 5:
\begin{align}
    &F_{p_L}(t \mid \text{data}, \mathcal{H}) 
= \nonumber \\
& \hspace{1cm}\int \cdots \int G\big(t \mid \{\mu_i, \nu_i\}_{i=1}^m, \text{data}\big) 
\prod_{i=1}^{m} Pr(\mu_i, \nu_i \mid C_i, h_i) \prod_{i=1}^{m} d\mu_i \, d\nu_i
\end{align}
where the integrals range over $\mu_i \in (0,1)$ and $\nu_i \in (0,\infty)$ for all $i=1,\ldots,m$.

This expresses the LLM posterior CDF as a weighted average of conditional CDFs, where the weights are given by the product of independent domain-level hyper-posteriors.

\medskip
\noindent\textbf{Step 7: Imprecise probability bounds.}
For each $\mathcal{H} \in \mathcal{A}_{\text{LLM}}$, the formula above defines a valid CDF 
$F_{p_L}(\cdot \mid \text{data}, \mathcal{H})$. The imprecise posterior is characterized by pointwise envelopes:
\begin{align}
    &\underline{F}_{p_L}(t \mid \text{data}) 
=  \inf_{\mathcal{H} \in \mathcal{A}_{\text{LLM}}} F_{p_L}(t \mid \text{data}, \mathcal{H}),
\nonumber \\
&\overline{F}_{p_L}(t \mid \text{data}) 
= \sup_{\mathcal{H} \in \mathcal{A}_{\text{LLM}}} F_{p_L}(t \mid \text{data}, \mathcal{H}) \nonumber
\end{align}
where, $t \in [0,1]$
\end{proof}

\subsection{Reliability for Multiple Consecutive Operations}
\label{appendix_nF_reliabilit_appendix}

Theorems~\ref{thm_subdomain_marginal}, \ref{thm_domain_posterior}, and~\ref{thm_LLM_posterior} establish posterior distributions for the reliability parameters $\theta_{ij}$, $p_i$, and $p_L$ at the subdomain, domain, and LLM levels, respectively. These parameters represent the probability of success on a \emph{single} future operation. However, in practical reliability assessment, we often need to evaluate performance over \emph{multiple consecutive operations}.

This appendix clarifies the mathematical framework for extending the case $n_F=1$ to $n_F > 1$ (reliability with the required number of future failure-free tasks), which is formalized in Theorems~\ref{thm_subdomain_nF_dist}, \ref{thm_domain_nF_dist}, and~\ref{thm_LLM_nF_dist}. We explain:
\begin{enumerate}
\item Why reliability over $n_F$ operations takes the form $\theta_{ij}^{n_F}$ (or $p_i^{n_F}$, $p_L^{n_F}$)
\item Why we characterize these distributions through CDFs rather than closed-form densities
\item How to derive the CDF integral formulas
\item How this framework extends across all hierarchical levels
\end{enumerate}

\subsubsection{Definition: Reliability with the required \texorpdfstring{$n_F$}{nF} failure-free future tasks}

At the subdomain level, we define reliability for $n_F$ consecutive future operations as:
\begin{equation}
R_{ij}(n_F) = \theta_{ij}^{n_F} \nonumber
\end{equation}

Assume that:
\begin{itemize}
\item Each task in subdomain $S_{ij}$ succeeds independently with probability $\theta_{ij}$
\item Operations are identically distributed (i.i.d. assumption)
\item We observe $n_F$ consecutive operations
\end{itemize}

Under these assumptions, the probability that \emph{all} $n_F$ operations succeed is:
\begin{equation}
\Pr(\text{all $n_F$ operations succeed}) = \underbrace{\theta_{ij} \times \theta_{ij} \times \cdots \times \theta_{ij}}_{n_F \text{ times}} = \theta_{ij}^{n_F} \nonumber
\end{equation}

Because $\theta_{ij}$ is uncertain—it is a random variable with posterior distribution $Pr(\theta_{ij} \mid C_i, h_i)$ from Theorem~\ref{thm_subdomain_marginal}—the $n_F$-operation reliability $R_{ij}(n_F) = \theta_{ij}^{n_F}$ is also a random variable. We must therefore characterize its full posterior distribution.

\subsubsection{Why the CDF Approach}

Unlike the posteriors in Theorems~\ref{thm_subdomain_marginal}--\ref{thm_LLM_posterior}, which characterize reliability for a \emph{single} future operation, the distribution of $R_{ij}(n_F) = \theta_{ij}^{n_F}$ generally does not have a closed-form probability density function, even when $\theta_{ij}$ follows a Beta distribution. This is because:
\begin{itemize}
\item The transformation $\theta_{ij} \mapsto \theta_{ij}^{n_F}$ is nonlinear for $n_F \neq 1$
\item At higher hierarchical levels (domain, LLM), we have weighted sums of dependent random variables raised to powers, making closed forms even less tractable
\end{itemize}

Therefore, we characterize the distribution of $R_{ij}(n_F)$ through its \emph{CDF}, which has a tractable integral representation and is sufficient for all practical reliability calculations (e.g., computing probabilities, quantiles, expected values).

\subsubsection{Deriving the CDF Formula}

What does CDF answers: ``What is the probability that $n_F$-operation reliability is at most $r$?''
\begin{equation}
F_{R_{ij}(n_F)}(t \mid C_i, h_i) = \Pr(R_{ij}(n_F) \leq t \mid C_i, h_i) = \Pr(\theta_{ij}^{n_F} \leq t \mid C_i, h_i) \nonumber
\end{equation}

To compute this probability, we use the following key observation. If $\theta_{ij}^{n_F} \leq t$, then taking the $n_F$-th root of both sides:
\begin{equation}
\theta_{ij} \leq t^{1/n_F} \nonumber
\end{equation}

Since $\theta_{ij} \in (0,1)$ and $t \in [0,1]$, both sides are positive, and the function $x \mapsto x^{1/n_F}$ is monotonically increasing on $[0,1]$ for $n_F > 0$. Therefore, the inequality is preserved under this transformation.

It follows that:
\begin{equation}
\Pr(\theta_{ij}^{n_F} \leq t \mid C_i, h_i) = \Pr(\theta_{ij} \leq t^{1/n_F} \mid C_i, h_i) \nonumber
\end{equation}

The right-hand side is simply the CDF of the posterior distribution $Pr(\theta_{ij} \mid C_i, h_i)$ from Theorem~\ref{thm_subdomain_marginal}, evaluated at $r^{1/n_F}$:
\begin{equation}
\Pr(\theta_{ij} \leq t^{1/n_F} \mid C_i, h_i) = \int_0^{r^{1/n_F}} Pr(\theta_{ij} \mid C_i, h_i) \, d\theta_{ij} \nonumber
\end{equation}

This gives us the formula in Theorem~\ref{thm_subdomain_nF_dist}.

\subsubsection{Extension to Domain and LLM Levels}

The same logic extends to domain and LLM levels:
\begin{itemize}
\item \textbf{Domain level:} $R_i(n_F) = p_i^{n_F}$ where $p_i = \sum_{j=1}^{n_i} \Omega_{ij} \theta_{ij}$
\item \textbf{LLM level:} $R_L(n_F) = p_L^{n_F}$ where $p_L = \sum_{i=1}^m W_i p_i$
\end{itemize}

At these levels, closed forms are even less available because we must:
\begin{enumerate}
\item Compute the distribution of weighted sums of dependent Beta random variables (from Theorems~\ref{thm_domain_posterior} and~\ref{thm_LLM_posterior})
\item Apply the power transformation to obtain $R_i(n_F)$ or $R_L(n_F)$
\end{enumerate}

Both steps lack analytical solutions, so we compute CDFs via numerical integration or Monte Carlo sampling.

\section{Global numerical settings used in all figures}
\label{sec_numerical_setting}

All figures in this paper are produced under the same
imprecise hierarchical Bayes setup and numerical settings:

\begin{itemize}
\item \textbf{Domains/subdomains.}
      D1 (Coding) = \{MBPP, DS-1000\}, \quad
      D2 (Reasoning) = \{BoolQ, RACE-H\}.

\item \textbf{Aggregation weights.}
      Within-domain weights
      $\Omega_1=[0.204,\,0.796]$ (MBPP, DS-1000) and
      $\Omega_2=[0.483,\,0.517]$ (BoolQ, RACE-H);
      LLM-level weights $W=[0.149,\,0.851]$ (D1, D2).
\item \textbf{Hyperpriors and imprecision ranges (per domain $i$).}\\
      $\mu_i\sim\mathrm{Beta}(a_i,b_i)$,\;
      $\nu_i\sim\mathrm{Gamma}(c_i,\text{rate}=d_i)$,\;
      with
      $a_i,b_i\in[1,12]$ and $c_i,d_i\in[1,25]$.     
\item \textbf{Simulation}
      For each domain and each sampled hyperparameter configuration:
      \emph{(1)} sample $(\mu_i,\nu_i)$ from the discretized posterior;
      \emph{(2)} draw subdomain accuracies $\theta_{ij}$ from their Beta posteriors;
      \emph{(3)} aggregate to $p_i$ and $p_L$;
      \emph{(4)} compute either empirical CDFs of $Z\in\{\theta_{ij},p_i,p_L\}$ on a fixed grid,
      or Monte Carlo expectations such as
      $\widehat{\mathbb{E}}[R_L(n_F)]=\frac{1}{S}\sum_{s=1}^{S}\big(p_L^{(s)}\big)^{n_F}$.
      Repeating over $K$ hyperparameter configurations yields a family of CDFs/expectations;
      the plotted bands are the \emph{pointwise min--max} across configurations (epistemic uncertainty).

\item \textbf{CDF envelopes.}
      For any quantity $Z\in\{\theta_{ij},p_i,p_L\}$, we compute an empirical CDF
      for each configuration, and plot the
      pointwise min--max envelope $[F_{\min}(t),F_{\max}(t)]$ across configurations.
\end{itemize}

\section{List of Mathematical Notations}
\label{sec_notation_appendix}

\begin{table}[H]
\centering
\label{tab_notation_1}
\begin{tabular}{ll}
\toprule
\textbf{Symbol} & \textbf{Description} \\
\midrule
$\mathcal{X}$ & Input-space of all possible tasks for an LLM \\
$\pi$ & Operational Profile (OP), probability distribution over tasks \\
$n$ & Number of future tasks \\
$n_F$ & Number of consecutive future operations \\
$I(x_t)$ & Indicator function: 1 for success, 0 for failure on task $t$ \\
$R(n, \pi)$ & LLM reliability: probability of failure-free operation over $n$ tasks \\
$M$ & Number of LLM models being evaluated \\
$m$ & Number of independent domains \\
$D_i$ & Domain $i$ ($i = 1, \ldots, m$) \\
$n_i$ & Number of subdomains in domain $i$ \\
$S_{ij}$ & Subdomain $j$ in domain $i$ ($j = 1, \ldots, n_i$) \\
$C_{ij}$ & Number of correct responses in subdomain $S_{ij}$ \\
$N_{ij}$ & Number of trials (tasks) in subdomain $S_{ij}$ \\
$C_i$ & Set of observed data in domain $i$: $\{(C_{ik}, N_{ik})\}_{k=1}^{n_i}$ \\
$\text{data}$ & Observed data across all domains: $\{C_1, \ldots, C_m\}$ \\
$\theta_{ij}$ & Subdomain reliability (success probability) for $S_{ij}$ \\
$\boldsymbol{\theta}_i$ & Vector of subdomain reliabilities: $(\theta_{i1}, \ldots, \theta_{in_i})$ \\
$p_i$ & Domain-level reliability \\
$p_L$ & LLM-level reliability\\
$R_{ij}(n_F)$ & Subdomain reliability for $n_F$ consecutive operations \\
$R_i(n_F)$ & Domain reliability for $n_F$ consecutive operations \\
$R_L(n_F)$ & LLM reliability for $n_F$ consecutive operations \\
$\Omega_{ij}$ & Operational weight for subdomain $j$ in domain $i$ \\
$W_i$ & Operational weight for domain $i$ \\
$\mu_i$ & Expected reliability (prior mean) for domain $i$ \\
$\nu_i$ & Prior strength (concentration parameter) for domain $i$ \\
$a_i, b_i$ & Hyperparameters for Beta prior on $\mu_i$ \\
$c_i, d_i$ & Hyperparameters for Gamma prior on $\nu_i$ \\
$h_i$ & Tuple of hyperparameters: $(a_i, b_i, c_i, d_i)$ \\
$\mathcal{A}_i$ & Admissible set of hyperparameters for domain $i$ \\
$\mathcal{H}$ & Tuple of all domain hyperparameters: $(h_1, \ldots, h_m)$ \\
$\mathcal{A}_{\text{LLM}}$ & Admissible set at LLM level: $\mathcal{A}_1 \times \cdots \times \mathcal{A}_m$ \\
$F_{p_i}(t \mid C_i, h_i)$ & CDF of domain reliability $p_i$ \\
$\underline{F}_{p_i}(t \mid C_i)$ & Lower CDF envelope for domain reliability \\
$\overline{F}_{p_i}(t \mid C_i)$ & Upper CDF envelope for domain reliability \\
$F_{p_L}(t \mid \text{data}, \mathcal{H})$ & CDF of LLM reliability $p_L$ \\
$\underline{F}_{p_L}(t \mid \text{data})$ & Lower CDF envelope for LLM reliability \\
$\overline{F}_{p_L}(t \mid \text{data})$ & Upper CDF envelope for LLM reliability \\
$F_{R_{ij}(n_F)}(t \mid C_i, h_i)$ & CDF of subdomain $n_F$-operation reliability \\
$F_{R_i(n_F)}(t \mid C_i, h_i)$ & CDF of domain $n_F$-operation reliability \\
$F_{R_L(n_F)}(t \mid \text{data}, \mathcal{H})$ & CDF of LLM $n_F$-operation reliability \\
\end{tabular}
\end{table}

\begin{table}[H]
\centering
\label{tab_notation_2}
\begin{tabular}{ll}
\toprule
\textbf{Symbol} & \textbf{Description} \\
\midrule
$f_{\text{marg}}(\theta_{ij}, C_i; h_i)$ & Unnormalized marginal posterior for $\theta_{ij}$ \\
$Z_{\text{marg}}(h_i)$ & Normalizing constant for marginal posterior \\
$f_{p_i}(\cdot \mid \mu_i, \nu_i, C_i)$ & Conditional density of domain reliability \\
$G(t \mid \{\mu_i, \nu_i\}_{i=1}^m, \text{data})$ & Conditional CDF of $p_L$ given all hyperparameters \\
$\mathcal{R}_i(t)$ & Constraint region for domain $i$ \\
$\mathcal{R}_L(t)$ & Constraint region for LLM: $\{(p_1, \ldots, p_m)$ \\
$f_{\text{marg}}(\theta_{ij}, C_i; h_i)$ & Unnormalized marginal posterior for $\theta_{ij}$ \\
$Z_{\text{marg}}(h_i)$ & Normalizing constant for marginal posterior \\
$f_{p_i}(\cdot \mid \mu_i, \nu_i, C_i)$ & Conditional density of domain reliability \\
$G(t \mid \{\mu_i, \nu_i\}_{i=1}^m, \text{data})$ & Conditional CDF of $p_L$ given all hyperparameters \\
$\mathcal{R}_i(t)$ & Constraint region for domain $i$ \\
$\mathcal{R}_L(t)$ & Constraint region for LLM: $\{(p_1, \ldots, p_m)$ \\
$t$ & Probability threshold at which the corresponding CDF is evaluated \\
\end{tabular}
\end{table}

\end{document}